\documentclass[12pt]{article}
\usepackage{amsmath}
\usepackage{graphicx,psfrag,epsf}
\usepackage{enumerate}
\usepackage{natbib}
\usepackage{url} 
\usepackage{amsthm}
\usepackage{amssymb}
\usepackage{mathtools}
\usepackage{appendix}
\usepackage{here}
\usepackage[T1]{fontenc}
\usepackage{color}

\newcommand{\blind}{1}

\date{\vspace{-5ex}}

\DeclareMathOperator{\Var}{Var}

\newcommand{\e}{\mathrm{e}}
\newcommand{\de}{\mathrm{d}}
\def\P{{\mathbb P}}
\newcommand{\R}{{\mathbb R}}
\newcommand{\Rd}{{\mathbb R^d}}
\def\E{{\mathbb E}}
\newcommand{\M}{{\mathcal M}}
\newcommand{\RdR}{{\mathbb R ^d \times \mathbb R}}
\newcommand{\RdRM}{{(\mathbb R ^d \times \mathbb R) \times \mathcal M}}
\newcommand{\BB}{{\mathcal B}}
\newcommand{\1}{{\mathbf 1}}
\newcommand{\nn}{\nonumber}
\DeclareMathOperator{\Cov}{Cov}
\newcommand{\V}{{\mathcal V}}

\newtheorem{theorem}{Theorem}
\newtheorem{cor}{Corollary}
\newtheorem{definition}{Definition}

\newtheorem{lemma}{Lemma}
\newtheorem{remark}{Remark}
\newtheorem{example}{Example}

\newcommand{\bea}{\begin{eqnarray}}
\newcommand{\eea}{\end{eqnarray}}
\newcommand{\be}{\begin{equation}}
\newcommand{\ee}{\end{equation}}
\newcommand{\beann}{\begin{eqnarray*}}
\newcommand{\eeann}{\end{eqnarray*}}
\newcommand{\balnn}{\begin{align*}}
\newcommand{\ealnn}{\end{align*}}

\begin{document}
\def\spacingset#1{\renewcommand{\baselinestretch}%
	{#1}\small\normalsize} \spacingset{1}


\if1\blind
{
	\title{\bf The second-order analysis of marked spatio-temporal point processes, with an application to earthquake data}
	\author{ A. Iftimi\thanks{Corresponding author: iftimi@uv.es}$\ ^{1}$
		, O. Cronie$^2$ and F. Montes$^1$\\
		 $^1$Department of Statistics and Operations Research, 
		\\
		 University of Valencia, Spain,\\
		 $^2$Department of Mathematics and Mathematical Statistics, \\
		 Ume{\aa} University, Sweden}
	\maketitle
} \fi

\if0\blind
{
	\bigskip
	\bigskip
	\bigskip
	\begin{center}
		{\LARGE\bf The second-order analysis of marked spatio-temporal point processes, with an application to earthquake data}
	\end{center}
	\medskip
} \fi

\bigskip
\begin{abstract}
To analyse interaction in marked spatio-temporal point processes (MSTPPs), we introduce marked (cross) second-order reduced moment measures and $K$-functions for general inhomogeneous second-order intensity reweighted stationary MSTPPs. 
These summary statistics, which allow us to quantify dependence between different mark-categories of the points, are depending on the specific mark space and mark reference measure chosen. We also look closer at how the summary statistics reduce under assumptions such as the MSTPP being multivariate and/or stationary. A new test for independent marking is devised and unbiased minus-sampling estimators are derived for all statistics considered. In addition, we treat Voronoi intensity estimators for MSTPPs and indicate their unbiasedness. These new statistics are finally employed to analyse the well-known Andaman sea earthquake dataset. We find that clustering takes place between main and fore-/aftershocks at virtually all space and time scales. 
In addition, we find evidence that, conditionally on the space-time locations of the earthquakes, the magnitudes do not behave like an iid sequence.

\end{abstract}

\noindent%
\vfill

\newpage


\section{Introduction}


An earthquake is characterised by the shaking of the surface of the Earth and can range from being imperceptible to being devastating, with enormous damage and thousands of people killed. 
Historical data of earthquakes have shown that, on a year to year basis, there are some general patterns to be found. There are mainly three large areas of the earth with significant activity: i) the world's greatest earthquake belt, the circum-Pacific seismic belt, also known as \emph{the Ring of Fire}, ii) \emph{the Alpide}, which extends from Java to Sumatra through the Himalayas, to the Mediterranean, towards the Atlantic, and iii) the submerged \emph{mid-Atlantic Ridge} (see \citep{USGS}).

On the 26th of December 2004 a huge earthquake, the {\em Sumatra-Andaman event} hit the Andaman sea with a magnitude of 8.8. As expected, and as most are aware of, the consequences were terrible, resulting in both tremendous material damage as well as a massive number of human lives ended. As stated in \citep{vingy2005}, after the Sumatra-Andaman earthquake there were further small co-seismic jumps detected up to over 3,000 kilometres (km) from the earthquake epicentre, within 10 minutes from the earthquake. Also, \cite{vingy2005} state that post-seismic motion continued for a long period; 50 days after the earthquake in 2004 and the island of Phuket moved 34 cm. Hence, the high magnitude earthquakes tend to produce a sort of domino effect, with small aftershocks triggering each other. Following this event, on the 28th of March 2005, another earthquake of 8.4 magnitude hit Nias, an area close to the Sumatra-Andaman region. This process started slowly and spread in two directions, first toward the north for approximately 100 km and then, after 40 seconds of delay, towards the south for about 200 km \citep{walker2005}. Later, on the 12th of September 2007, two more earthquakes occurred in the Mentawai area, with magnitudes 8.5 and 8.1. According to \cite{konca2008}, the potential for a large event in this area remains high.

Earthquake records often come in the format where the $i$th event (shock), in addition to having a spatial location $x_i$ and an event time $t_i$, also carries further helpful information $m_i$, such as magnitude. 
In the language of point processes, $m_i$ is referred to as the {\em mark} of the $i$th event. 
When a mark $m_i$ is attached to a space-time point $(x_i,t_i)$ in this fashion, the random element/mechanism assumed to have generated the total collection of data is referred to as a {\em marked spatio-temporal point process (MSTPP)}, with the corresponding data referred to as a {\em marked spatio-temporal point pattern} \citep{daley03, diggle_book, Vere-Jones2009}. 
Other applications of MSTPPs include occurrences of disease incidents, crimes and fires. 

In this paper we are specifically interested in analysing the seismic activity in the Andaman sea region, during the years 2004-2008. 
Our aim is to develop point process tools which allow us to perform so-called second-order non-parametric analyses of marked spatio-temporal point patterns. In particular, we want to explicitly study the behaviour of earthquakes when we treat the magnitudes as marks. With such a setup we may then study in detail how shocks of different magnitudes interact with each other, in order to quantify how far in space-time one may find foreshocks/aftershocks of different sizes. Once such an analysis has been carried out, based on the outcome, one may then proceed to fitting appropriate models to the data. 

Classically, when analysing earthquakes within the framework of (marked) STPPs, the analysis has been based on {\em conditional intensity functions (CIs)} (see e.g.\ \citep{ChoiHall1999,daley03, marsan2008, Ogata98, schoenberg:brillinger:guttorp:02}). In principle, a conditional intensity function gives us the expected number of further events in a coming infinitesimal period, given the history of events up to that point. The beauty and appeal of CIs is that, when existing, they specify the whole distribution of the MSTPP. As pointed out by e.g.\ \cite{diggle_book}, however, not all MSTPP models have available/tractable CIs. Furthermore, much of the CI-based analysis is carried out within the framework of a given class of models. 

Recalling that we want to define a general fully non-parametric analysis, 
we will proceed with a non-CI based approach, thus following a random set/random measure formulation (see e.g. \citep{stoyan,daley03,diggle_book,MCbook,moller}). 
In this context, when analysing marked spatio-temporal point patterns, the first thing one starts with is to try to explain where and when events of a given mark category of the data tended to happen. Since where and when is a univariate property, in the sense that we are not dealing explicitly with possible dependencies between the points, we are dealing with analysing {\em intensity}. Before proceeding to proposing specific models for the the intensity structure, through the observed point pattern, one usually starts by obtaining a non-parametric estimate of the {\em intensity function} (see e.g.\ \citep{diggle_book}). The intensity function, in essence, reflects the infinitesimal probability of finding a point of the MSTPP at a given spatial location, at a given time, with a given mark. Note that it is different from the previously mentioned conditional intensity, which is defined as a conditional equivalence. 
In the simplest of worlds, we would simply assume homogeneity, i.e.\ that univariately it is equally likely to observe an event, with any mark, at any space-time position. This is, however, not the slightest realistic so we proceed by assuming inhomogeneity. 
Although the most natural candidate for this type of non-parametric estimation is kernel estimation \citep{Silverman, MCintensity, diggle_book}, due to the abrupt changes in activity of the earthquakes, both spatially and temporally, we here make the choice of to consider an adaptive approach, namely a {\em Voronoi intensity estimation} approach (see e.g.\ \citep{BarrSchoenberg}).

Having obtained a non-parametric estimate of the intensity function, so that we have a description of the univariate properties, we may proceed to studying the inherent dependence structure of the data-generating mechanism, i.e.\ the underlying MSTPP. 
We here focus on second-order summary statistics, thus ignoring e.g.\ the spatio-temporal $J$-function and its components \citep{CronieSTPP} and the marked $J$-functions and their components \citep{cronie_marked,VanLieshoutMPP}. 
In the context of unmarked spatio-temporal point processes, \cite{DiggleKstat} extended Ripley's $K$-function $K(r)$ \citep{ripley76,ripley77} to the stationary spatio-temporal context. Recall that this function, $K(r,t)$, gives us the expected number of further space-time points from an arbitrary space-time point of the process, given that the points in question have space and time separation $r\geq0$ and $t\geq0$, respectively. 
After the introduction of the spatial inhomogeneous $K$-function $K_{\rm inhom}(r)$ \citep{baddeley_Kfunction}, which is defined as an integral of the pair correlation function, \cite{GabrielDiggle} extended its definition to the spatio-temporal context, resulting in the function $K_{\rm inhom}(r,t)$. Note that under inhomogeneity, given only one realisation, we cannot e.g.\ visually distinguish between regions of high intensity and clustering/aggregation. For general marks in the purely spatial setting, \cite{VanLieshoutMPP} defined a marked version $K^{CD}(r)$ of Ripley's $K$-function; loosely speaking it gives us Ripley's $K$-function under the condition that we restrict the interaction to take place between points with marks belonging to some mark set (category) $C$ and points with marks in a mark set $D$. In addition, inspired by \citep{VanLieshoutMPP}, \cite{cronie_marked} introduced a marked version of the inhomogeneous $K$-function, $K_{\rm inhom}^{CD}(r)$, which reduces to the multivariate version introduced in \citep{moller}, when we assume that the marks are integer-valued (a multivariate/multi-type inhomogeneous point process). It reduces to the one in \citep{VanLieshoutMPP} when we assume stationarity. 
In the current study we aim at combining the ideas of \cite{GabrielDiggle} with those of \cite{cronie_marked} to define a $K$-function $K_{\rm inhom}^{CD}(r,t)$ for inhomogeneous MSTPPs, which reduces to a combination of $K(r,t)$ and $K^{CD}(r)$ when we assume stationarity. Loosely speaking, $K_{\rm inhom}^{CD}(r,t)$ describes the interaction, in a \citep{GabrielDiggle} sense, between points belonging to mark set $C$ and points belonging to mark set $D$, for an inhomogeneous MSTPP. 
Note that for all summary statistics above, one of the main foci has been to consider their non-parametric estimation. Here, as well we will allocate a significant part of this paper to the estimation. 
Having developed $K_{\rm inhom}^{CD}(r,t)$ and its estimation schemes, it turns out that we may also devise some statistical testing procedures which we will also look a bit closer at.

Once we have developed the statistical tools, we analyse the earthquake data with the aim of quantifying the interactions, so that we may asses the space-times propagations of the shocks. 

The paper is structured as follows. In Section \ref{SectionData} we present the earthquake dataset, which has largely motivated the development of this study. Section \ref{preliminaries} introduces marked spatio-temporal point processes, together with a summary on mark spaces, reference measures and intensity functions. Section \ref{preliminaries} also formally introduces the pair correlation function, Palm distributions and different marking structures. In Section \ref{SectionModels} we give examples of some MSTPP models which will be used for evaluation throughout the paper. In Section \ref{marked_K_definitions}, we introduce second-order intensity-reweighted stationarity for MSTPP and we define the marked spatio-temporal second-order reduced moment measure together with the marked spatio-temporal inhomogeneous $K$-function $K_{\rm inhom}^{CD}(r,t)$. In Section \ref{marked_K_definitions} we also provide some representation results. In Section \ref{inference} we propose estimators for intensity functions (a Voronoi tessellation based approach), as well as for the new second-order summary statistics. In addition, we consider ideas for testing independence assumptions of the marks. Section \ref{app_phuket} gives the second-order analysis of the earthquake dataset and in Section \ref{conlusions} we give some conclusions and a discussion on future work. The Appendix includes some technical details on spatio-temporal distances and the proofs of the results presented in the paper. Also, in the Appendix we look closer at how $K_{\rm inhom}^{CD}(r,t)$ reduces under assumptions of stationarity, multivariate marking and anisotropy.

\section{Data: Earthquakes}\label{SectionData}

Earthquakes are registered using a seismographic network and the most common measure is the magnitude, which is a measure of the size of the earthquake source; this number is considered location independent \citep{USGS}. Earthquakes of magnitude 3 or lower are almost undetectable and rarely felt. Earthquakes of magnitudes higher than 3 can cause landslides, which in turn can have fatal outcomes. Shocks of magnitude 7 and higher can cause severe landscape and building damage, and consequently human fatalities. When the epicentre of the earthquake is located offshore, there is also the possibility of tsunami development. Furthermore, often very large earthquakes are followed by a sequence of aftershocks, where the magnitude of the aftershocks can vary and some large aftershocks can have their own associated aftershock sequences \citep{ptprocess}.


In this paper we use earthquake data from the Sumatra region, registered from 2004 to 2008. The data in question can be downloaded freely from the R package \verb|PtProcess| \citep{ptprocess}. It was originally extracted from the {\em preliminary determination of epicentres catalogue}, provided by the {\em US Geology Survey} (\url{ftp://hazards.cr.usgs.gov/pde/}). More specifically, it includes earthquakes registered in the area of Sumatra, Indonesia (part of the Alpide), with magnitudes (rounded to one decimal) larger than or equal to 5. The spatial region considered has boundaries $89^\circ$ E, $105^\circ$ E, $16^\circ$ N and $5^\circ$ S. 
We transform the spatial coordinates from longitude/latitude to UTM scale \citep{Snyder}. 
Also, the time frame stretches from the midnight of the 1st of January 2004 until the 30th of December 2008, the day of the last registered shock. The first registered shock took place on the 16th of February 2004. A total of 1248 earthquakes were recorded during this period. Figure \ref{phuket_data} shows the spatial distribution of the point pattern of all 1248 earthquakes registered in the Sumatra area from the 16 February 2004 to 30 December 2008. The sizes of the black dots are proportional to the magnitudes of the events. The red X:s represent the four important earthquakes described previously. 
Furthermore, Figure \ref{phuket_data_year} shows all earthquakes annually as well as the temporal development of the magnitudes. 

\begin{figure}[!htbp]
\centering
  \includegraphics*[width=0.45\textwidth]{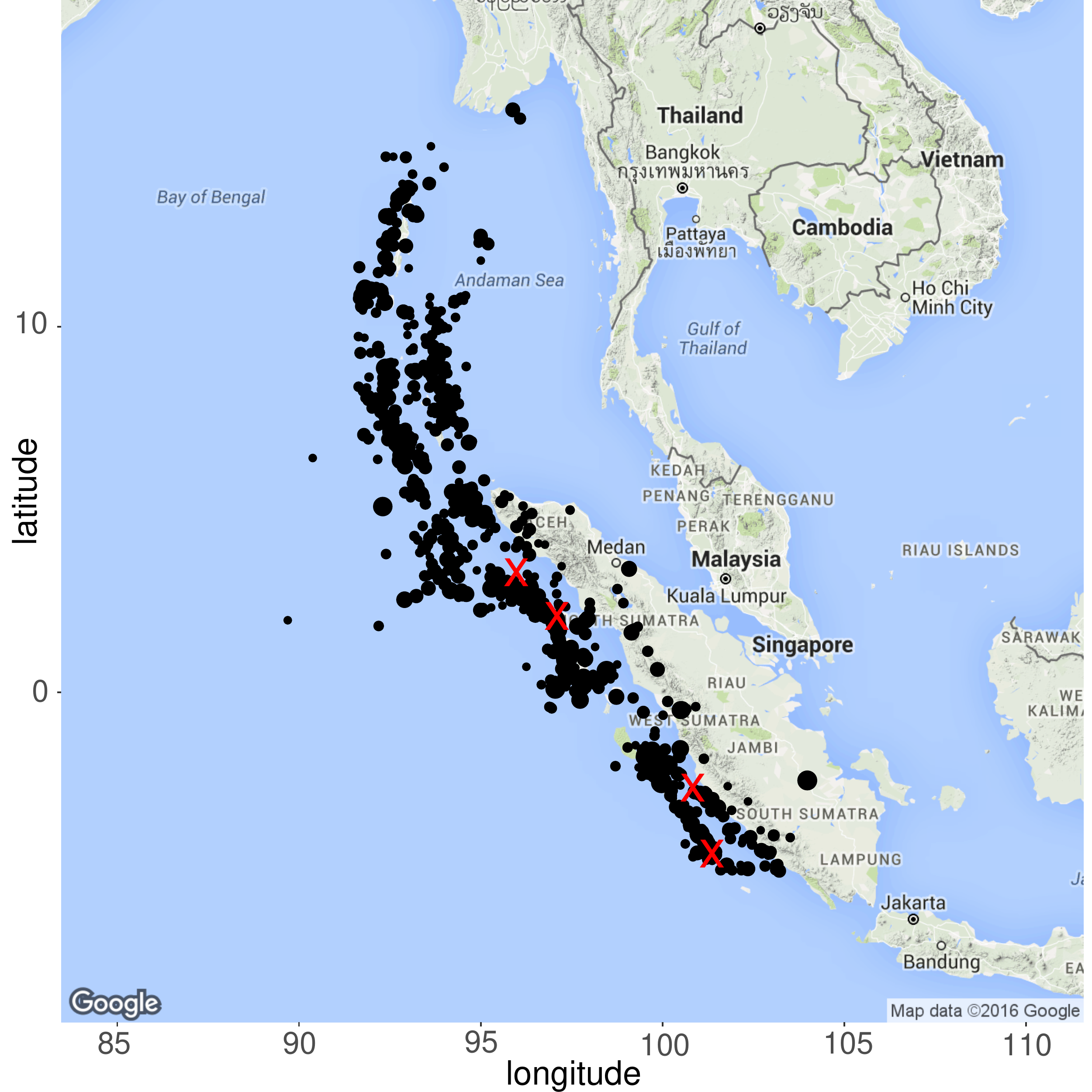}
        \caption{Spatial locations and magnitudes of the 1248 earthquakes registered in the Sumatra area. The sizes of the dots are proportional to the magnitudes. The red X:s correspond to the four important earthquakes described above.}
\label{phuket_data}
\end{figure}

\begin{figure}[!htbp]
\centering
  \includegraphics*[width=0.7\textwidth]{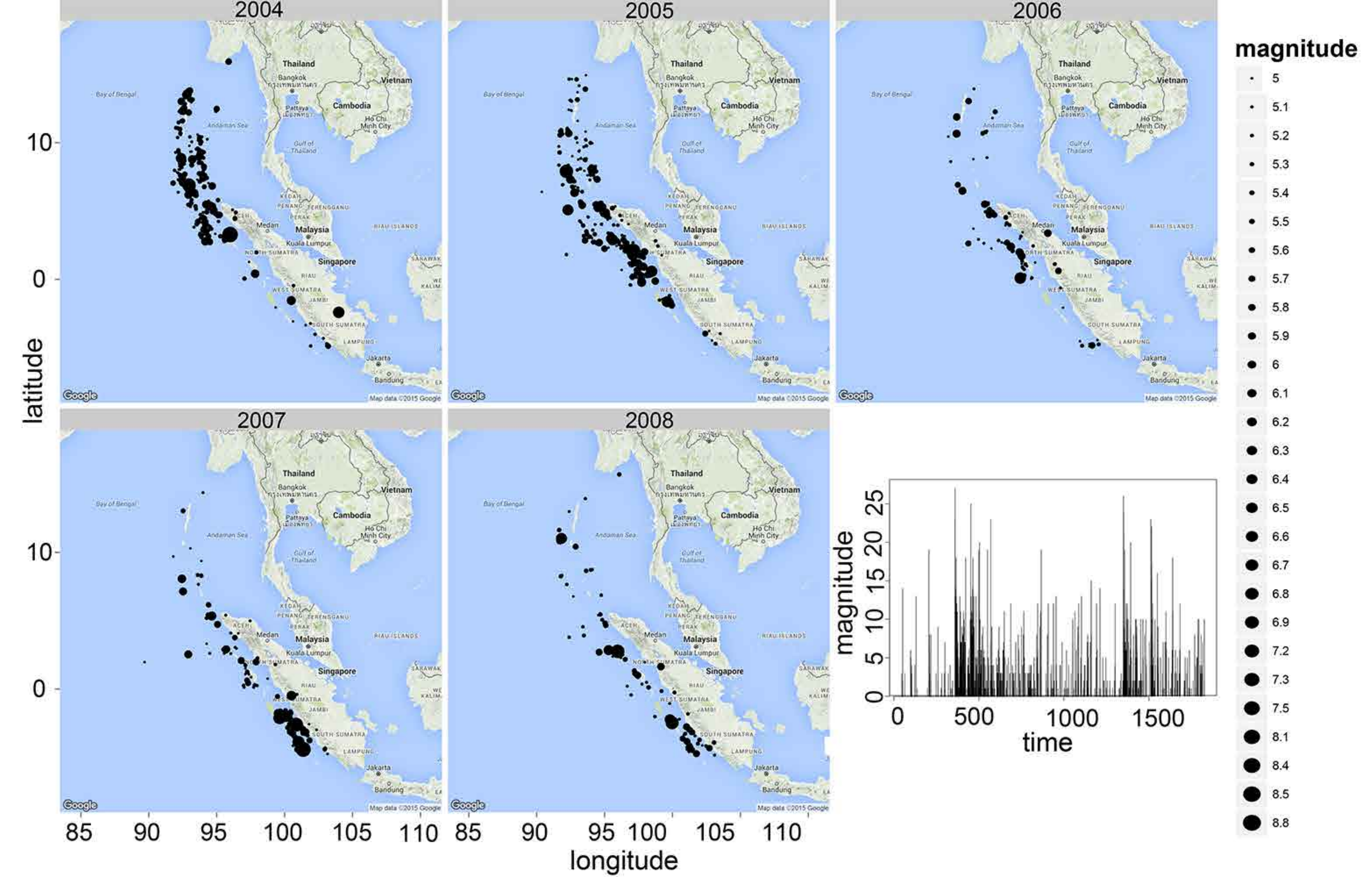}
        \caption{Spatial locations and magnitudes of the earthquake data, annually, from 2004 to 2008 (plots 1 to 5) and the time series of magnitudes of all earthquakes (plot 6).}
\label{phuket_data_year}
\end{figure}

Looking at Figure \ref{phuket_data} and Figure \ref{phuket_data_year}, we note that the earthquakes tend to appear in the same spatial region each year, the region being a reflection of the seismic belt. In other words, the spatial distribution of points in a given time period (Figure \ref{phuket_data_year}) is, essentially, a rescaling of an overall distribution (Figure \ref{phuket_data}). 
This observation will be exploited in the analysis (space-time separability). 

From the last plot of Figure \ref{phuket_data_year}, we further note that there is dependence between the event times and the magnitudes, which is to be expected since earthquakes give aftershocks. In other words, we will not explicitly assume (first order) independence between the temporal component and the mark component of the data.

\section{Marked spatio-temporal point processes}\label{preliminaries}

In order to formally define a {\em marked spatio-temporal point process} $Y$, with locations $x_i$ in $\Rd$, event times $t_i$ in $\R$ and marks $m_i$ in some suitable mark space $\M$, there are some technical details that need to be tended to.

\subsection{The underlying space}

Depending on what kind of mark types we want to consider for the data, in the construction of the related MSTPP model we have to choose an appropriate mark space and for the purpose of integration also appropriate associated reference measure. E.g., having recorded earthquakes we may either partition the magnitude scale, so that we consider a {\em multivariate STPP} (see the Appendix), or treat the marks as continuous. Depending on the choice, the statistical analyses differ so the choice made can be important.

Formally, regarding the mark space $\M$, we assume that it is a \emph{complete separable metric} (csm) space with corresponding metric $d'(\cdot,\cdot)$ and Borel sets $\BB(\M)$.  Recalling the space-time metric $d_{\infty}(\cdot,\cdot)$ from Appendix \ref{appendix_metric} on the (csm) space $\RdR$, which gives rise to the Borel sets $\mathcal{B}(\RdR)$, we equip $\RdRM$ with the the Borel sets $\BB(\RdRM)$, which become the product $\sigma$-algebra $\BB(\Rd)\otimes\BB(\R)\otimes\BB(\M)$. It is mostly natural to generate this structure through the metric 
\begin{align}
\label{Metric}
d((x_1,t_1,m_1),(x_2,t_2,m_2))
&=\max\{d_{\infty}((x_1,t_1),(x_2,t_2)),d'(m_1,m_2)\}
\\
&=\max\{\|x_1-x_2\|_{\Rd},|t_1-t_2|,d'(m_1,m_2)\},
\nonumber
\end{align}
where $(x_1,t_1,m_1),(x_2,t_2,m_2)\in\RdRM$.


\subsubsection{Mark spaces}

For a general discussion on mark spaces and their metric structures, see e.g.\ \citep[Appendix B.3]{moller}, \citep{stoyan} and \citep{MCbook}. 

When e.g.\ $\M\subseteq\R^l$, $l\geq1$, we let $d'(m_1,m_2)=\|m_1-m_2\|_{\R^l}$. 
In the case of our application, naturally we will consider $\M\subseteq\R$, i.e.\ $d'(m_1,m_2)=|m_1-m_2|$. 
For the case where $\M$ is a finite collection of labels $\{1,\dots,k\}$, $k\geq 2$, (let $d'(m_1,m_2)=|m_1-m_2|$), this is referred to as the multivariate/multi-type case and is covered separately in Appendix \ref{appendix_Multivariate_Stationary}.

\subsubsection{Reference measures and integration}

For the purpose of e.g.\ integration over $\RdRM$, we need to endow the underlying space $(\RdRM, \BB(\RdRM))$ with a reference measure. The choice of reference measure may seem as a mathematical detail and of little practical importance at first, but it will become clear that it plays a significant role also in the statistical analysis. A fact that is often overlooked in statistical settings. 

We will choose as reference measure the product reference measure 
$$
\ell\otimes\nu=\ell_{d+1}\otimes\nu=\ell_d\otimes\ell_1\otimes\nu, 
$$
where $\ell_d$ is the Lebesgue measure on $\Rd$, $d\geq1$, and $\nu$ is some suitable finite reference measure on the mark space. 
Throughout, $(\ell\otimes \nu)^n$ will represent the $n$-fold product measure of $\ell\otimes \nu$ with itself. 

When well-defined, we write 
\[
\int
f(x,t,m) [\ell\otimes\nu](d(x,t,m)) 
= 
\int\int\int f(x,t,m) \nu(dm) \de t \de x
\] 
for the integral of some $f:\RdRM\rightarrow\R$. When $\M=\R^k$, $k\geq1$, it is reasonable to choose $\nu(\cdot)$ as some suitable probability law and when $\M\subseteq\R^k$, $k\geq1$, is bounded, we may simply let $\nu(\cdot)=\ell_k(\cdot)$ (or normalise to have a uniform distribution as reference measure). 
For the case where $\M$ is a finite collection of labels $\{1,\dots,k\}$ (the multivariate case), see Appendix \ref{appendix_Multivariate_Stationary}. For other mark spaces, see e.g.\ \citep{stoyan}.


\subsection{Marked spatio-temporal point processes}

In analogy with \cite{CronieSTPP}, let the \textit{unmarked/ground process} \citep{daley03} $Y_g$ of space-time events $(x_i,t_i)$ be given by a spatio-temporal point process (STPP), as defined in Definition \ref{def:space_time_pp} in Appendix \ref{appendix_metric}. Informally, we assign marks $m_i\in\M$ (random variables) to the points of $Y_g$ to obtain the marked spatio-temporal point process $Y$. 

More formally, consider first the collection $N_{lf}$ of all simple non-negative integer valued measures $\varphi(\cdot)=\sum_{i=1}^n\delta_{(x_i,t_i,m_i)}(\cdot)=\sum_{i=1}^n\1\{(x_i,t_i,m_i)\in \cdot\}$, $0\leq n\leq\infty$ ($n=0$ corresponds to the null measure), on $\mathcal{B}(\RdRM)$ which are locally finite, i.e.\ $\varphi(B\times C)<\infty$ for bounded $B\times C\in \mathcal{B}(\RdRM)$, with the additional property that the (spatio-temporal) ground measure $\varphi_g(\cdot)=\varphi(\cdot\times\M)$ is locally finite on $\BB(\RdR)$. 
Note that the term simple refers to $\varphi(\{(x,t,m)\})\in\{0,1\}$ for any $(x,t,m)\in\RdRM$. 
The support of such a measure $\varphi(\cdot)\in N_{lf}$ will also be denoted by $\varphi$, hence, $\varphi=\{(x_i,t_i,m_i)\}_{i=1}^n\subseteq \RdRM$. 

Let $\mathcal{N}$ be the smallest $\sigma$-algebra on $N_{lf}$ to make the mappings $\varphi\mapsto\varphi(A)\in\{0,1,\ldots\}$ measurable for bounded $A\in \mathcal{B}(\RdRM)$. 
Denoting also the collection of all supports by $N_{lf}$, we note that there analogously exists a $\sigma$-algebra $\mathcal{N}$ that is generated by the mappings $\varphi\mapsto|\varphi\cap A|\in\{0,1,\ldots\}$, for bounded $A\in \mathcal{B}(\RdRM)$ and all supports $\varphi$. 

\begin{definition}\label{def1}
A marked spatio-temporal point process (MSTPP) $Y(\cdot)=\sum_{i=1}^N\delta_{(x_i,t_i,m_i)}(\cdot)$, $0\leq N\leq\infty$, is a measurable mapping from some probability space $(\Omega,\mathcal{F},\mathbb{P})$ into the measurable space $(N_{lf},\mathcal{N})$. 
If $N<\infty$ almost surely (a.s.)\ then $Y$ is called a {\em finite} MSTPP. 

\end{definition}

By the above arguments we may treat a MSTPP $Y$ as a random measure as well as a random subset $Y=\{(x_i,t_i,m_i)\}_{i=1}^N$ of $\RdRM$ and thus conveniently jump between the two notions. 
By this duality, $Y(A)$ and $|Y\cap A|$ may both be used to denote the cardinality of the number of points of $Y$ belonging to $A\in\mathcal{B}(\RdRM)$.
Note that by definition the ground process $Y_g=\{(x_i,t_i)\}_{i=1}^N$ is a well defined (spatio-temporal) point process on $\RdR$. We also write $P(\cdot)=\P(Y\in\cdot)$ for the distribution of $Y$, i.e.\ the probability measure that $Y$ induces on $(N_{lf}, \mathcal{N})$.  

If $\{(x_i,t_i)\}_{i=1}^N=Y_g\stackrel{d}{=}Y_g+(a,b)=\{(x_i+a,t_i+b)\}_{i=1}^N$ for any $(a,b)\in\RdR$, where $\stackrel{d}{=}$ denotes equality in distribution, we say that $Y$ is stationary  \citep{stoyan,daley03}. In practise stationarity is rarely realistic. 


\subsection{Intensity functions}\label{factorial_moment_measures}

Let $Y$ be a MSTPP with ground process $Y_g$. We will next consider the joint distributional properties of the points of $Y$, which we describe through the so-called product densities. 

For any $n\geq1$, assume that the $n$th \textit{factorial moment measure} $\alpha^{(n)}(\cdot)$ of $Y$ exists (as a locally finite measure on $\BB(\RdRM)^n$) and assume that $\alpha^{(n)}$ is absolute continuous with respect to $(\ell\otimes \nu)^n$. Then its permutation invariant Radon-Nikodym derivative $\rho^{(n)}(\cdot)\geq0$ \citep{stoyan,daley03,diggle_book}, the so-called $n$th \emph{intensity function/product density/factorial moment density}, may be defined through the so-called \emph{Campbell formula}: For any measurable function $f\geq 0$,  
\begin{align} 
\label{CampbellMPP}
&\E\left[
\mathop{\sum\nolimits\sp{\ne}}_{(x_1,t_1,m_1),\ldots,(x_n,t_n,m_n)\in Y}
f((x_1,t_1,m_1),\ldots,(x_n,t_n,m_n))
\right]=
\\
&=
\int \dots \int 
f((x_1,t_1,m_1),\ldots,(x_n,t_n,m_n))
\rho^{(n)}((x_1,t_1,m_1),\ldots,(x_n,t_n,m_n))\prod_{i=1}^{n}\nu(dm_i)\de x_i \de t_i,
\nn
\end{align}
which includes the case where both sides are infinite. 
Here $\sum^{\neq}$ denotes summation over $n$-tuples $((x_1,t_1,m_1),\ldots,(x_n,t_n,m_n))$ of distinct points. 
Regarding the interpretation of $\rho^{(n)}(\cdot)$, by the simpleness of $Y$, 
\begin{align*}
&\P(Y(d(x_1,t_1,m_1))=1,\ldots,Y(d(x_n,t_n,m_n))=1)
=\\
&=\rho^{(n)}((x_1,t_1,m_1),\ldots,(x_n,t_n,m_n))\prod_{i=1}^n
\nu(dm_i)\de x_i \de t_i 
.
\end{align*}
This is the infinitesimal probability of observing points of $Y_g$ in the space-time neighbourhoods $d(x_i,t_i)\subseteq\RdR$ of $(x_i,t_i)$, with associated marks $m_i\in dm_i\subseteq\M$, where $[\ell\otimes\nu](d(x_i,t_i,m_i))=\nu(dm_i)\ell(d(x_i,t_i))=\nu(dm_i)\de x_i \de t_i$, $i=1,\ldots,n$. 
%
%
Note that $\rho^{(n)}(\cdot)$ does not give us the joint density of {\em all} points of $Y$, unless we condition on the total number of points $Y(\RdRM)=N=n$ \citep[Lemma 5.4.III]{daley03}.

To make the statistical analysis more practically feasible, we sometimes make the additional pragmatic assumption that $Y_g$ may be treated as either of the point processes
\bea
\label{Projections}
Y_S&=&\{x:(x,t)\in Y_g\}\subseteq\Rd,
\\
Y_T&=&\{t:(x,t)\in Y_g\}\subseteq\R,
\nonumber
\eea
with marks in $\R$ and $\R^d$, respectively (c.f.\ \cite{ghorbani}). 
Note that e.g.\ in the former case this holds when we have $\E[Y_g(B\times\R)]
<\infty$ for any bounded $B\in\BB(\Rd)$, which in turn holds e.g.\ when $Y_g$ a.s.\ has no points outside $\R^d\times W_T$, for some bounded $W_T\in\BB(\R)$. The other case is analogous. 
Both are naturally permitted if $Y_g$ (and thus $Y$) is a finite point process, i.e.\ if $N<\infty$ a.s..
Hence, from a practical point of view it is a very mild assumption.

\begin{remark}
Such additional marking is facilitated by the imposed space-time metric $d_{\infty}(\cdot,\cdot)$ \citep[p.\ 8]{MCbook}; see the Appendix for details. 
\end{remark}

Since the ground process $Y_g$ is well-defined by definition  we may also define its $n$th factorial moment measure
$$
\alpha_g^{(n)}(B_1\times\cdots\times B_n) = \alpha^{(n)}((B_1\times\M),\ldots,(B_n\times\M)),
\quad 
B_1,\ldots,B_n\in\BB(\RdR),
$$ 
assuming local finiteness. 
The next result, which is standard and a slight modification of e.g.\ \citep[Section 4.1.2]{Heinrich2013}, shows that $\rho^{(n)}(\cdot)$ can be written as a product of the ground product density $\rho_g^{(n)}(\cdot)$ and a conditional density of the marks, given the spatio-temporal locations. Note that $\rho_g^{(n)}((x_1,t_1),\ldots,(x_n,t_n)) \prod_{i=1}^{n}\de x_i \de t_i$ gives us the probability of finding points of $Y_g$ in infinitesimal neighbourhoods of $(x_i,t_i)\in\RdR$, $i=1,\ldots,n$.  

\begin{lemma}\label{LemmaProdDens}
If $\alpha_g^{(n)}(\cdot)$ exists, then
\begin{align}
\label{eq:lambda_rho_ground}
\rho^{(n)}((x_1,t_1,m_1),\ldots,(x_n,t_n,m_n))
=
f_{(x_1,t_1),\ldots,(x_n,t_n)}^{\M}(m_1,\ldots,m_n) \rho_g^{(n)}((x_1,t_1),\ldots,(x_n,t_n))
\end{align}
almost everywhere (a.e.), 
where $\rho_g^{(n)}(\cdot)$ is the $n$th product density of $Y_g$ and $f_{(x_1,t_1),\ldots,(x_n,t_n)}^{\M}(\cdot)$ is the density of the conditional probability $M^{(x_1,t_1),\ldots,(x_n,t_n)}(C)$, $C\in\BB(\M^n)$, of the marks of $n$ points of $Y$, given that they have space-time locations $(x_1,t_1),\ldots,(x_n,t_n)\in\RdR$. 

When, in addition, $Y_S$ is well defined, 
$$
\rho_g^{(n)}((x_1,t_1),\ldots,(x_n,t_n))
=
f_{x_1,\ldots,x_n}^{T}(t_1,\ldots,t_n) \rho_S^{(n)}(x_1,\ldots,x_n), 
$$
and if $Y_T$ is well defined, 
$$
\rho_g^{(n)}((x_1,t_1),\ldots,(x_n,t_n))
=
f_{t_1,\ldots,t_n}^{S}(x_1,\ldots,x_n) \rho_T^{(n)}(t_1,\ldots,t_n)
,
$$
where $\rho_S^{(n)}(\cdot)$ and $\rho_T^{(n)}(\cdot)$ denote the respective $n$th product densities of $Y_S$ and $Y_T$. 

\end{lemma}

Turning to the explicit univariate properties of $Y$, 
setting $n=1$ we obtain the {\em intensity measure} $\Lambda(B\times C)=\alpha^{(1)}(B\times C)=\E[Y(B\times C)]=\int_{B\times C} \lambda(x,t,m) \nu(dm)\de x \de t$, $B\times C\in \BB(\RdRM)$, where, as indicated in \citep{Vere-Jones2009}, the {\em intensity function} is given by 
$$
\lambda(x,t,m)=\rho^{(1)}(x,t,m) = f_{(x,t)}^{\M}(m)\lambda_g(x,t)
.
$$ 
Here $\lambda_g(x,t)=\rho_g^{(1)}(x,t)$ is the intensity function of the ground process. 
Also, when $Y_S$ and $Y_T$ are well defined, 
$
\lambda_g(x,t) = f_{x}^{T}(t)\lambda_S(x)
$ 
and 
$
\lambda_g(x,t) = f_{t}^{S}(x)\lambda_T(t), 
$
respectively, 
where $\lambda_S(\cdot)$ and $\lambda_T(\cdot)$ are the respective intensity functions of $Y_S$ and $Y_T$. 
Heuristically, in order to obtain $\lambda_g(x,t)$, we rescale $\P(Y_S\cap dx\neq\emptyset)=\lambda_S(x)\de x$ by the conditional infinitesimal probability that this event, with spatial location in $Y_S\cap dx$, occurs at time $t$.

At times one makes the assumption that the intensity is constant as a function of space, time or both. This is referred to as homogeneity.
\begin{definition}\label{DefHom} 

If $\lambda_g(x,t)=\lambda_T(t)$ only depends on $t\in\R$, we say that $Y$ is {\em spatially homogeneous}, whereas if $\lambda_g(x,t)=\lambda_S(x)$ only depends on $x\in\R^d$, we say that $Y$ is {\em temporally homogeneous}. 

We say that $Y$ is {\em (spatio-temporally) homogeneous} if its ground process is homogeneous, i.e.\ if $\lambda_g(x,t)\equiv\lambda>0$ and 
$
\lambda(x,t,m)=f_{(x,t)}^{\M}(m)\lambda, 
$
and we call it {\em inhomogeneous} otherwise.

\end{definition}

Some things should be noted here. Firstly, stationarity implies homogeneity. Secondly, the functions $\lambda_S(\cdot)$ and $\lambda_T(\cdot)$ are non-unique since e.g.\ $\lambda_g(x,t)=\lambda_T(t)=c\frac{\lambda_T(t)}{c}=c\widetilde\lambda_T(t)$ for any $c>0$. Also, statistically, homogeneity is a strongly simplifying assumption and it is seldom realistic nor advised to assume that the data under consideration is generated by a homogeneous process (unless one is very confident that the application in mind behaves accordingly).


\subsubsection{Separability}

We next consider the notion of separability \citep{ghorbani, GabrielDiggle}, of which homogeneity is an example. 
\begin{definition}
If the ground intensity can be expressed as a (non-unique) product $\lambda_g(x,t) = \lambda_1(x)\lambda_2(t)$ of two non-negative measurable functions $\lambda_1(\cdot)$ and $\lambda_2(\cdot)$, we say that $Y$ is {\em separable}. 

When $\lambda_S(\cdot)$ and $\lambda_T(\cdot)$ exist we may e.g.\ set $\lambda_1(x)=\lambda_S(x)$ and $\lambda_2(t)=f^T(t)=\lambda_T(t)/\int_{\R}\lambda_T(s)\de s$, or $\lambda_1(x)=f^S(x)=\lambda_S(x)/\int_{\Rd}\lambda_1(y)\de y$ and $\lambda_2(t)=\lambda_T(t)$ (note that $f^S(x)$ and $f^T(t)$ are probability densities). 

\end{definition}

It should be noted that separability mainly is a practical assumption, imposed to simplify the analysis, and it is not always justified. It is mainly suitable when $Y_g$ has a repetitive behaviour in the sense that the intensity may be treated as a temporal/spatial rescaling of an overall temporal/spatial intensity, where the rescaling happens independently. 



\subsection{Pair correlation functions}

Having defined the product densities, we may proceed to defining a further central summary statistic for point processes, the \emph{pair correlation function (pcf)} \citep{stoyan,moller},
\beann
g((x_1,t_1,m_1),(x_2,t_2,m_2)) 
= 
\frac{\rho^{(2)}((x_1,t_1,m_1),(x_2,t_2,m_2))}{\lambda(x_1,t_1,m_1)\lambda(x_2,t_2,m_2)}
.
\eeann
By expression \eqref{eq:lambda_rho_ground}, the pcf satisfies
\bea
\label{eq:relation_g_gg}
g((x_1,t_1,m_1),(x_2,t_2,m_2)) 
&=& 
\frac{f_{(x_1,t_1),(x_2,t_2)}^{\M}(m_1,m_2)}
{f^{\M}_{(x_1,t_1)}(m_1) f_{(x_2,t_2)}^{\M}(m_2)}
\frac{\rho_g^{(2)}((x_1,t_1),(x_2,t_2))}
{\lambda_g(x_1,t_1) \lambda_g(x_2,t_2)}
\\ 
&=& 
\frac{f_{(x_1,t_1),(x_2,t_2)}^{\M}(m_1,m_2)}
{f_{(x_1,t_1)}^{\M}(m_1) f_{(x_2,t_2)}^{\M}(m_2)} 
g_g((x_1,t_1),(x_2,t_2))
,
\nonumber
\eea
where $g_g(\cdot)$ is the pcf of the ground process $Y_g$. 
Due to expression \eqref{ProdDensPoi} below, for a Poisson process on $\RdRM$ the pcf satisfies 
$
g(\cdot)=g_g(\cdot)\equiv1. 
$ 
Hence, for a MSTPP $Y$ with intensity $\lambda(\cdot)$ and $g((x_1,t_1,m_1),(x_2,t_2,m_2))>1$, there is clustering between points of $Y$ located around $(x_1,t_1)$ and $(x_2,t_2)$, with associated marks $m_1$ and $m_2$. 
Similarly, $g((x_1,t_1,m_1),(x_2,t_2,m_2))<1$ indicates inhibition. 
Non-parametric estimates of pcf:s are used extensively to analyse whether data exhibits interaction \citep{diggle_book}.

\subsection{Specific marking structures}
\label{SectionMarkingStructures}

Below follow some possible marking structures that may be imposed. 
We will consider these in more depth further on and, in particular, we will see how they influence summary statistics that we will derive. Hereby they also play a role in the statistical analysis.

\subsubsection{Common mark distribution}
Starting with the univariate properties, 
we next introduce the notion of a common (marginal) mark distribution.

\begin{definition}\label{DefCommonMark}
We say that a MSTPP $Y$ has a {\em common (marginal) mark distribution} $M(C)$, $C\in\BB(\M)$, if all marks have the same marginal distributions; $M^{(x,t)}(\cdot)\equiv M(\cdot)$ for any $(x,t)\in\RdR$ and $f_{(x,t)}^{\M}(\cdot)\equiv f^{\M}(\cdot)$, $(x,t)\in\RdR$, for a {\em common mark density}. 

If, in addition, $M(\cdot)$ and the reference measure $\nu(\cdot)$ coincide, so that $f_{(x,t)}^{\M}(\cdot)\equiv1$ and $\lambda(x,t,m)=\lambda_g(x,t)$, we say that the reference measure is given by {\em the} mark distribution \citep[p.\ 119]{stoyan}.
\end{definition}

It should be emphasised that $Y$ having a common mark distribution means that all marks $m_1,\ldots,m_N$ have the same marginal distribution $M(\cdot)$; they may, however, very well be mutually dependent. 
Note that $Y$ being homogeneous with a common mark distribution results in 
$
\lambda(x,t,m)=f^{\M}(m)\lambda, 
$
so that $\lambda(x,t,m)=\lambda$ if the reference measure is given by the mark distribution. 


\subsubsection{Independent marks and random labelling}

In order to provide a complete marking structure for $Y$ we have to define all joint distributions of the marks $m_i$, $i=1,\ldots,N$ (conditionally on the ground process). 
This includes e.g.\ such elaborate structures as geostatistical marking (see e.g.\ \citep{IllianEtAl}). 
However, one possible simplifying assumption is to let the marks be independent. 
Following e.g.\ \citep[Def.\ 6.4.III]{daley03}, we consider the following two definitions:

\begin{enumerate}
\item $Y$ has \emph{independent marks} if, given the ground process $Y_g$, the marks are mutually independent random variables such that the distribution of a mark depends only on the spatio-temporal location of the corresponding event. Here we have 
\(
f_{(x_1,t_1),\dots,(x_n,t_n)}^{\M}(m_1,\dots,m_n)=\prod_{i=1}^{n} f_{(x_i,t_i)}^{\M}(m_i) 
\)
for any $n\geq1$. 

\item If, in addition to independent marking, $Y$ has a common mark distribution, i.e.\ if the marks are independent and identically distributed, then 
we say that $Y$ has the \emph{random labelling property}. Here \(
f_{(x_1,t_1),\dots,(x_n,t_n)}^{\M}(m_1,\dots,m_n)=\prod_{i=1}^{n} f^{\M}(m_i)
\)
for any $n\geq1$, 
where we recall the common mark density $f^{\M}(\cdot)$. 

\end{enumerate}


\subsection{Palm distributions}
In order to consider conditioning on the event that $Y$ has a point somewhere in $\RdRM$ (this will needed for our summary statistics), we turn to Palm distributions  \citep{daley03,MCbook,stoyan}. 
The family of {\em reduced Palm distributions} of $Y$,
$
\{P^{!(x,t,m)}(\cdot): (x,t,m)\in\Rd\times\R\times\M\}, 
$ 
may formally be defined as the family of regular probabilities \citep{daley03} satisfying the \emph{reduced Campbell-Mecke formula} (see e.g.\ \citep{MCbook}): For any measurable function $f:(\Rd\times\R\times\M)\times N_{lf}\rightarrow[0,\infty)$,
\begin{align}\label{reducedCM}
&\E\left[\sum_{(x,t,m)\in Y} f((x,t,m), Y\setminus\{(x,t,m)\}) \right]
=
\\ 
&= \int_{\RdR\times\M}
\E^{!(x,t,m)} \left [f( (x,t,m), Y ) \right]
\lambda(x,t,m) \nu(dm)\de x\de t \nonumber
.
\end{align}
Note that $\E^{!(x,t,m)}[\cdot]$ is the expectation corresponding to the probability measure $P^{!(x,t,m)}(\cdot)=\P^{!(x,t,m)}(Y\in\cdot)$. 
Concerning its interpretation, the MSTPP with distribution $P^{!(x,t,m)}(\cdot)$ on $(N_{lf},\mathcal N)$, the {\em reduced Palm process} at $(x,t,m)$, may be interpreted as the conditional MSTPP $(Y|\{Y\cap\{(x,t,m)\}\neq\emptyset\})\setminus\{(x,t,m)\}$. 
Under stationarity, $P^{!(x,t,m)}(\cdot)$ is constant as a function of $(x,t,m)$, whereby one sets $P^{!(x,t,m)}(\cdot)\equiv P^{!(0,0,m)}(\cdot)$. 

\subsubsection{Reduced Palm distributions with respect to the mark sets}
It will sometimes be convenient to consider conditioning with respect to a whole mark set $C\in\BB(M)$, instead of just one specific mark value as in $P^{!(x,t,m)}(\cdot)$. 
To do so, following \cite{cronie_marked}, we may define $\nu$-averaged reduced Palm distributions. 
\begin{definition}\label{def:RedPalmMark}
The {\em $\nu$-averaged reduced Palm distribution (at $(x,t)\in\RdR$), with respect to $C\in\BB(\M)$}, is defined as 
\bea
\label{ReducedPalmMark}
P_C^{!(x,t)}(R)=\P_C^{!(x,t)}(Y\in R)
=\frac{1}{\nu(C)}\int_C P^{!(x,t,m)}(R)\nu(dm)
,
\qquad
R\in\mathcal{N}
.
\eea
Note that this is a probability measure since $0\leq P^{!(x,t,m)}(\cdot)\leq1$. Expectation under $P_C^{!(x,t)}(\cdot)$ is given by $\E_C^{!(x,t)}[\cdot]=\frac{1}{\nu(C)}\int_C \E^{!(x,t,m)}[\cdot]\nu(dm)$, by Fubini's theorem.
\end{definition}
In the case that the reference measure is given by the mark distribution (recall Definition \ref{DefCommonMark}), 
$$
\P_C^{!(x,t)}(Y\in\cdot)
=
\frac{\int_C P^{!(x,t,m)}(\cdot)M(dm)}{M(C)}
$$ 
may be interpreted as the conditional distribution 
\[
\P\left(\left.Y\setminus(\{(x,t)\}\times\M)\in\cdot \right| Y\cap(\{(x,t)\}\times C)\neq\emptyset\right)
.
\]
Under stationarity, where
$
\P_C^{!(0,0)}(Y\in\cdot)=\P_C^{!(x,t)}(Y+(x,t)\in\cdot)
$
for almost any $(x,t)\in\RdR$, 
we refer to $P_C^{!(x,t)}(\cdot)$ as the {\em reduced Palm distribution with respect to the mark set $C$} (see \citep{VanLieshoutMPP} and \cite[p.\ 135]{stoyan}).

\section{Examples of models}\label{SectionModels}

We next briefly recall and consider some properties of two particular models that we will consider in this paper.

Poisson processes are the benchmark models for absence of (spatio-temporal) interaction \citep{stoyan,diggle_book,MCbook}. 
For a Poisson process $Y$ on $\RdRM$, due to the independence of its points, the product densities and the pcf satisfy
\begin{align}
\label{ProdDensPoi}
\rho^{(n)}((x_1,t_1,m_1),\ldots,(x_n,t_n,m_n))
&=
\prod_{i=1}^n \lambda(x_i,t_i,m_i)
=
\prod_{i=1}^n f_{(x_i,t_i)}^{\M}(m_i)\lambda_g(x_i,t_i)
,
\quad 
n\geq1
,
\nn
\\
g((x_1,t_1,m_1),(x_2,t_2,m_2))
&=g_g((x_1,t_1),(x_2,t_2))\equiv1
.
\end{align} 
Hence, it may be regarded as independently marked (see Section \ref{SectionMarkingStructures}). 
We stress that this differs from a Poisson process $Y_g$ on $\RdR$ to which we assign marks according to families $\{f_{(x_1,t_1),\ldots,(x_n,t_n)}^{\M}(\cdot):(x_1,t_1),\ldots,(x_n,t_n)\in\RdR\}$, $n\geq1$, of densities on $\M^n$; its pcf is given by   
$
f_{(x_1,t_1),(x_2,t_2)}^{\M}(m_1,m_2) (f^{\M}_{(x_1,t_1)}(m_1) f_{(x_2,t_2)}^{\M}(m_2))^{-1}
. 
$
Indeed, the two concepts coincide when we have independent marking for the latter (see e.g.\ \citep[Theorem 7.5]{Haenggi}).

\begin{example}[Poisson process]\label{ExamplePoisson}

We consider a spatio-temporal (ground) Poisson process $Y_g=\{(x_i,y_i,t_i)\}_{i=1}^N$ on $W_S\times W_T=[0,1]^2\times[0,1]$, with intensity function 
$$
\lambda(x_1,y_1,t_1)=5 t_1 \text{e}^{5+0.5x_1}, 
\qquad (x_1,y_1,t_1)\in W_S\times W_T.
$$ 
Conditionally on the number of points, $N$, we further consider $N$ independent  Bernoulli distributed random variables $m_1,\ldots,m_N$, with parameter $p=0.4$, and assign these to $Y_g$, as marks. Hereby the mark space is $\M=\{0,1\}$ and $Y=\{(x_1,y_1,t_1,m_1): (x_1,y_1,t_1)\in Y_g\}\subseteq W_S\times W_T\times\M$ is the resulting MSTPP. The reference measure considered is the counting measure (see the Appendix for details on multivariate STPPs). 

Figure \ref{poisson_no_marks} shows a realisation of such a process, together with spatial projections for two different time intervals, $[0,0.5]$ (middle) and $(0.5,1]$ (right).

\begin{figure}[!htbp]
\centering
  \includegraphics*[width=0.3\textwidth]{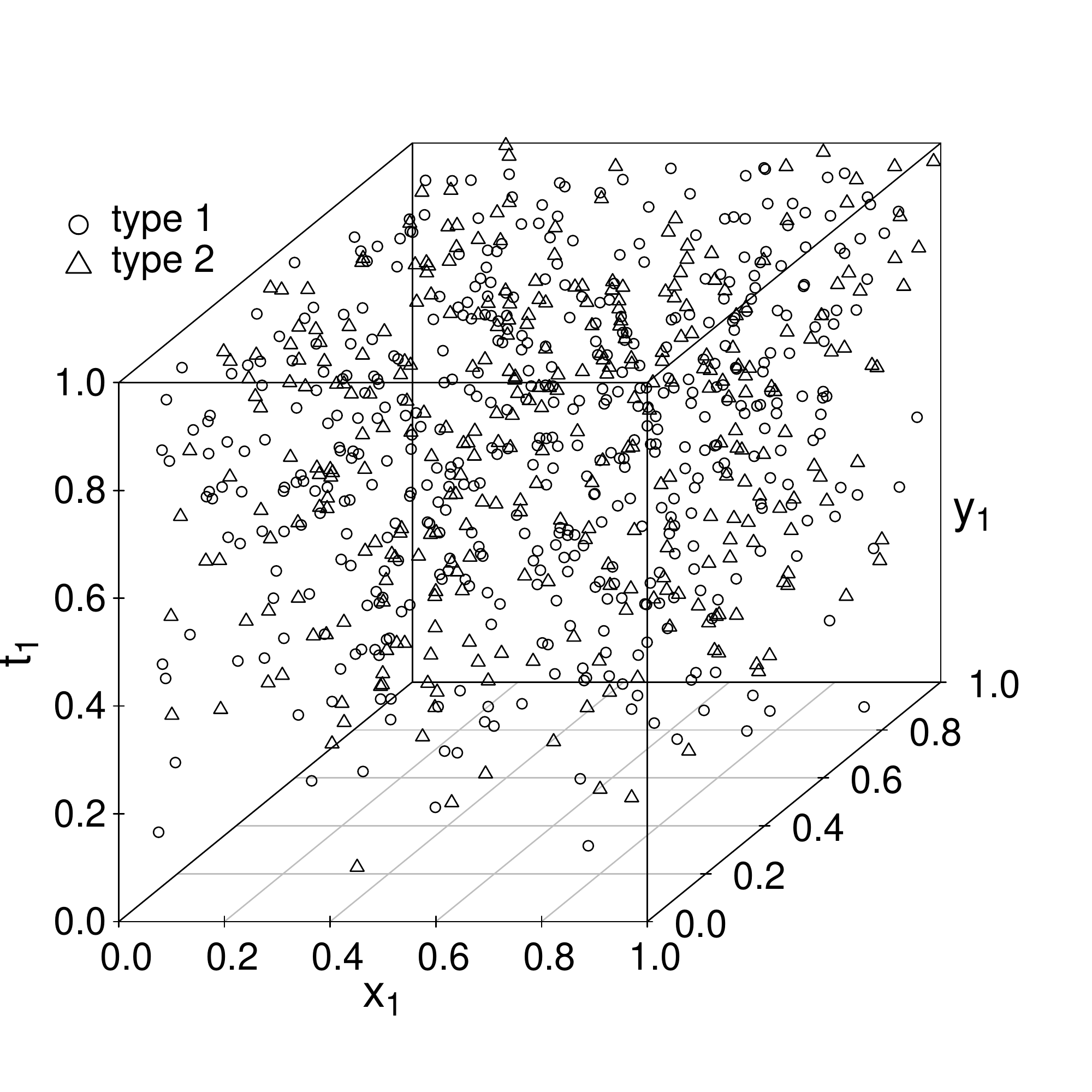}
  \includegraphics*[width=0.3\textwidth]{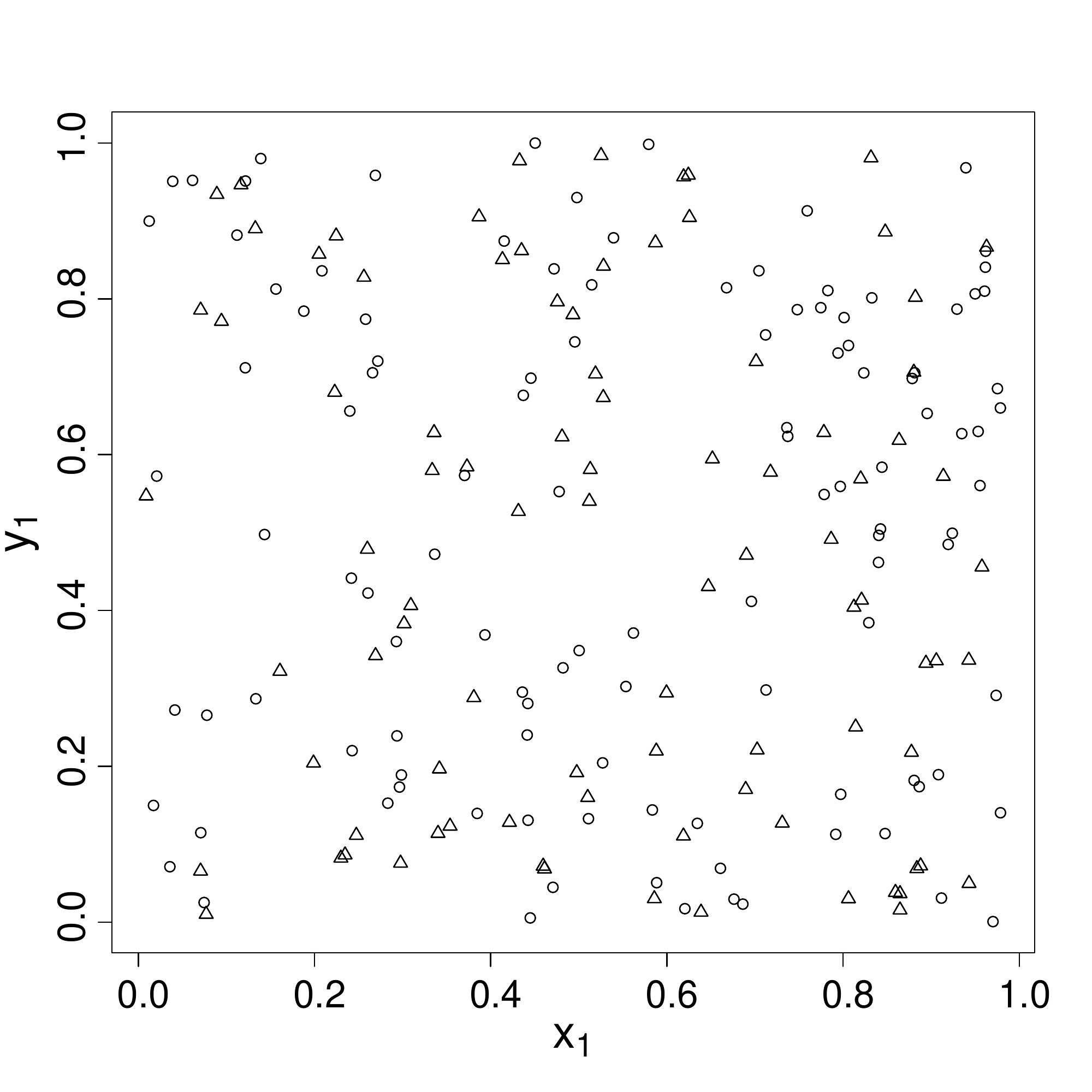}
  \includegraphics*[width=0.3\textwidth]{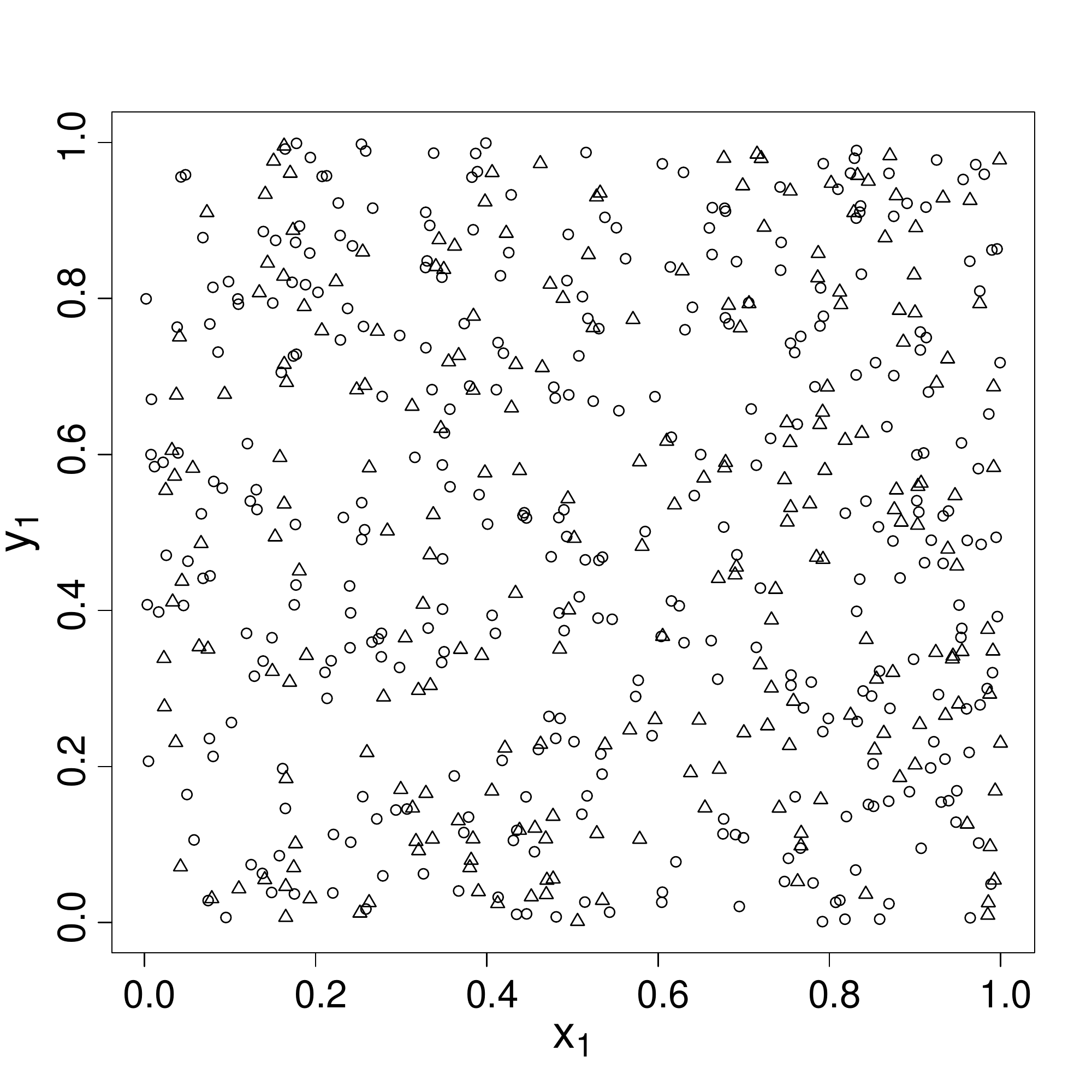}
        \caption{Spatio-temporal Poisson process with intensity function $\lambda(x_1,y_1,t_1)=5 t_1 \text{e}^{5+0.5x_1}$ on $[0,1]^2\times[0,1]$ with independent Bernoulli distributed marks (parameter $p=0.4$) (left); spatial projections for two different time intervals, $[0,0.5]$ (middle) and $(0.5,1]$ (right). 
Here $\M=\{0,1\}$ and "type 1" refers to a point having mark $0$.
}\label{poisson_no_marks}
\end{figure}  

\end{example}

Recall that a {\em spatio-temporal log-Gaussian Cox process (LGCP)} $Y_g$ \citep{moller_lgcp,CronieSTPP,diggle_book} is a spatio-temporal Poisson process for which the intensity function is given by the realisation of some (a.s.\ locally integrable non-negative) random field $X(x,t)=\e^{\mu(x,t)+Z(x,t)}$, where $Z(x,t)$ is a zero-mean Gaussian random field on $\RdR$. 
Such a random field $Z$ is characterised by its expectation function $\E[Z(x,t)]$ and its covariance function $\Cov(Z(x_1,t_1),Z(x_2,t_2))$. The simplest class of space-time covariance models are separable models, which are given by
$$
\Cov(Z(x_1,t_1),Z(x_2,t_2))=\Cov((x_1,t_1),(x_2,t_2))=\Cov_S(x_1,x_2)\Cov_T(t_1,t_2),
$$
where $\Cov_S$ is a covariance function on $\Rd$ and $\Cov_T$ is a covariance function on $\R$.
If, in addition, we assume stationarity in space and time, the covariance function depends only on the space-time lag between the points, whereby
\begin{align}\label{STcov}
\Cov(Z(x_1,t_1),Z(x_1+h,t_1+u))=C(h,u)=C_S(h)C_T(u),
\end{align}
where $(h,u)\in \RdR$ is the space-time lag between the points \citep{gneiting_cov}. 
To simulate random fields, we employ the R package \verb|RandomFields| \citep{randomfields_article}.

\begin{example}[Independently marked LGCP]\label{ExampleIndependentMarkedLGCP}

We consider a univariate spatio-temporal LGCP, $Y_g=\{(x_i,y_i,t_i)\}_{i=1}^N$, on the spatio-temporal domain $W_S\times W_T=[0,1]^2\times[0,1]$, with mean function given by $\mu(x_1,y_1,t_1)=\log(750)-0.5(y_1+t_1)-\sigma^2/2$, where $\sigma^2=(1/4)^2=1/16$. We further consider a separable space-time covariance function for $Z$, where the spatial covariance function is given by the stationary and isotropic {\it Whittle-Mat\'{e}rn} covariance model:
$$
C_S(h)=\sigma^2\dfrac{2^{1-\nu}}{\Gamma(\nu)}(c\|h\|)^\nu K_{\nu}(c\|h\|),
$$
where $\nu>0$ is a smoothness parameter, $c$ is a nonnegative scaling parameter and $K_\nu$ denotes the modified Bessel function of the second kind of order $\nu$. The temporal covariance function is constant and given by $C_T(u)=1$.

As in Example \ref{ExamplePoisson}, we consider $N$ independent random Bernoulli distributed marks, with parameter $p=0.4$, and obtain the MSTPP $Y=\{(x_1,y_1,t_1,m_1): (x_1,y_1,t_1)\in Y_g\}\subseteq W_S\times W_T\times\M$, where we again consider the counting measure as mark reference measure $\nu(\cdot)$. Appendix \ref{appendix_Multivariate_Stationary} provides further details on multivariate STPPs. 

\begin{figure}[!htbp]
\centering
  \includegraphics*[width=0.3\textwidth]{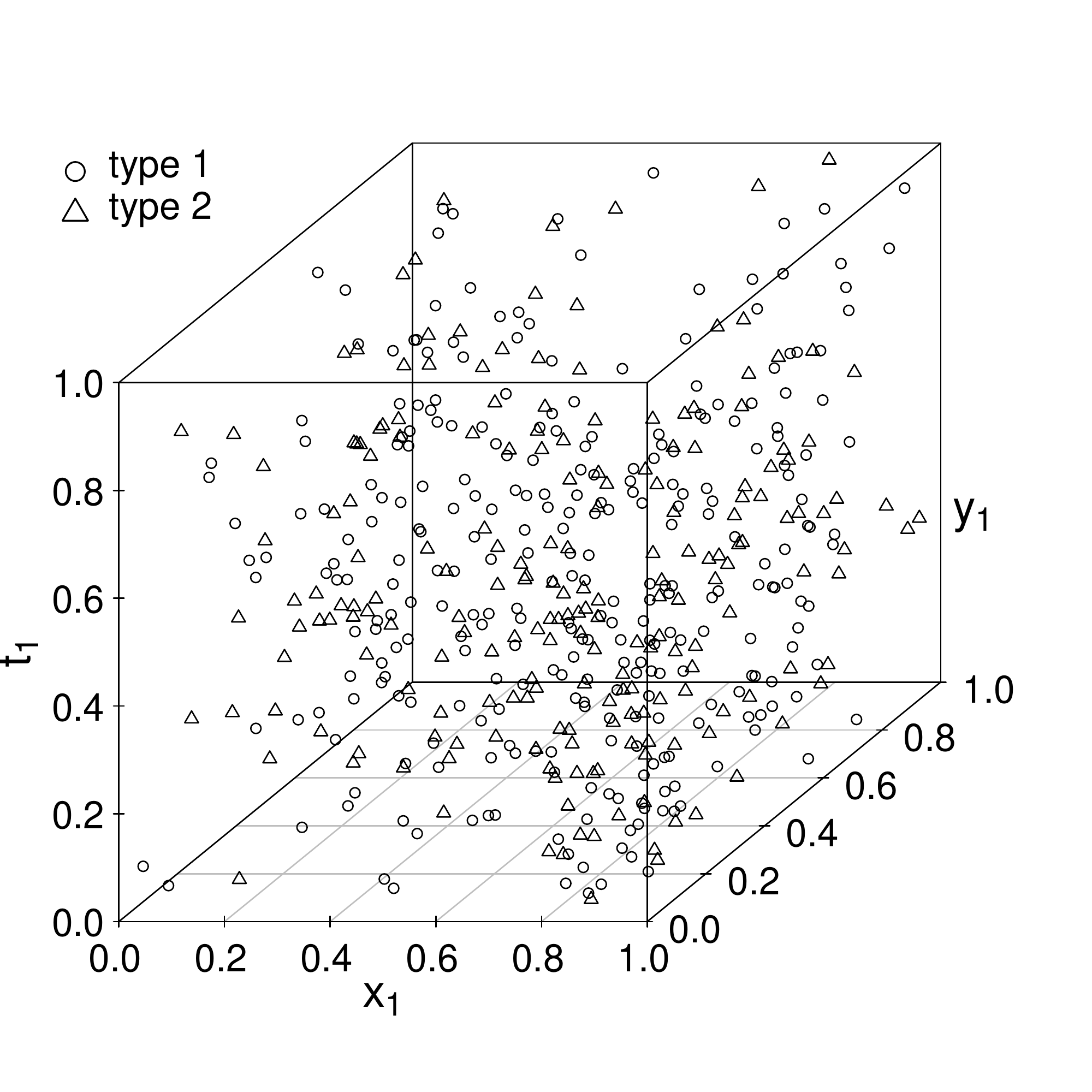}\label{lgcp_b}
  \includegraphics*[width=0.3\textwidth]{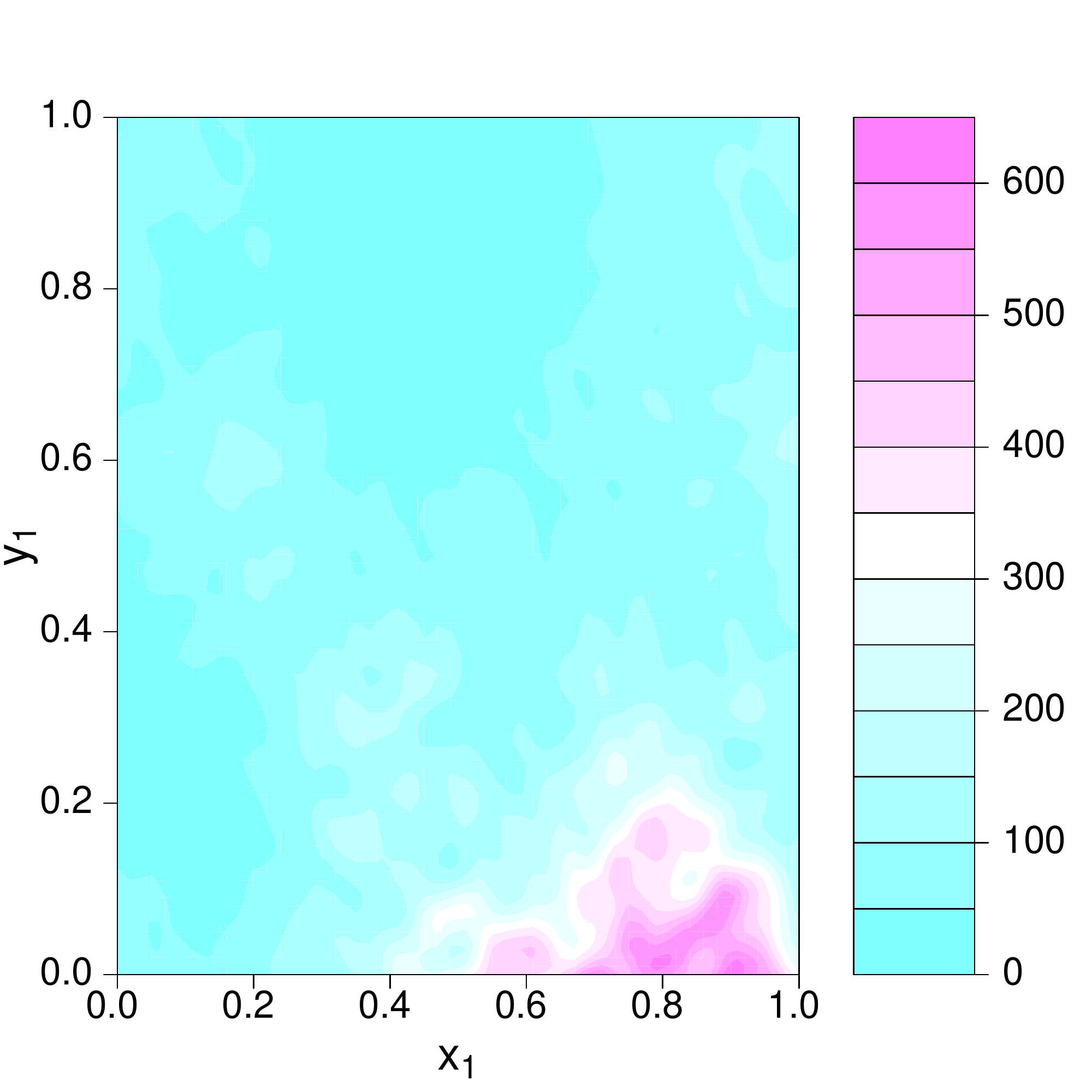}\label{proj_t_lgcp_b.pdf}
  \includegraphics*[width=0.3\textwidth]{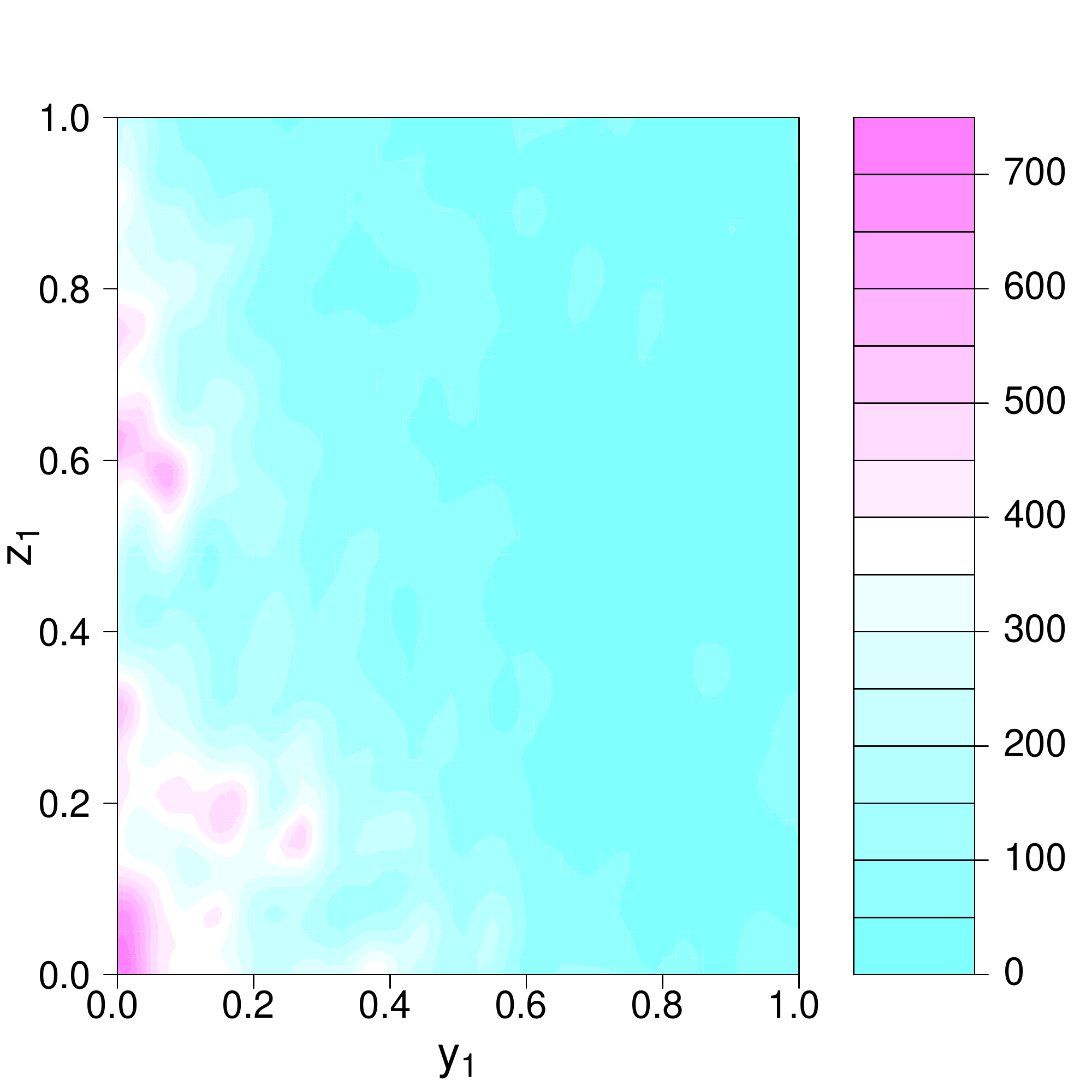}\label{proj_s_lgcp_b.pdf}
        \caption{Randomly labelled spatio-temporal log-Gaussian Cox process with iid $Bernoulli(0.4)$-distributed marks (left); projections of the random intensity field of the log-Gaussian Cox process $Y_g$, at time $t_1=0.5$ (middle) and at spatial coordinate $x_1=0.5$ (right). 
Here $\M=\{0,1\}$ and "type 1" refers to a point having mark $0$.
}\label{lgcp_bernoulli}
\end{figure}  

Figure \ref{lgcp_bernoulli} shows a realisation of such an independently marked spatio-temporal log-Gaussian Cox process (left), together with a temporal projection ($t_1=0.5$) and spatial projection ($x_1=0.5$) of the Gaussian random field (middle and right).
\end{example}

\begin{example}[Bivariate spatio-temporal process]\label{ExampleIndependentComp}

We consider a spatio-temporal Poisson process, $Y_1$, with the same intensity function as in Example \ref{ExamplePoisson}, on the spatio-temporal domain $W_S\times W_T=[0,1]^2\times[0,1]$. In the same spatio-temporal observation window we consider a spatio-temporal log-Gaussian Cox process, $Y_2$, with mean function given by $\mu(x_1,y_1,t_1)=\log(750)-1.5(y_1+t_1)-\sigma^2/2$, where $\sigma^2=(1/4)^2=1/16$. We consider the same spatio-temporal covariance function as in Example \ref{ExampleIndependentMarkedLGCP}.

We assign the numerical mark $1$ to all points coming from $Y_1$ and the numerical mark $2$ to the second component process, $Y_2$. Hence, the mark space is $\M=\{1,2\}$ and the bivariate STPP $Y$ is obtained by combining $Y_1$ and $Y_2$ into $Y=Y_1\cup Y_2$. Note that this is a multivariate STPP (see Appendix \ref{appendix_Multivariate_Stationary} for details) and as usual it is natural to let $\nu(\cdot)$ be given by the counting measure. 

Figure \ref{poisson_lgcp_indep} shows a realisation of such a bivariate STPP. 
Figure \ref{poisson_lgcp_indep} also shows projections of a realisation of the random intensity field of $Y_1$, at time $t_1=0.5$ (middle) and at spatial coordinate $x_1=0.5$ (right). 

\begin{figure}[!htbp]
\centering
  \includegraphics*[width=0.3\textwidth]{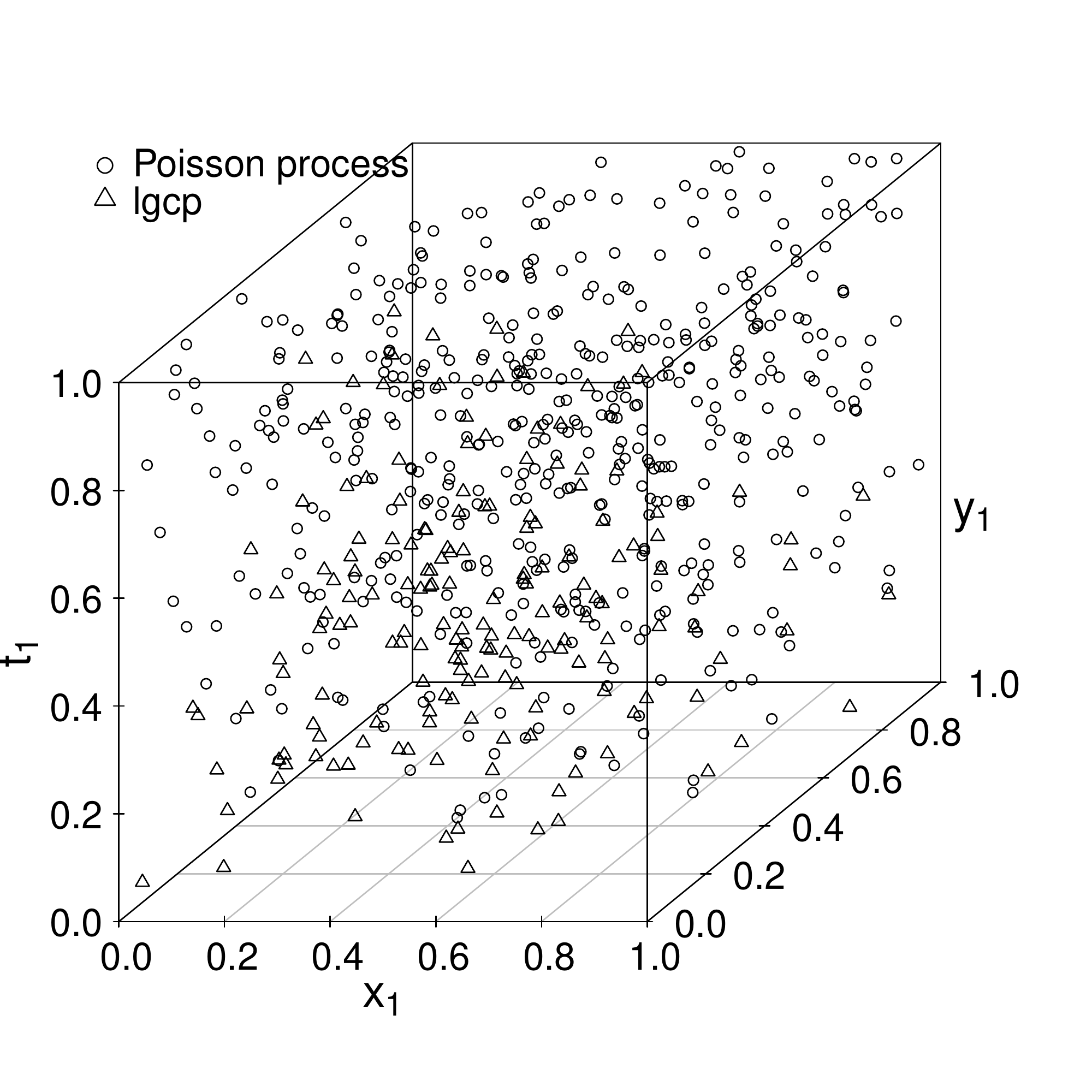}
  \includegraphics*[width=0.3\textwidth]{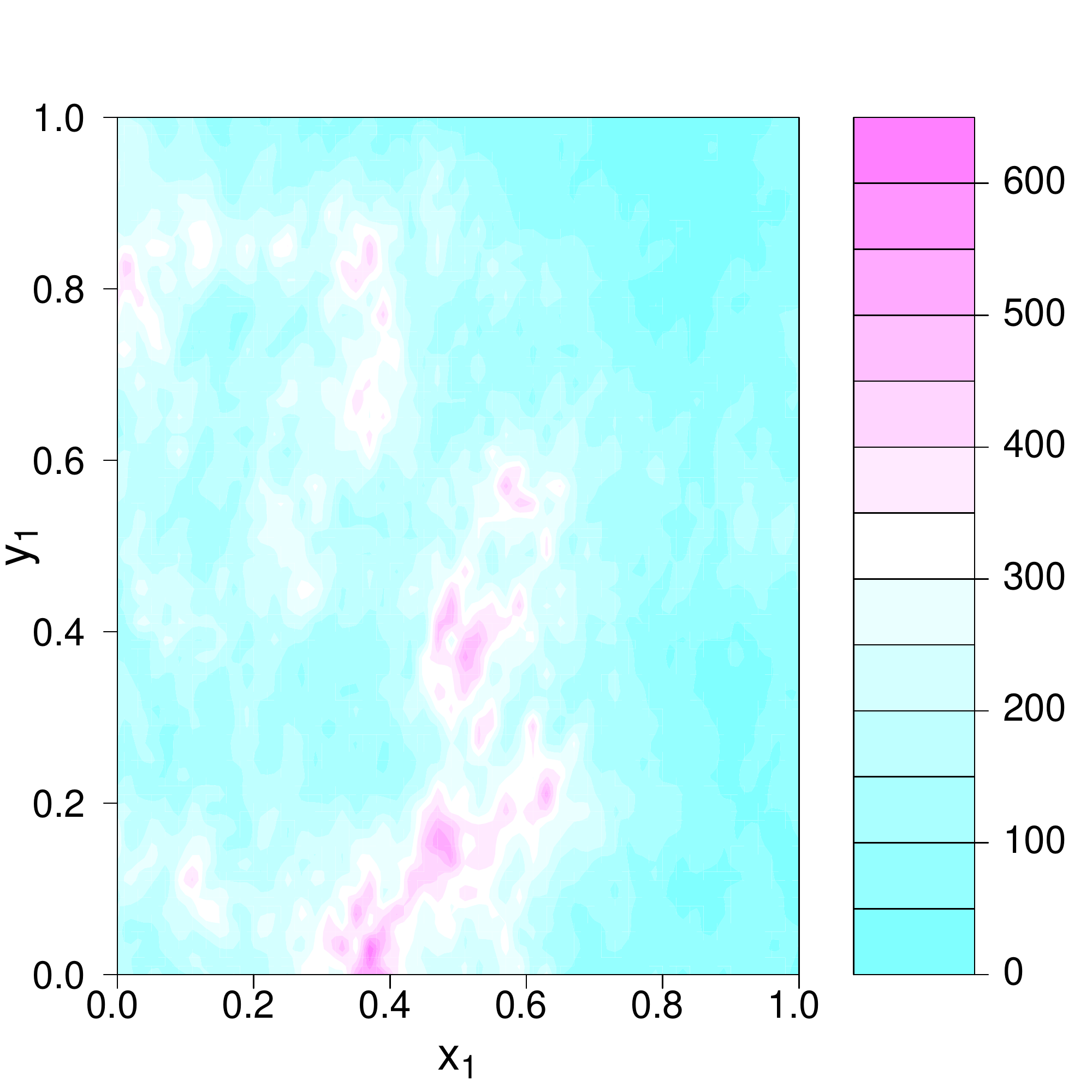}
  \includegraphics*[width=0.3\textwidth]{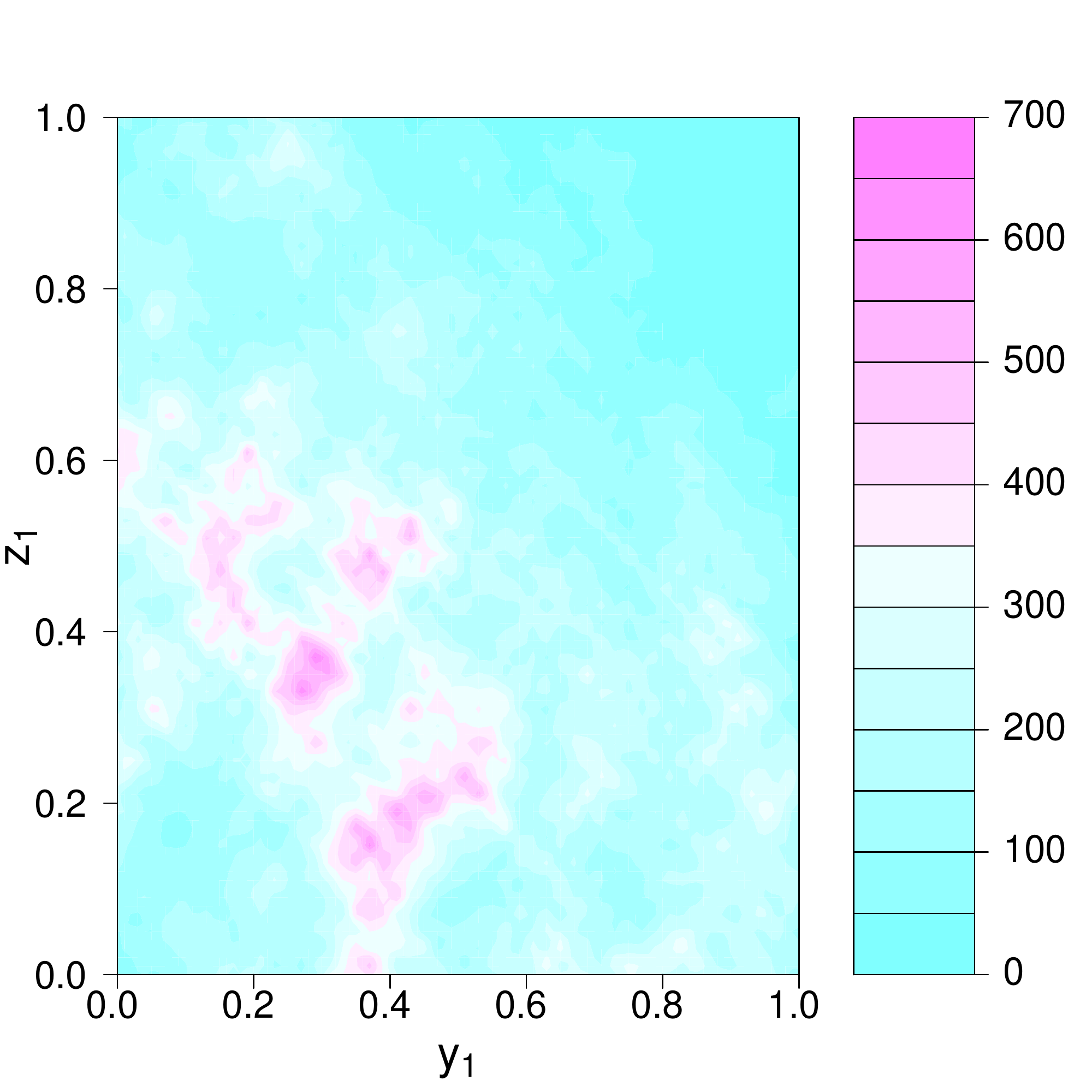}
        \caption{A realisation of a bivariate spatio-temporal process (left) together with projections of the random intensity field of $Y_2$, at time $t_1=0.5$ (middle) and at spatial coordinate $x_1=0.5$ (right).}\label{poisson_lgcp_indep}
\end{figure}  
\end{example}

\begin{example}[Geostatistically marked LGCP]\label{ExampleGeostatisticallyMarkedLGCP}

We consider a spatio-temporal log-Gaussian Cox process, $Y_g$, on the spatio-temporal domain $W_S\times W_T = [0,1]^2\times[0,1]$, with underlying mean function $\mu(x_1,y_1,t_1)=\log(750)-0.5(y_1+t_1)+\sigma^2/2$, where $\sigma^2=(1/4)^2=1/16$. As covariance function, we consider the separable spatio-temporal covariance function described in Example \ref{ExampleIndependentMarkedLGCP}. 
We then simulate a spatio-temporal Gaussian random field, $\{R(x,y,t):(x,y,t)\in [0,1]^2\times[0,1]\}$, with covariance function given by the stationary isotropic exponential model, $C(h,u)=C_S(h)C_T(u)=exp(-h)$; here $h\geq 0$ is the spatial Euclidean distance between two points (a separable model). 
In order to assign marks to $Y_g$, we let $m_i=R(x_i,t_i)$ for all $(x_i,t_i)\in Y_g$, whereby $\M=\R$ and the mark reference measure $\nu(\cdot)$ is assumed to be the Lebesgue measure on $\R$. 

Figure \ref{geost_lgcp} (left) shows a realisation of such a geostatistically marked spatio-temporal log-Gaussian Cox process, where the size of a circle around a point is proportional to the value of its continuous mark, together with the Gaussian random field of the marks (right).

\begin{figure}[!htbp]
\centering
  \includegraphics*[width=0.4\textwidth]{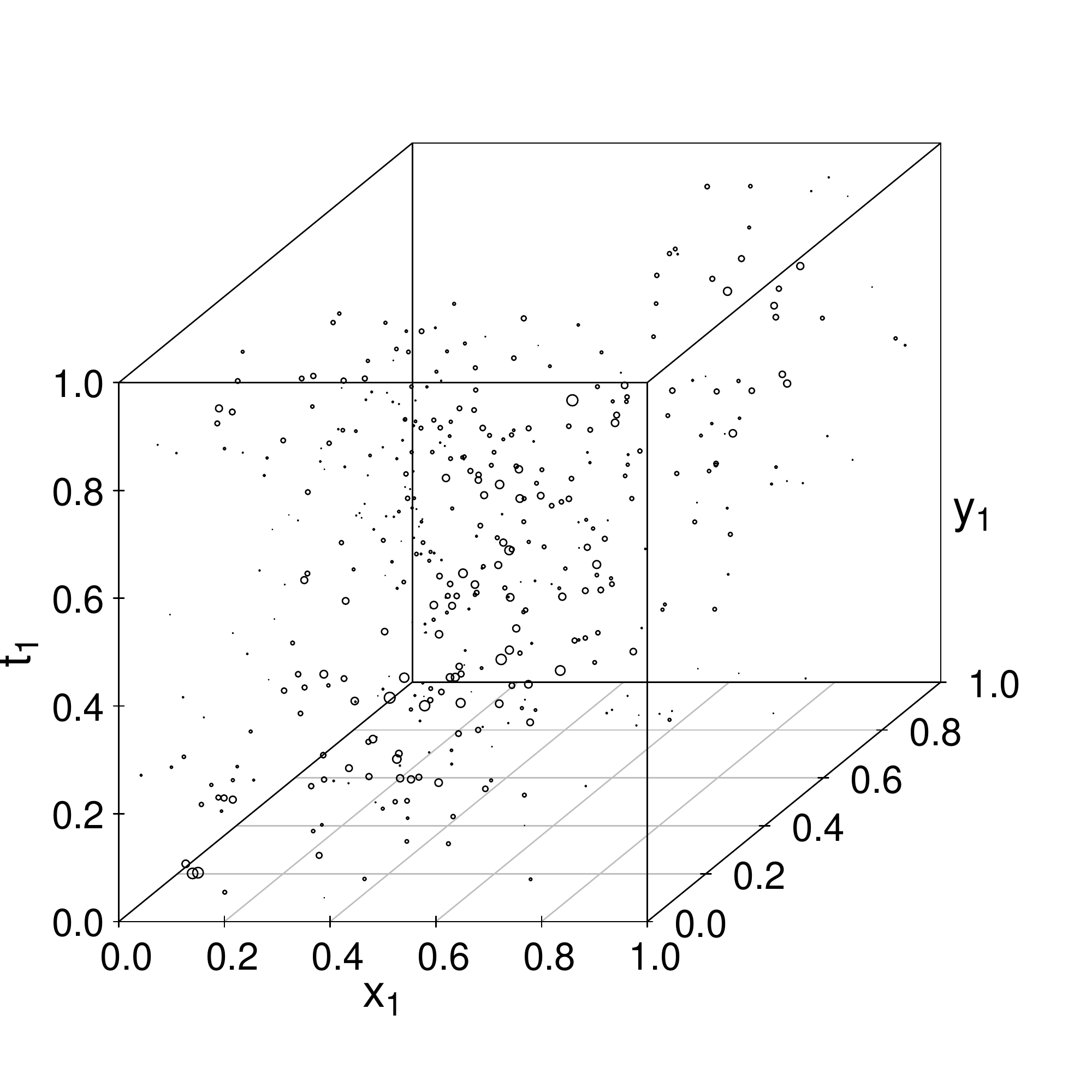}\label{geost_lgcp}
  \includegraphics*[width=0.4\textwidth]{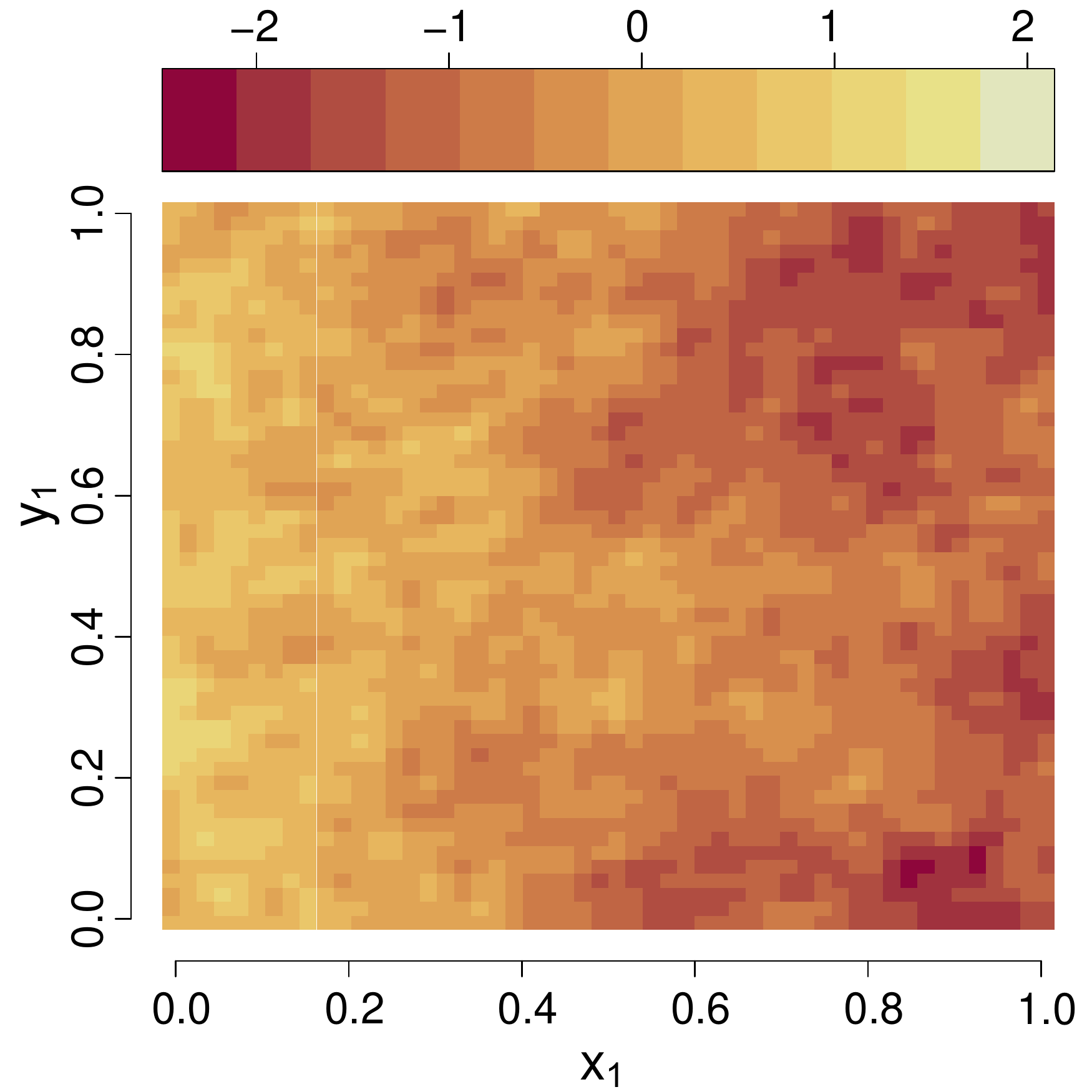}\label{grf_marks}
        \caption{A realisation of the above defined geostatistically marked spatio-temporal log-Gaussian Cox process (left); the size of a point is proportional to the value of its continuous mark. The Gaussian random field generating the marks (right).}
\end{figure}  
\end{example}

\section{Marked inhomogeneous second-order measures of spatio-temporal interaction}\label{marked_K_definitions}

As seen above, the intensity function of a MSTPP governs its univariate properties and the pcf governs second-order interactions. We now proceed by defining cumulative summary statistics/measures of spatio-temporal interaction for MSTPPs. The key idea is the extension of the marked inhomogeneous $K$-function of \cite{cronie_marked} to the spatio-temporal context.

\subsection{Second order intensity-reweighted stationarity}

A weaker form of stationarity that we impose when we consider the inhomogeneous MSTPPs below is \emph{second-order intensity-reweighted stationarity (SOIRS)} (see e.g.\ \citep{baddeley_Kfunction}).

\begin{definition}\label{DefinitionSOIRS}
We say that a MSTPP $Y$ is \emph{second-order intensity-reweighted stationary (SOIRS)} if the pcf exists and satisfies 
\begin{align*}
g((x_1+a,t_1+b,m_1),(x_2+a,t_2+b,m_2)) 
= g((x_1,t_1,m_1),(x_2,t_2,m_2))
\end{align*}
a.e., for any $(a,b)\in \RdR$. 
\end{definition}
Avoiding the degenerate case where $\lambda(x,t,m)=\lambda_g(x,t)\equiv0$ a.e., we must require that $\lambda_g(x,t)>0$ a.e. \citep{baddeley_Kfunction}. Furthermore, a homogeneous SOIRS point process is second-order stationary. 
If in addition to SOIRS we have that 
\bea
\label{SOIRSI}
g((x_1,t_1,m_1),(x_2,t_2,m_2))=\overline{g}(\|x_1-x_2\|_{\Rd},|t_1-t_2|,m_1,m_2)
,
\eea
i.e.\ the pcf is given by some function $\overline{g}(\cdot)$ that spatio-temporally depends only on the spatial distances and the temporal distances, we refer to $Y$ as {\em SOIRS with isotropy (SOIRSI)} \citep{GabrielDiggle,diggle_book}. 
C.f.\ the isotropy part of \cite[p.34]{moller}.

\subsection{Marked spatio-temporal second-order reduced moment measures}

As an alternative to the pcf as a marked measure of second-order spatio-temporal interaction, we may instead consider cumulative versions of it. 
Throughout we will assume that $Y$ is SOIRS (see Definition \ref{DefinitionSOIRS}). 

We start by defining the {\em marked spatio-temporal second-order reduced moment measure} (c.f.\ \cite[Definition 4.5]{moller}), which is our main building block. It describes how points of $Y$, with marks in some Borel set $C\subseteq\M$, interact with points of $Y$ with marks in some Borel set $D\subseteq\M$, when their spatio-temporal separation vectors lie in some $E\subseteq\RdR$. 

\begin{definition}\label{def_K}
Let $B\subseteq \RdR$, $\ell(B)>0$, be arbitrary and let $C,D\subseteq\M$ be fixed Borel sets with $\nu(C),\nu(D)>0$. 
The {\em marked spatio-temporal second-order reduced moment measure} of a SOIRS MSTPP $Y$ is defined as
\begin{align}\label{eq:K_measure}
&\mathcal K^{CD} (E)
=
\frac{1}{\ell(B)\nu(C)\nu(D)}
\times
\\
&\times\E\left[ 
\mathop{\sum\nolimits\sp{\ne}}_{(x_1,t_1,m_1),(x_2,t_2,m_2)\in Y}
\dfrac{\1\{ (x_1,t_1,m_1)\in B\times C \} \1 \{(x_2,t_2)\in (x_1,t_1)+E\}\1\{m_2\in D \}}{\lambda(x_1,t_1,m_1)\lambda(x_2,t_2,m_2)} \right],
\nonumber
\end{align} 
for $E\in\BB(\RdR)$ (through measure extension of locally finite measures on the ring of bounded Borel sets (see e.g.\ \cite{Halmos})).
\end{definition}

By the Campbell formula, Fubini's theorem and the translation invariance obtained under SOIRS, the measure $\mathcal K^{CD}(\cdot)$ defined in \eqref{eq:K_measure} satisfies 
\begin{align}\label{eq:K_measure_def_g}
\mathcal K^{CD}(E)
&=
\frac{1}{\ell(B)\nu(C)\nu(D)}
\int_{B}\int_C\int_{(x_1,t_1)+E}\int_D
g((x_1,t_1,m_1),(x_2,t_2,m_2))
\de x_2 \de t_2
\de x_1 \de t_1
\nn\\
&=\frac{1}{\nu(C)\nu(D)}
\int_C\int_D 
\left[ \int_E g((0,0,m_1),(x,s,m_2))\de x \de s \right] \nu(dm_2)\nu(dm_1)
,
\end{align}
whereby expression \eqref{eq:K_measure} does not depend on the choice of $B$. 
Note that $\mathcal K^{CD}(\cdot)$ does not depend on the choice of $B\in\BB(\RdR)$. 
For a Poisson process on $\RdRM$ we have that 
$$
\mathcal K^{CD}(E) = \ell(E),
$$
since $\rho^{(2)}((x_1,t_1,m_1),(x_2,t_2,m_2)) = \lambda(x_1,t_1,m_1)\lambda(x_2,t_2,m_2)$.

Writing 
\bea
\label{Y_C}
Y_C=\{(x,t):(x,t,m)\in Y, m\in C\}\subseteq Y_g
\eea 
for the collection of points of $Y_g$ that have marks belonging to $C\in\BB(\M)$ (i.e.\ the projection of $Y$ on $\RdR$), note that we do not necessarily have that $Y_C\cap Y_D=\emptyset$, since we have allowed that $C\cap D\neq\emptyset$. However, it may be highly unnatural to consider $C$ and $D$ such that $C\cap D\neq\emptyset$. 

%
%
%
%

Recalling $\E_C^{!(x,t)}[\cdot]$ from Definition \ref{def:RedPalmMark}, we may obtain a further representation and interpretation of $\mathcal K^{CD}(E)$. By the reduced Campbell-Mecke formula and \eqref{ReducedPalmMark}, 
\begin{align}
\label{ReducedMomentPalm}
\mathcal K^{CD}(E)
=
\frac{1}{\ell(B)\nu(D)}
\int_{B}
\E_C^{!(x_1,t_1)}
\left[
\sum_{(x_2,t_2,m_2)\in Y\cap((x_1,t_1)+E)\times D}
\frac{1}{\lambda(x_2,t_2,m_2)}
\right]
\de x_1 \de t_1
.
\end{align} 
In other words, $\mathcal K^{CD}(E)$ may be obtained either through averaging over the mark space, as in \eqref{eq:K_measure_def_g}, or through averaging over the spatio-temporal domain, as in \eqref{ReducedMomentPalm}.

\subsubsection{Changing the order of the mark sets} 
It may be noted that $\mathcal K^{CD}(\cdot)$ 
is not necessarily symmetric in $C$ and $D$, i.e.\ it is not certain that $\mathcal K^{CD}(\cdot)= \mathcal K^{DC}(\cdot)$ in general. The next result, which is proved in Appendix \ref{appendix_proofs},  provides some conditions under which this is satisfied. The main function of the result is to indicate that estimators of marked spatial dependence between points with marks in $C$ and $D$, which are based on Definition \ref{def_K}, may look a bit different depending on the order chosen for $C$ and $D$. In addition, it may be used to test hypotheses for the marking of $Y$ (see Section \ref{RandomLabellingTest}). 

\begin{theorem}\label{TheoremCommutativity}
Let $Y$ be a SOIRS MSTPP and consider any Borel mark sets $C,D\subseteq\M$, $C\neq D$, with $\nu(C),\nu(D)>0$. 
Either of
\begin{enumerate}
\item $f_{(x_1,t_1),(x_2,t_2)}^{\M}(m_1,m_2) = f_{(x_1,t_1)}^{\M}(m_1) f_{(x_2,t_2)}^{\M}(m_2)$,  which includes $Y$ being independently marked (and thus randomly labelled), 

\item $Y$ has a common mark distribution $M(\cdot)$ and, conditional on the associated locations in $\RdR$, any two marks $m_i$, $m_j$, $i\neq j$, are exchangeable random variables (this includes them being pairwise independent),

\end{enumerate} 
implies that the measures $\mathcal K^{CD}(\cdot)$ and $ \mathcal K^{DC}(\cdot)$ coincide.
\end{theorem}

Note that the conditional exchangeability in Theorem \ref{TheoremCommutativity} refers to that, for almost every $(x_1,t_1)\neq(x_2,t_2)$,
\beann
M^{(x_1,t_1),(x_2,t_2)}(C_1\times C_2)
&=&
\int_{C_1\times C_2}
f_{(x_1,t_1),(x_2,t_2)}^{\M}(m_1,m_2)
\nu(dm_1)\nu(dm_2)
\\
&=&
\int_{C_1\times C_2}
f_{(x_1,t_1),(x_2,t_2)}^{\M}(m_2,m_1)
\nu(dm_1)\nu(dm_2)
\\
&=&
M^{(x_1,t_1),(x_2,t_2)}(C_2\times C_1)
,
\quad 
C_1, C_2\in\BB(\M)
.
\eeann
By de Finetti's theorem, this is equivalent to saying that, conditionally on the ground locations, pairwisely the marks can be expressed as mixtures of iid random variables. 

\begin{remark}
As an alternative one could proceed by considering a symmetrised version $ \overline{\mathcal K}^{DC}(E)=(\mathcal K^{CD}(E)+\mathcal K^{DC}(E))/2$, $E\in\BB(\RdR)$. 

\end{remark}


\subsection{Marked spatio-temporal inhomogeneous $K$-functions}

We have defined a marked inhomogeneous spatio-temporal measure, $\mathcal K$, to quantify second-order interactions. 
By specifying the set $E$ we may obtain different measures of spatio-temporal interaction between points with different mark classifications $C$ and $D$. In what follows we will look closer at such choices. 

Assume that $Y$ is SOIRS and consider two mark sets $C,D\in\BB(\M)$, with $\nu(C),\nu(D)>0$. 
A first natural candidate for $E$ would be the closed, origin centred ball $B[(0,0),r]=\{(x,s):d_{\infty}((0,0),(x,s))\leq r\}=\{(x,s):\|x\|\leq r, |s|\leq r\}$ of radius $r\geq0$ (recall the space-time metric in Appendix \ref{appendix_metric}), where $\|\cdot\|$ is an abbreviation of $\|\cdot\|_{\Rd}$. Hereby we would obtain a direct extension of the marked inhomogeneous $K$-function of \cite{cronie_marked} to the spatio-temporal setting:
\begin{align*}
&K_{\rm inhom}^{CD}(r) = \mathcal K^{CD} (B[(0,0),r]) =
\frac{1}{\ell(B)\nu(C)\nu(D)}\times
\\
&\times
\E\left[ 
\sum_{(x_1,t_1,m_1)\in Y\cap B\times C}
\sum_{(x_2,t_2,m_2)\in Y\setminus\{(x_1,t_1,m_1)\}}
\frac{\1 \{ (x_2,t_2,m_2)\in B[(x_1,t_1),r]\times D\}}
{
\lambda(x_1,t_1,m_1)\lambda(x_2,t_2,m_2)}
\right]
\\
&=\frac{1}{\nu(C)\nu(D)}
\int_C\int_D 
\int_{B[(0,0),r]} g((0,0,m_1),(x,s,m_2))\ell(d(x,s)) 
\nu(dm_2)\nu(dm_1)
.
\end{align*}
However, since the spatial scale is different than the temporal scale, it is more natural to treat space and time lags separately. 
Hence, we instead choose $(x_1,t_1)+E$ to be the closed cylinder $\mathcal{C}_r^t(x_1,t_1)$, with centre $(x_1,t_1)\in\RdR$, radius $r\geq0$ and height $t\geq0$, i.e.
$$
\mathcal{C}_r^t(x_1,t_1)
=
(x_1,t_1)+\mathcal{C}_r^t(0,0)
=
\{(x_2,t_2)\in \RdR: \|x_1-x_2\|\leq r, |t_1-t_2|\leq t\}.
$$
Note that when $d=2$, $\mathcal{C}_r^t(0,0)$ is obtained by taking a disk ($2$-dimensional Euclidean ball) with radius $r$ and stretching it in the $t$-dimension until it becomes the cylinder of height $2t$. 
Furthermore, $B[(0,0),r]=\mathcal{C}_r^r(0,0)$, whereby $K_{\rm inhom}^{CD}(r) 
= \mathcal K^{CD}(\mathcal{C}_r^r(0,0))$. 

\begin{definition}\label{def:K_function}
For any SOIRS MSTPP $Y$ and mark sets $C,D\in\BB(\M)$, $\nu(C),\nu(D)>0$, the {\em marked inhomogeneous spatio-temporal $K$-function} is defined as
\bea
\label{eq:K_function}
K_{\rm inhom}^{CD}(r,t) &=& \mathcal K^{CD}(\mathcal{C}_r^t(0,0))
\\
&=&
\frac{1}{\ell(B)\nu(C)\nu(D)}
\times
\nonumber
\\
&&\times\E\left[ 
\sum_{(x_1,t_1,m_1)\in Y\cap B\times C}
\sum_{(x_2,t_2,m_2)\in Y\setminus\{(x_1,t_1,m_1)\}}
\frac{\1\{(x_2,t_2,m_2)\in\mathcal{C}_r^t(x_1,t_1)\times D\}}
{
\lambda(x_1,t_1,m_1)\lambda(x_2,t_2,m_2)}
\right]
\nonumber
\\
&=&
\frac{1}{\nu(C)\nu(D)}
\int_C\int_D 
\int_{\|x\|\leq r} \int_{-t}^{t} g((0,0,m_1),(x,s,m_2))\de x \de s 
\nu(dm_2)\nu(dm_1)
\nonumber
\eea
%
for $r,t\geq0$ and any $B\in\RdR$, $\ell(B)>0$, by expression  \eqref{eq:K_measure_def_g}. 
Note that $K_{\rm inhom}^{CD}(r,r) =K_{\rm inhom}^{CD}(r)$. 
The special cases of multivariate, directional and/or stationary versions of $K_{\rm inhom}^{CD}(r,t)$ are covered in Appendix \ref{appendix_Multivariate_Stationary}. 
\end{definition}

To connect Definition \ref{def:K_function} with $K_{\rm inhom}^{Y_g}(r,t) = K_{\rm inhom}(r,t)$, i.e.\ the inhomogeneous spatio-temporal (ground) $K$-function in \citep{GabrielDiggle,ghorbani}, we note that 
$K_{\rm inhom}^{\M\M}(r,t)$ reduces to $K_{\rm inhom}^{Y_g}(r,t)$ if the reference measure is given by the mark distribution (recall Definition \ref{DefCommonMark}). 
Furthermore, 
when $Y$ is SOIRSI (recall \eqref{SOIRSI}) with $g((x_1,t_1,m_1),(x_2,t_2,m_2))=\overline{g}(\|x_1-x_2\|,|t_1-t_2|,m_1,m_2)$,  
by a transformation (to hyper-spherical coordinates),
\begin{align*}
K_{\rm inhom}^{CD}(r,t) 
& = \frac{1}{\nu(C)\nu(D)}
\int_C\int_D \int_{\|x\|\leq r} \int_{-t}^{t}
\overline{g}(\|x\|,|s|,m_1,m_2)
\de x \de s
\nu(dm_2)\nu(dm_1)\\
& =\frac{1}{\nu(C)\nu(D)}
\int_C\int_D
\int_{-t}^{t} \int_{0}^r 
\omega_d \overline{g}(u,v,m_1,m_2)u^{d-1}
\de u \de v 
\nu(dm_2)\nu(dm_1)
\end{align*}
and we note the resemblance with $K_{\rm inhom}(r,t)$. 

To give the motivation behind $K_{\rm inhom}^{CD}(r,t)$, recall that for a Poisson process on $(\R^2\times\R)\times\M$ we have that 
\begin{align*}
K_{\mathrm{inhom}}^{CD}(r,t) = \ell(\mathcal C_r^t(0,0))
= 2t r^d\omega_d = 2t r^d \pi^{d/2}/(\Gamma(d/2 +1)),
\end{align*}
where $\omega_d=\pi^{d/2}/(\Gamma(d/2 +1))$ is the (Lebesgue) volume of the $d$-dimensional Euclidean unit ball and $\Gamma(\cdot)$ is the Gamma function. 
In other words, for any $r,t\geq0$  \citep{GabrielDiggle,diggle_book}:
%
\begin{itemize}
\item
If $K_{\mathrm{inhom}}^{CD}(r,t) > 2\omega_d r^d t$ we have an indication that points of $Y$ with marks in $D$ have a tendency to cluster around the points with marks in $C$ (in a pairwise sense), having compensated for the  inhomogeneity. This is referred to as {\em clustering/aggregation}. 

\item When $K_{\mathrm{inhom}}^{CD}(r,t) < 2\omega_d r^d t$, points with marks in $D$ tend to avoid being close to the points with marks in $C$ (in a pairwise sense), taking the inhomogeneity into account. This is called {\em regularity/inhibition}.

\end{itemize}
Here closeness is understood in terms of one of the points being inside the cylinder neighbourhood $\mathcal C_r^t$ of the other. 
In other words, we have defined a way of measuring spatio-temporal interaction between points belonging to two mark sets $C$ and $D$, in terms of spatial lags $r\geq0$ and temporal lags $t\geq0$, in the presence of inhomogeneity.

\subsection{Further properties}

\subsubsection{Independent thinning}\label{Section_Thinning}
The first thing that may be pointed out is that when applying independent thinning to $Y$, i.e.\ when we retain each $(x,t,m)\in X$ according to some probability function $0\leq p(x,t,m)\leq1$, $(x,t,m)\in\RdRM$, the pcf of the thinned process coincides with the original one \citep{baddeley_Kfunction}. This implies that $\mathcal K^{CD}(\cdot)$ and $K_{\mathrm{inhom}}^{CD}(r,t)$ are not affected by this type of thinning. 

\subsubsection{Scaling}\label{Section_Scaling}
We next give a scaling result, which indicates the relationship between the two definitions of $K$-functions. 
Its proof is given in Appendix \ref{appendix_proofs}. 
\begin{theorem}\label{TheoremRescaling}
Consider $C,D\in\BB(\M)$ with $\nu(C)$ and $\nu(D)$ positive. 
For any $\beta=(\beta_S,\beta_T)\in(0,\infty)^2$ and a SOIRS MSTPP $Y=\{(x_i,t_i,m_i)\}_{i=1}^N$, define the rescaling 
$$
\beta Y=\{(\beta_S x_i,\beta_T t_i,m_i)\}_{i=1}^N. 
$$
The marked inhomogeneous spatio-temporal K-function $K_{\rm inhom}^{CD}(r, t;\beta)$ of $\beta Y$ satisfies
\[
K_{\rm inhom}^{CD}(r, t;(\beta_S,\beta_T)) 
= K_{\rm inhom}^{CD}(r/\beta_S, t/\beta_T;(1,1))
,
\quad r,t\geq0,
\]
where $K_{\rm inhom}^{CD}(r, t;(1,1))=K_{\rm inhom}^{CD}(r, t)$ is the marked inhomogeneous spatio-temporal K-function of $Y$. 

\end{theorem}

Theorem \ref{TheoremRescaling} essentially tells us two things. To begin with, if we rescale the spatial and/or the temporal domain, and thereby the space-time locations of $Y$, then $K_{\rm inhom}^{CD}(\cdot)$ changes in a natural way. Secondly, we note that, equivalently,
\[
K_{\rm inhom}^{CD}(r, t)
=
K_{\rm inhom}^{CD}(r\beta_S, t\beta_T;(\beta_S,\beta_T)) 
,
\quad r,t\geq0,
\]
whereby $K_{\rm inhom}^{CD}(r, t)=K_{\rm inhom}^{CD}(r, r;(1,r/t))=K_{\rm inhom}^{CD}(r;(1,r/t))$, $r\geq0$. In other words, $K_{\rm inhom}^{CD}(r, t)$ may always be obtained through $K_{\rm inhom}^{CD}(r)$ by applying proper scaling of $Y_g$, i.e.\ considering $\beta Y=\{(x_i,\beta_T t_i,m_i)\}_{i=1}^N$,  where $\beta_T=r/t$. 
There are practical implications of this results; it is sufficient to define an estimator for $K_{\rm inhom}^{CD}(r)$, $r\geq0$ (however, this is not the choice that we will make when we define our estimators). 

It may be noted from the proof of Theorem \ref{TheoremRescaling} that we may obtain a more general result, pertaining to $\mathcal K^{CD}(E)$. More specifically, we have that the marked spatio-temporal second-order reduced moment measure $\mathcal K^{CD}(\cdot;\beta)$ of $\beta Y$ satisfies $\mathcal K^{CD}(E;\beta)=\mathcal K^{CD}(\{(\beta_S x,\beta_Ts):(x,s)\in E\})$.

\section{Statistical inference}\label{inference}

The intensity function as well as our second-order summary statistics are probabilistic entities used to quantify first and second-order properties of a given point process. Turning to the real world, where we are given a marked spatio-temporal point pattern $\{(x_i,t_i,m_i)\}_{i=1}^n$, such as the earthquake data set, we are naturally interested in how we statistically can estimate these quantities, to better understand the data-generating mechanism in question. 
We do this by assuming that we have observed a realisation of a SOIRS MSTPP $Y$. Its ground process, $Y_g$, is formally defined on $\RdR$ but in practice we treat it as only observed within some bounded spatio-temporal region $W_S\times W_T\subseteq \RdR$, which is often referred to as the {\em study region}. 
We also restrict ourselves to the case where only one single point pattern is observed but we point out that most arguments below can be averaged over if one would have repetitions. 

Being able to estimate the relevant quantities, we then proceed to considering different specific marking structures (recall Section \ref{SectionMarkingStructures}). In particular we will consider some related hypothesis testing. In Appendix \ref{appendix_Multivariate_Stationary} we look closer at the multivariate/multi-type, stationary and anisotropic cases.

\subsection{Voronoi intensity estimation}\label{Section_IntensityEstimation}

Writing $N=Y(W_S\times W_T\times\M)$, if we can assume homogeneity in space-time, with $\nu(\cdot)=M(\cdot)$, so that $\lambda(x,t,m)\equiv\lambda>0$, we simply estimate $\lambda$ by means of $N/[\ell\otimes\nu](W_S\times W_T\times\M)$. This is, however, a scenario that is rarely or never seen in practise, in particular not in the case of earthquakes.

As pointed out in \citep[Section 3.2]{Vere-Jones2009}, when estimating the intensity function of a MSTPP, unless one can assume homogeneity, one should use a local/adapted/variable approach, as opposed to global smoothing techniques, such as single bandwidth kernel estimators \citep{diggle_book,MCintensity,Silverman}. 
Motivated by \cite{BarrSchoenberg}, and in particular their study of earthquakes (in a purely spatial setting), we choose to consider a marked spatio-temporal version of the {\em Voronoi intensity estimator}.

We start by defining the Voronoi estimators for $\lambda(x,t,m)$  and $\lambda_g(x,t)$. 
They are constructed through Voronoi tessellations (see e.g.\ \citep{stoyan}) generated by the metrics $d_{\infty}(\cdot,\cdot)$ and $d(\cdot,\cdot)$ in expression \eqref{Metric}. 

\begin{definition}
The {\em spatio-temporal Voronoi intensity estimator} is defined by
\bea
\label{STPPVoronoi}
\widehat\lambda_g(x,t)&=&\sum_{(y,v)\in Y_g\cap W_S\times W_T}
\frac{\1\{(x,t)\in\V_{(y,v)}^g\cap W_S\times W_T\}}{\ell(\V_{(y,v)}^g\cap W_S\times W_T)},
\qquad (x,t)\in W_S\times W_T,
\eea
where the {\em spatio-temporal Voronoi tessellation} is given by 
\beann
\V_g = \{\V_{(x,t)}^g\}_{(x,t)\in Y_g}
=
\big\{
(u,v)\in\RdR &:& d_{\infty}((u,v),(x,t))\leq d_{\infty}((u,v),(y,s)) 
\\
&&\text{ for any } (y,s)\in Y_g\setminus\{(x,t)\}
\big\}_{(x,t)\in Y_g}
.
\eeann

Recalling the metric $d'(\cdot,\cdot)$ in \eqref{Metric}, the {\em marked spatio-temporal Voronoi tessellation} generated by $Y$ is defined as $\V=\{\V_{(x,t,m)}\}_{(x,t,m)\in Y}$, where 
\beann
\V_{(x,t,m)}
=\big\{(u,v,z)\in\RdRM &:& d((x,t,m),(u,v,z))\leq d((y,s,k),(u,v,z)) 
\\
&& \text{ for any } (y,s,k)\in Y\setminus\{(x,t,m)\}
\big\}_{(x,t,m)\in Y}
\\
=\big\{(u,v,z)\in\RdRM &:& 
\max\{\|x-u\|_{\Rd},|t-v|,d'(m,z)\}\leq 
\\
&&\leq \max\{\|y-u\|_{\Rd},|s-v|,d'(k,z)\}
\\
&& \text{ for any } (y,s,k)\in Y\setminus\{(x,t,m)\}
\big\}_{(x,t,m)\in Y}
.
\eeann
Furthermore, the {\em marked spatio-temporal Voronoi intensity estimator} is defined as
\bea
\label{MSTPPVoronoi}
\widehat\lambda(x,t,m)
=
\sum_{(x_i,t_i,m_i)\in Y\cap W_S\times W_T\times\M}
\frac{
\1\{(x,t,m)\in\V_{(x_i,t_i,m_i)}\}
}{[\ell\otimes\nu](\V_{(x_i,t_i,m_i)}\cap W_S\times W_T\times\M)}
,
\eea
for $(x,t,m)\in W_S\times W_T\times\M$. 

\end{definition}

Note the explicit dependence on the choice of space-time-mark metric and reference measure above. 

We next give the mass preservation and the unbiasedness of the estimators above. 

\begin{theorem}\label{TheoremVoronoi}
The estimators \eqref{STPPVoronoi} and \eqref{MSTPPVoronoi} are mass-preserving and unbiased, i.e.\ they integrate to the total number of points $N$ and their expectations coincide with the actual corresponding intensities at almost every location ($W_S\times W_T$ or $W_S\times W_T\times\M$). 
\end{theorem}

\subsubsection{Simplifying assumptions}
Ideally, one does not impose too many conditions when finding the intensity estimate, unless convinced that specific conditions such as separability hold. 
We will next look at a few scenarios where we impose simplifying assumptions and we note that the need for them 
often is related to computational expenses. 

We here need to introduce the Voronoi cells of the projections of $Y$ (assuming that they are well defined). 
Recalling the projections $Y_S$ and $Y_T$ from \eqref{Projections} and defining the projection $Y_{M}$ of $Y$ on $\M$ in an identical fashion, let
\bea
\label{PartialVoronoi}
\V_S &=& \{\V_{x}^S\}_{x\in Y_S}
=\{u\in\Rd : \|u-x\|_{\R^d}\leq\|u-y\|_{\R^d} \text{ for any } y\in Y_S\setminus\{x\}\}_{x\in Y_S}
,
\nonumber
\\
\V_T &=& \{\V_{t}^T\}_{t\in Y_T}
=\{v\in\R : |v-t|\leq |v-s| \text{ for any } s\in Y_T\setminus\{t\}\}_{t\in Y_T}
,
\nonumber
\\
\V_M &=& \{\V_{m}^M\}_{m\in Y_M}=\{z\in\M : d'(m,z) \leq d'(k,m)\text{ for any } k\in Y_M\setminus\{m\}\}_{m\in Y_M}
,
\nonumber
\\
\V_{T\times M}
&=&
\{\V_{(t,m)}^{T\times M}\}_{(t,m)\in\R\times\M}
\nonumber
\\
&=&
\big\{
(v,z)\in\R\times\M : 
\max\{|t-v|,d'(m,z)\}\leq \max\{|s-v|,d'(k,z)\}
\nonumber
\\
&&\text{ for any } (s,k)\in Y_T\times Y_M\setminus\{(t,m)\}
\big\}_{(t,m)\in Y_T\times Y_M}
.
\eea
Some simplified setups are given by:
\begin{enumerate}

\item Separability and a common mark distribution: 
\begin{align*}
&\widehat\lambda(x,t,m)
=
\frac{1}{N^2}
\widehat\lambda_S(x)\widehat\lambda_T(t)\widehat\lambda_M(m)
\\
&=
\frac{1}{N^2}
\sum_{y\in Y_S\cap W_S}\frac{\1\{x\in\V_y^S\cap W_S\}}{\ell_d(\V_y^S\cap W_S)}
\sum_{v\in Y_T\cap W_T}\frac{\1\{t\in\V_v^T\cap W_T\}}{\ell_1(\V_v^T\cap W_T)}
\sum_{z\in Y_M}\frac{\1\{m\in\V_z^M\}}{\nu(\V_z^M)}
.
\end{align*}
If we assume that the common mark distribution is given by $\nu(\cdot)$, we set $\lambda_M(m)/N\equiv1$ above.

\item Non-separability and a common mark distribution: 
\[
\widehat\lambda(x,t,m)
=\widehat f^{\M}(m) \widehat\lambda_g(x,t)
=\frac{\widehat\lambda_M(m)}{N} \widehat\lambda_g(x,t)
.
\]
If the mark distribution and the reference measure coincide, we set $\widehat f^{\M}(m) \equiv1$ above.

\item Separability and time-mark dependence:
\bea
\label{IntEstSep}
\widehat\lambda(x,t,m)
=\frac{\widehat\lambda_S(x)}{N}
\sum_{(v,z)\in Y_T\times Y_M\cap W_T\times\M}
\frac{\1\{(t,m)\in\V_{(v,z)}^{T\times M}\cap W_T\times\M\}
}{[\ell_1\otimes\nu](\V_{(v,z)}^{T\times M}\cap W_T\times\M)}
.
\eea
The case of separability and space-mark dependence is analogous.

\end{enumerate}

As a corollary to Theorem \ref{TheoremVoronoi} (the proof is identical), we obtain mass preservation and unbiasedness for the estimators above.

\begin{cor}
All the estimators above are mass preserving and unbiased. 
\end{cor}

\subsection{Estimation of the second-order summary statistics}\label{SectionEstimation}

%


We next give the definitions of the estimators of our previously defined second-order statistics. 
In order to account for edge effects \citep{CronieEdge, stoyan, GabrielEdgeEffects} when defining the estimators below, we apply a \emph{minus sampling/border correction scheme}. 
Denoting the boundaries of $W_S$ and $W_T$ by $\partial W_S$ and $\partial W_T$, respectively, we write $W_S^{\ominus r} = \{ x\in W_S: d_{\Rd}(x,\partial W_S)\geq r\}=\{x\in W_S: B_{\R^d}[x,r]\subseteq W_S\}$ and $W_T^{\ominus t} =\{s\in W_T: d_{\R}(x,\partial W_T)\geq t\}$ for the eroded spatial and temporal domains, respectively. 
Here $B_{\R^d}[x,r]$ is the closed ball in $\Rd$ with centre $x$ and radius $r$. 

Throughout we consider a SOIRS MSTPP $Y$, and assume that $\ell_d(W_S^{\ominus r})>0$, $\ell_1(W_T^{\ominus t})>0$ and $C,D\in\BB(\M)$, with $\nu(C),\nu(D)>0$. 
 
\begin{definition}\label{def:K_estimator}
The estimator $\widehat{K}_{\mathrm{inhom}}^{CD}(r,t)$ of the marked inhomogeneous spatio-temporal $K$-function $K_{\mathrm{inhom}}^{CD}(r,t)$, $r,t \geq 0$, based on $Y\cap W_S\times W_T\times\M$, is defined by
\begin{align}
\label{eq:K_estimator}
& 
\ell_d(W_S^{\ominus r})\ell_1(W_T^{\ominus t})\nu(C)\nu(D)\widehat{K}_{\mathrm{inhom}}^{CD}(r,t)
= 
\\
& 
= 
\sum_{(x_1,t_1,m_1)\in Y\cap W_S^{\ominus r}\times W_T^{\ominus t}\times C}
\sum_{(x_2,t_2,m_2)\in Y\cap \mathcal{C}_r^t(x_1,t_1)\times D\setminus\{(x_1,t_1,m_1)\}}
\frac{1}{\lambda(x_1,t_1,m_1)\lambda(x_2,t_2,m_2)}
.\nonumber
\end{align}
By replacing $\mathcal{C}_r^t(x_1,t_1)$ by $(x_1,t_1) + E$ in \eqref{eq:K_estimator}, $E\in\BB(\RdR)$, we obtain an estimator $\widehat{\mathcal{K}}^{CD}(E)$ of the marked spatio-temporal second-order reduced moment measure $\mathcal{K}^{CD}(E)$.


\end{definition}

Next, in Lemma \ref{lemma_estimate}, we turn to the unbiasedness of the estimators above (see Appendix \ref{appendix_proofs} for the proof). 

\begin{lemma}\label{lemma_estimate}
The estimators in Definition \ref{def:K_estimator} 
are unbiased. 
The variance of $\widehat{\mathcal{K}}^{CD}(E)$ is given in expression \eqref{eq:Variance_K_measure}.
\end{lemma}

Clearly, in practise $\lambda(\cdot)$ is not known so each $\lambda(x_i,t_i,m_i)$ must be replaced by an estimate $\widehat{\lambda}(x_i,t_i,m_i)$, which may obtained by e.g.\ the Voronoi estimation approach presented previously. 
Note that a further desirable property, the so-called Hamilton principle, is satisfied by the Voronoi intensity estimation approach: $\sum_{(x,t,m)\in Y\cap W_S\times W_T\times\M}\widehat\lambda(x,t,m)^{-1}=[\ell\otimes\nu](W_S\times W_T\times\M)$ (see \citep{StoyanStoyanRatio}).
Further remarks on the Hamilton principle can be found in Appendix \ref{appendix_Hamilton}.

When there is a common mark distribution $M(\cdot)$, which coincides with the reference measure $\nu(\cdot)$ (recall Definition \ref{DefCommonMark}), 
we may estimate $\nu(C)\nu(D)=M(C)M(D)$, $C,D\in\BB(\M)$, by 
\(
\widehat{\nu(C)}\widehat{\nu(D)}=Y(W_S\times W_T\times C)Y(W_S\times W_T\times D)/Y_g(W_S\times W_T)^2 
\)
to obtain 
%
\begin{align*}
&\ell_d(W_S^{\ominus r})\ell_1(W_T^{\ominus t})\widehat{\nu(C)}\widehat{\nu(D)}
\widehat{K}_{\mathrm{inhom}}^{CD}(r,t)
=
\\
&=
\sum_{(x_1,t_1)\in Y_C\cap W_S^{\ominus r}\times W_T^{\ominus t}}
\sum_{(x_2,t_2)\in Y_D\cap \mathcal{C}_r^t(x_1,t_1)\setminus\{(x_1,t_1)\}} 
\frac{1}{\lambda_g(x_1,t_1)\lambda_g(x_2,t_2)}
,
\end{align*}
where we plug in an estimate of $\lambda_g(\cdot)$ in practise. 

\subsubsection{Smoothing}
Recall from Section \ref{Section_Thinning} that $K_{\rm inhom}^{CD}(r,t)$ is invariant under independent thinning. This may be exploited to obtain a smoothing/thinning/bootstrapping scheme for the estimation of $K_{\rm inhom}^{CD}(r,t)$. More specifically, let $\widehat{K}_{\mathrm{inhom}}^{CD}(r,t;Y_i^p)$, $i=1,\ldots,n$, be the estimators generated by $n$ independent thinnings $Y_i^p$, $i=1,\ldots,n$, of $Y$, using retention probability function $p(x,t,m)\equiv p\in(0,1)$. The resulting smoothed estimator is given by
\[
\widetilde{K}_{\mathrm{inhom}}^{CD}(r,t)
=
\frac{1}{n}\sum_{i=1}^n 
\widehat{K}_{\mathrm{inhom}}^{CD}(r,t;Y_i^p)
.
\]
In essence, we are averaging over $n$ different unbiased estimators of $K_{\mathrm{inhom}}^{CD}(r,t)$; hereby also $\widetilde{K}_{\mathrm{inhom}}^{CD}(r,t)$ is unbiased. 
A clear gain with this approach is that we even out the negative effects of using only one misspecified plug-in intensity estimate, which has been generated by only one sample, as is the case of $\widehat{K}_{\mathrm{inhom}}^{CD}(r,t)$. The drawback is that we get an increased variance. 
Regarding the choice of $p\in(0,1)$, we generally consider $p=0.5$ to be a decent choice (unless the dataset is small, which requires additional caution). 

\begin{remark}
In principle, one could consider bootstrap-type regions/envelopes for $K_{\mathrm{inhom}}^{CD}(r,t)$, based on $\widehat{K}_{\mathrm{inhom}}^{CD}(r,t;Y_i^p)$, $i=1,\ldots,n$, provided that we choose some suitable function space metric (c.f.\ e.g.\ \citep{Myllymaki}). 

\end{remark}

\subsection{Independence assumptions}

We next look closer at how $K^{CD}_{\rm inhom}(r,t)$ is affected by making different independence assumptions that are related to the marking structure. 
Recalling the definitions from Section \ref{SectionMarkingStructures}, 
we start by looking at independent marking, which includes random labelling, to see how $K^{CD}_{\rm inhom}(r,t)$ is affected. We then proceed to considering the scenario where points of $Y$ with marks that belong to different mark sets $C$ and $D$ are independent. 
It should be noted that the main part of the results below, in essence, are translated versions of the results in \citep{cronie_marked}. 

Lemma \ref{LemmaIndependentMarks} below, which is proved in Appendix \ref{appendix_proofs}, suggests finding evidence of independent marking by comparing $K_{\mathrm{inhom}}^{CD}(r,t)$ with its unmarked counterpart, i.e.\ considering $K_{\mathrm{inhom}}^{CD}(r,t)-K_{\mathrm{inhom}}^{Y_g}(r,t)$, where we recall the inhomogeneous $K$-function of the ground process, $K^{Y_g}_{\mathrm{inhom}}(r,t)=\int_{\mathcal C_r^t(0,0)} g_g((0,0),(x,s))\de x \de s$.  

\begin{lemma}\label{LemmaIndependentMarks}
Let $C,D\subseteq\M$ be Borel sets with $\nu(C)$, $\nu(D)>0$ and assume that $Y$ has independent marks. Then, $Y$ and $Y_g$ have the same pcf's (note the equivalence in SOIRS) and $K_{\mathrm{inhom}}^{CD}(r,t)=K_{\mathrm{inhom}}^{Y_g}(r,t)$.
\end{lemma}

We next evaluate Lemma \ref{LemmaIndependentMarks} numerically, to ensure that our estimator is behaving properly.
In order to do so, we simulate 99 realisations of the model given in Example \ref{ExampleIndependentMarkedLGCP} and for the fixed temporal lags $t\in\{0.05,0.10,0.15,0.30\}$ we construct min/max-envelopes (see e.g.\ \cite{diggle_book}) for $\widehat K_{\mathrm{inhom}}^{CD}(r,t) - \widehat K_{\mathrm{inhom}}^{Y_g}(r,t)$, where $C=\{0\}$ and $D=\{1\}$, based on these 99 realisations. 
Figure \ref{bi_lgcp_different_t} shows the envelopes obtained for the different values of $t$ and we see that our estimator is behaving properly since the envelopes centre around $0$. Also, in Figure \ref{bi_lgcp_different_t} we find the estimates of $K_{\mathrm{inhom}}^{CD}(r,t) - K_{\mathrm{inhom}}^{Y_g}(r,t)$ for space lags $r\in[0,0.3]$ and time lags $t\in[0,0.3]$. One can see that the values of the estimated $\widehat K_{\mathrm{inhom}}^{CD}(r,t) - \widehat K_{\mathrm{inhom}}^{Y_g}(r,t)$ are close to $0$.

\begin{figure}[!htbp]
\centering
  \includegraphics*[width=0.3\textwidth]{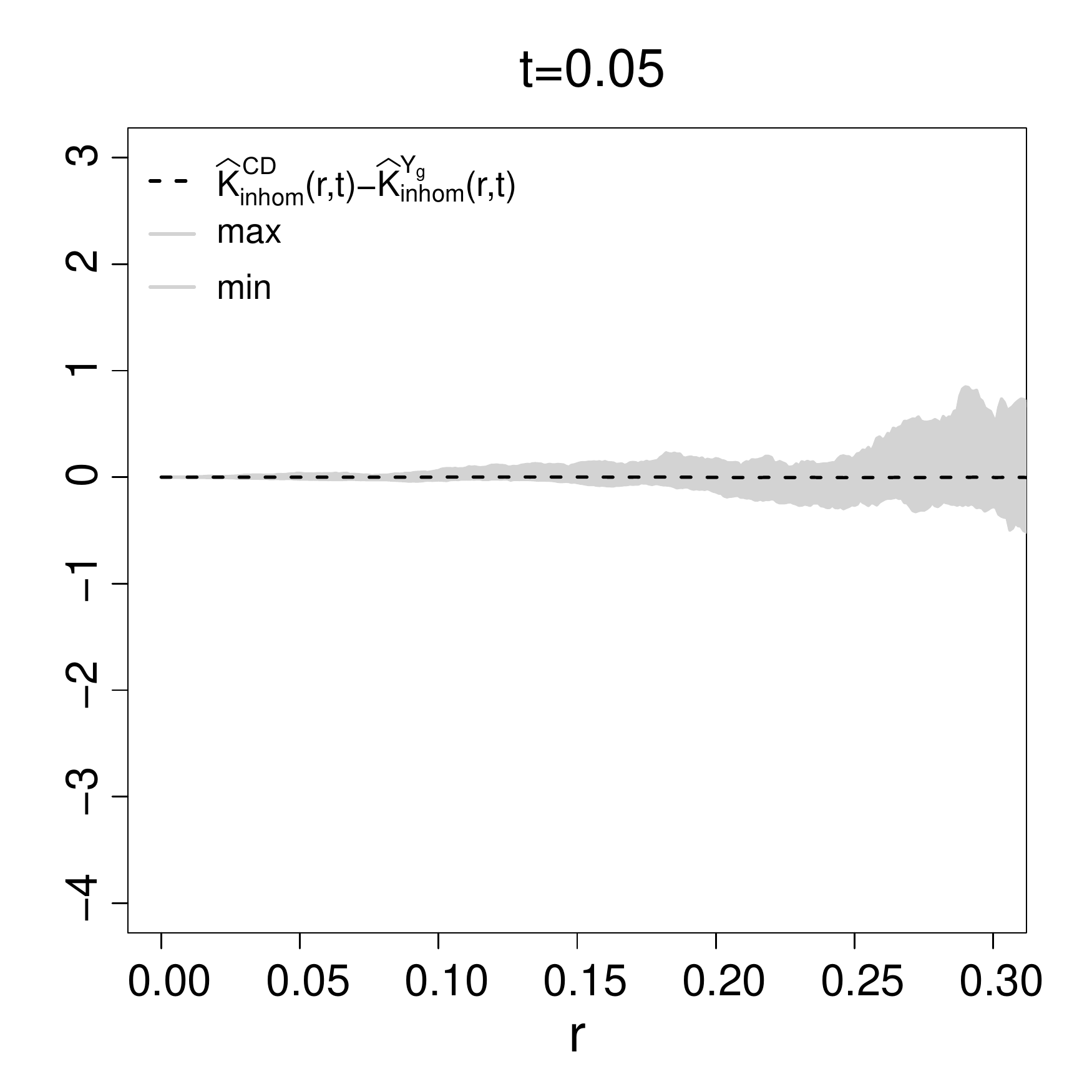}
  \includegraphics*[width=0.3\textwidth]{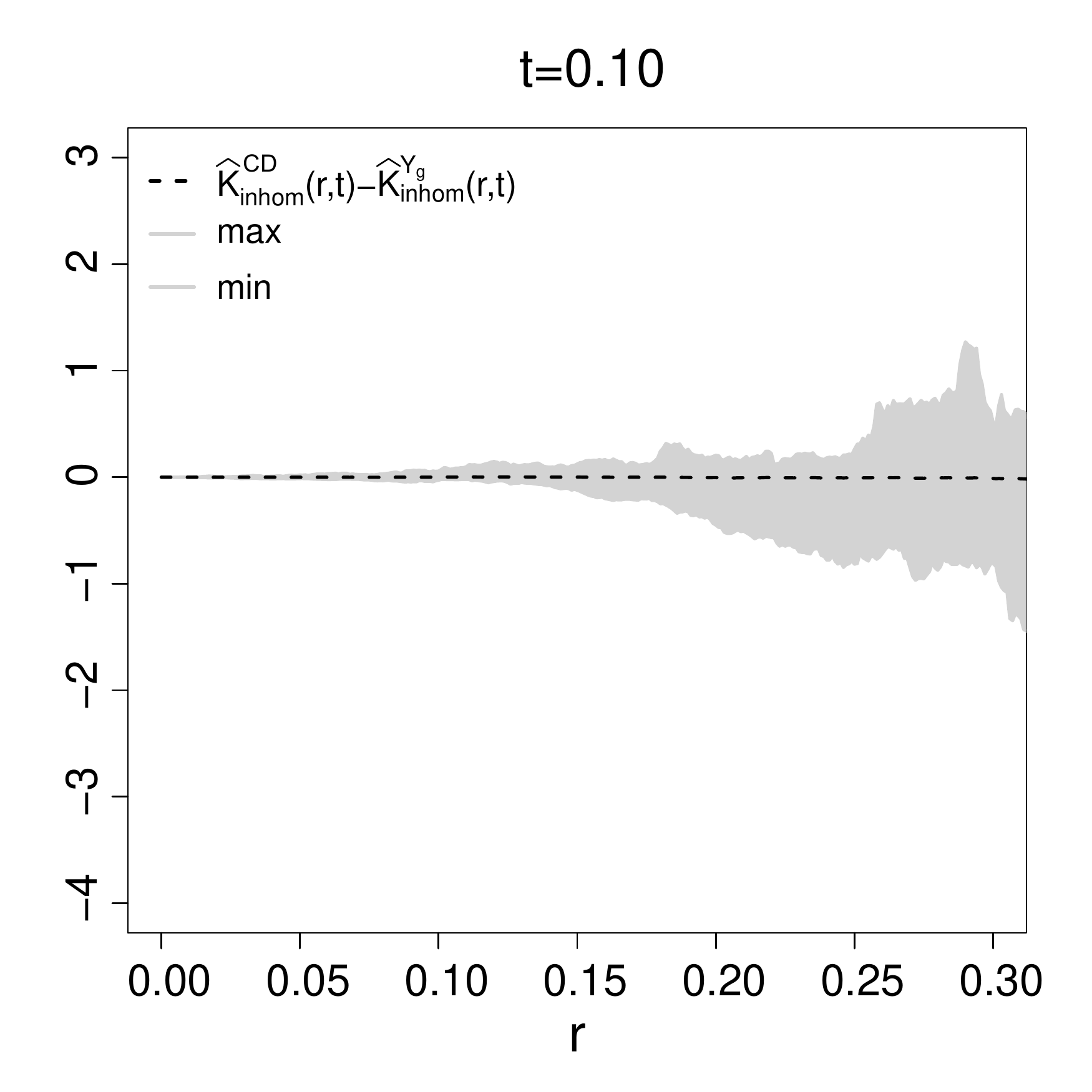}\\
  \includegraphics*[width=0.3\textwidth]{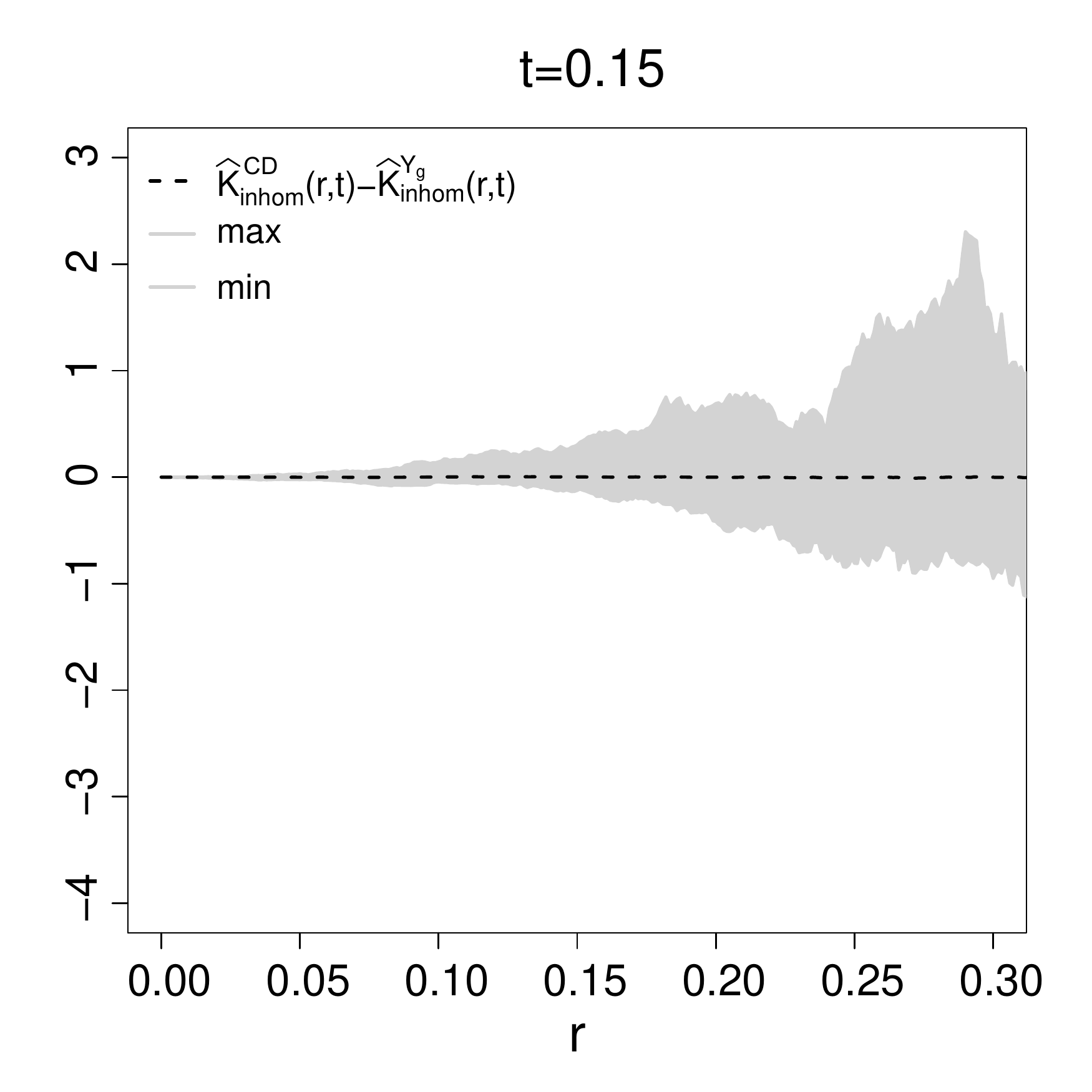}
  \includegraphics*[width=0.3\textwidth]{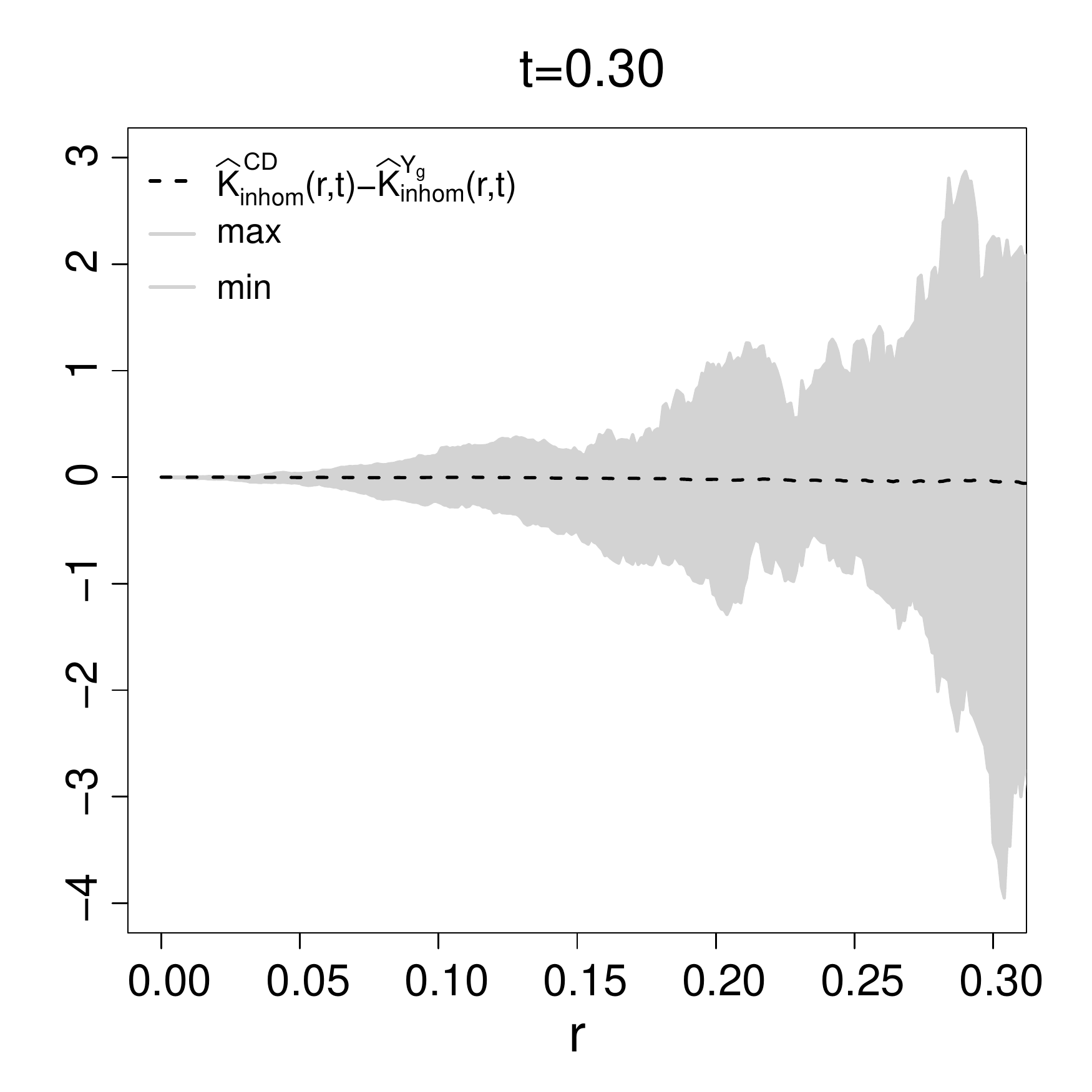}\\
  \includegraphics*[width=0.3\textwidth]{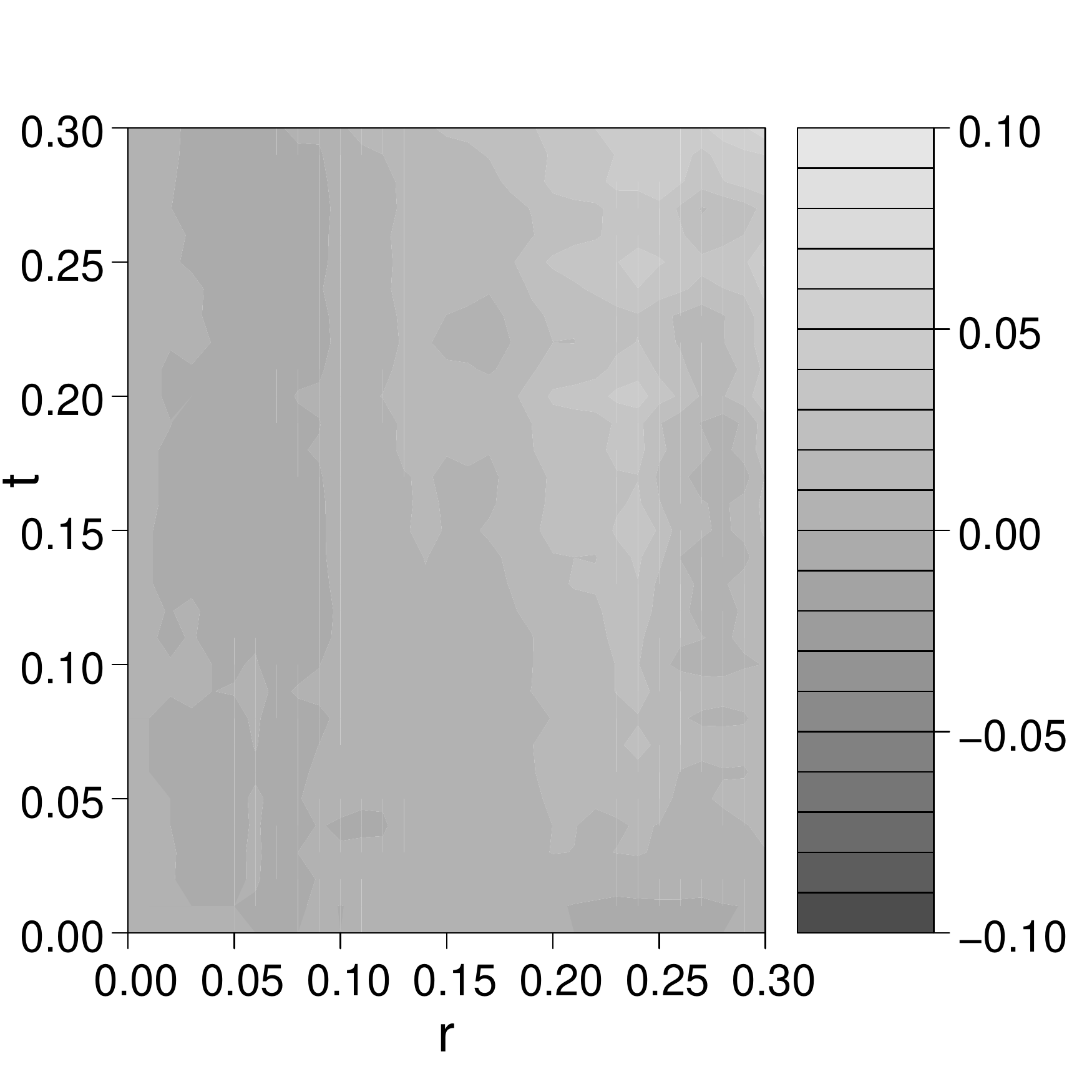}
 \caption{
Upper part: 
Envelopes for the estimate $\widehat K_{\mathrm{inhom}}^{CD}(r,t) - \widehat K_{\mathrm{inhom}}^{Y_g}(r,t)$, where $C=\{0\}$ and $D=\{1\}$, based on 99 realisations of the the model given in Example \ref{ExampleIndependentMarkedLGCP}, for fixed temporal lags $t=0.05$  and $t=0.10$ (upper row), $t=0.15$ and $t=0.30$ (middle row). Lower row: The estimate $\widehat K_{\mathrm{inhom}}^{CD}(r,t) - \widehat K_{\mathrm{inhom}}^{Y_g}(r,t)$, for all space-time lags.}
\label{bi_lgcp_different_t}
\end{figure}

Consider next the concept of independent components, which is the scenario where the restrictions $Y|_C=Y\cap(\RdR\times C)$ and $Y|_D=Y\cap(\RdR\times D)$, with ground processes $Y_C$ and $Y_D$, are independent. This can be exemplified by considering a marked bivariate process $Y=(Y_1,Y_2)$, where each component $Y_j=\{(x_{ij},t_{ij},m_{ij})\}_{i=1}^{N_j}$, $j=1,2$, is a (dependently) marked process, but where $Y_1$ and $Y_2$ are mutually independent. In essence, this is the merging of two mutually independent populations, which have dependent marking structures within. 
Assessing possible dependence between $Y|_C$ and $Y|_D$, Lemma \ref{LemmaIndependentComponents} below, which is proved in the Appendix, suggests comparing  
$K_{\mathrm{inhom}}^{CD}(r,t)$ with $2\omega_dr^dt$; when $D=\M\setminus C$, it further suggests comparing $K_{\mathrm{inhom}}^{C\M}(r,t)$ with $\frac{\nu(\M\setminus C)}{\nu(\M)} 2\omega_dr^dt  + \frac{\nu(C)}{\nu(\M)}  K_{\mathrm{inhom}}^{CC}(r,t)$.

\begin{lemma}\label{LemmaIndependentComponents}
Let $C,D\in\BB(\M)$, with $\nu(C)$ and $\nu(D)>0$, and let be $Y$ is SOIRS, with $Y|_C$ and $Y|_D$ mutually independent. 
It follows that $K_{\mathrm{inhom}}^{CD}(r,t)=2\omega_dr^dt$ and when $D=\M\setminus C$, we have that 
$
K_{\mathrm{inhom}}^{C\M}(r,t)=\frac{\nu(\M\setminus C)}{\nu(\M)} 2\omega_dr^dt  + \frac{\nu(C)}{\nu(\M)}  K_{\mathrm{inhom}}^{CC}(r,t).
$

\end{lemma}

To evaluate the above results numerically, we simulate 99 realisations of the model in Example \ref{ExampleIndependentComp} and consider $K_{\mathrm{inhom}}^{CD}(r,t)-2\omega_dr^dt$, where $C=\{0\}$ and $D=\{1\}$, for each one. The corresponding envelopes, which cover 0, are illustrated in Figure \ref{pois_lgcp_different_t}.

\begin{figure}[!htbp]
\centering
  \includegraphics*[width=0.3\textwidth]{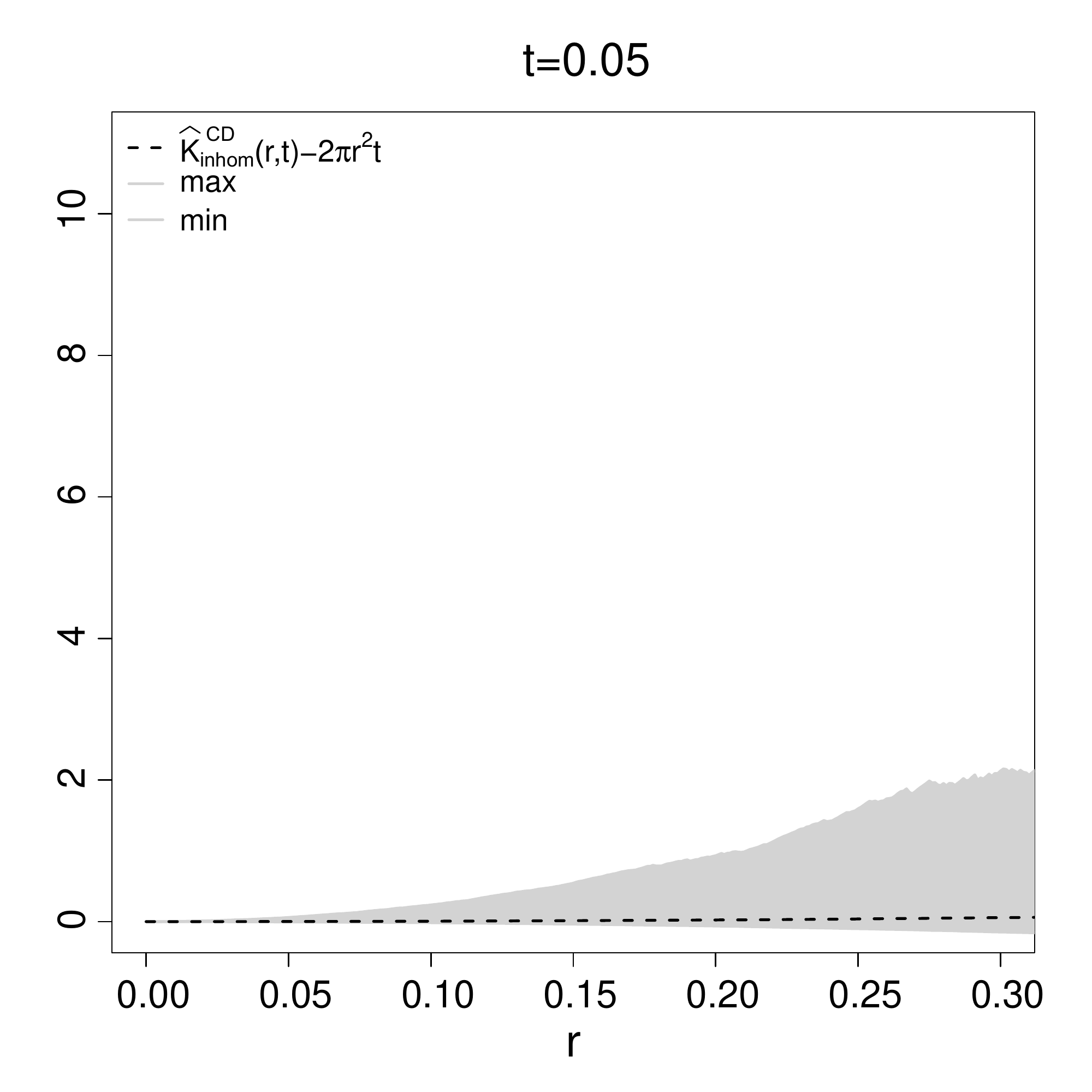}
  \includegraphics*[width=0.3\textwidth]{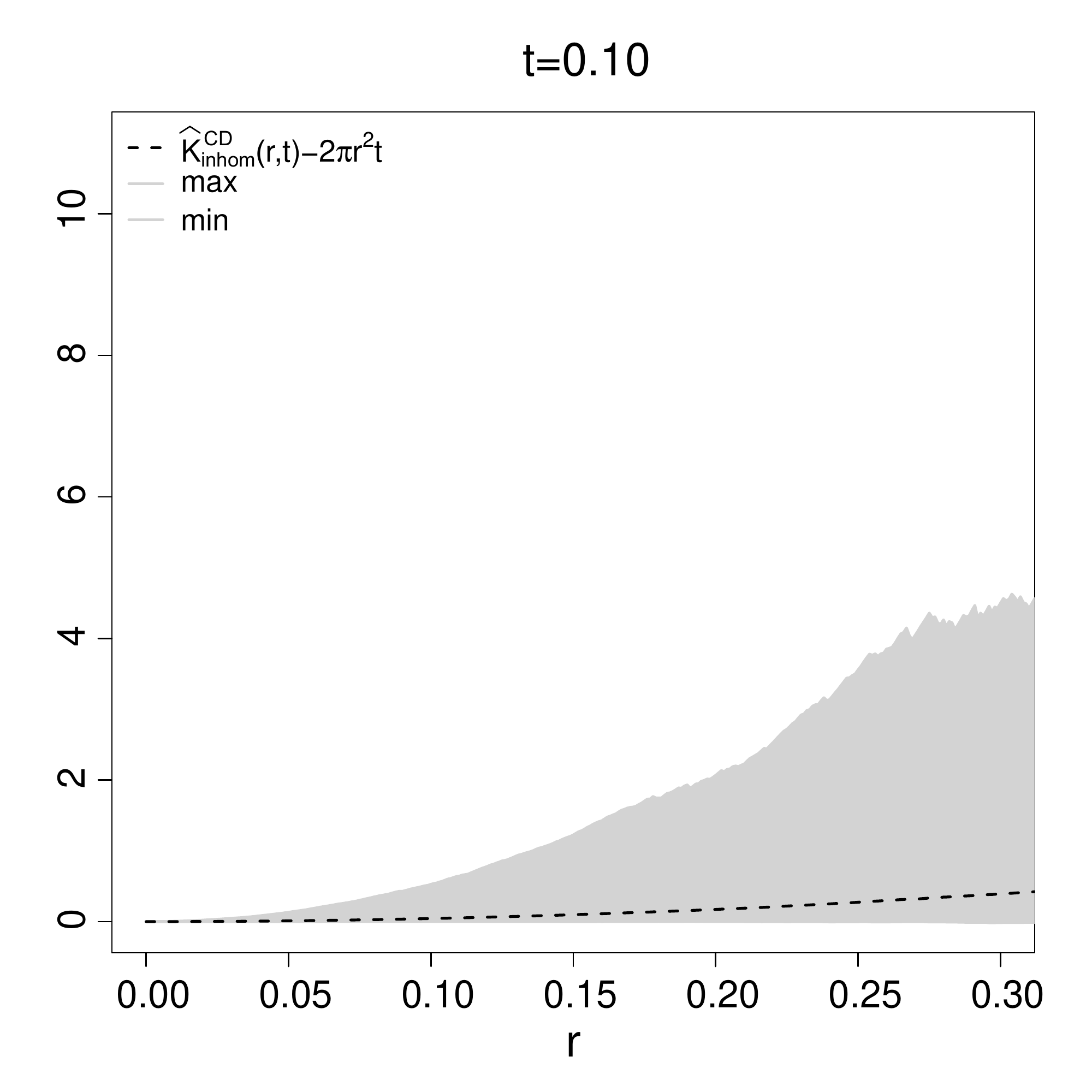}\\
  \includegraphics*[width=0.3\textwidth]{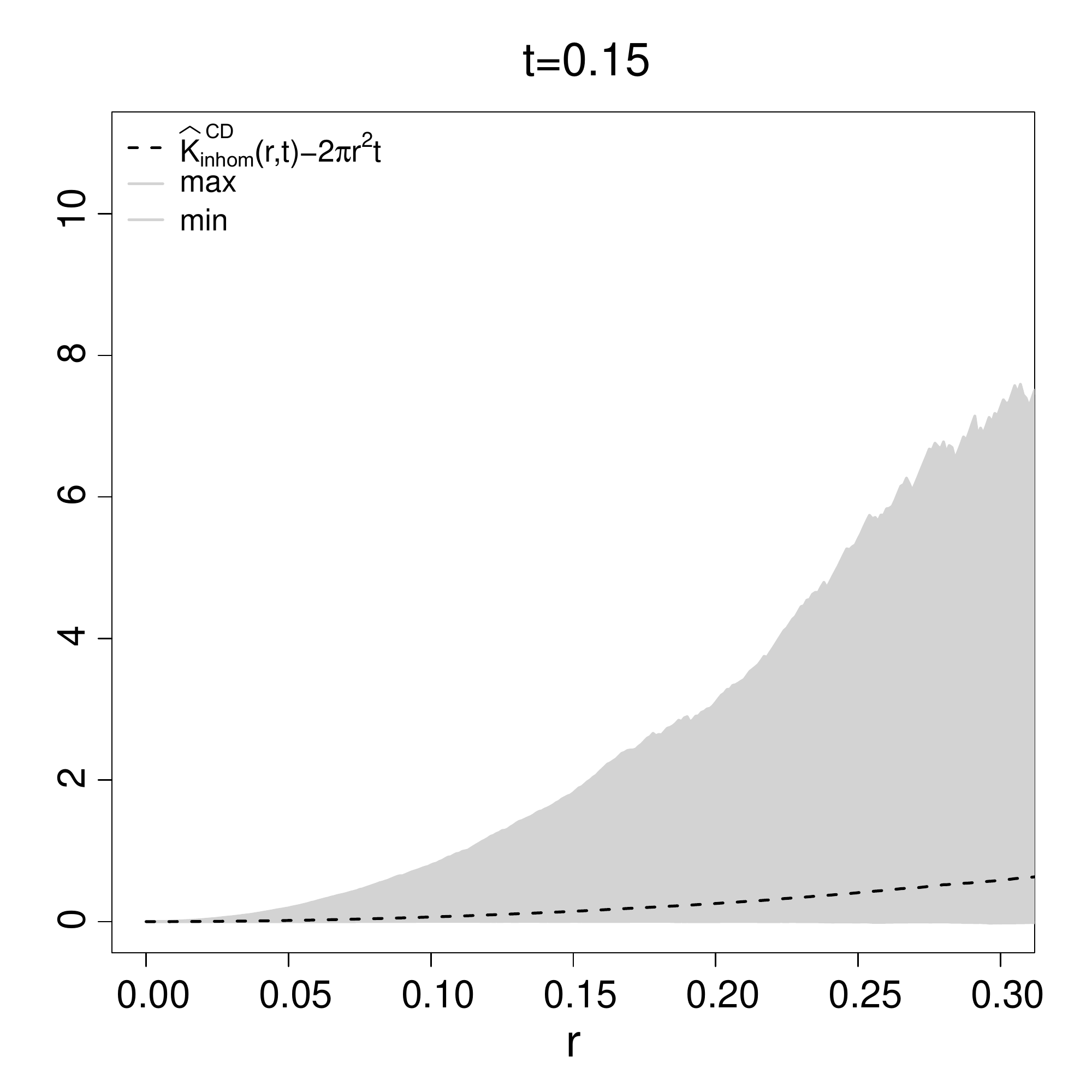}
  \includegraphics*[width=0.3\textwidth]{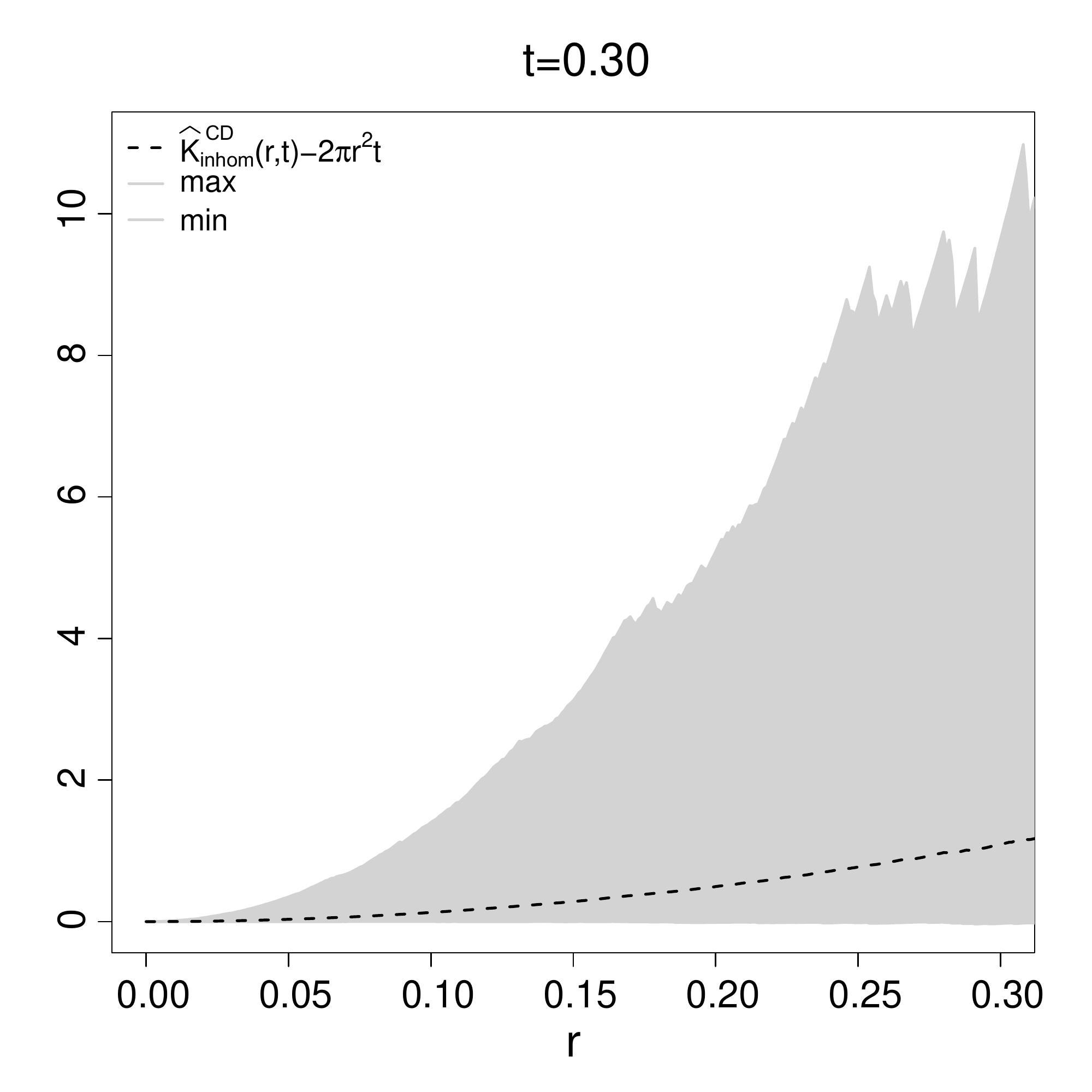}\\
  \includegraphics*[width=0.3\textwidth]{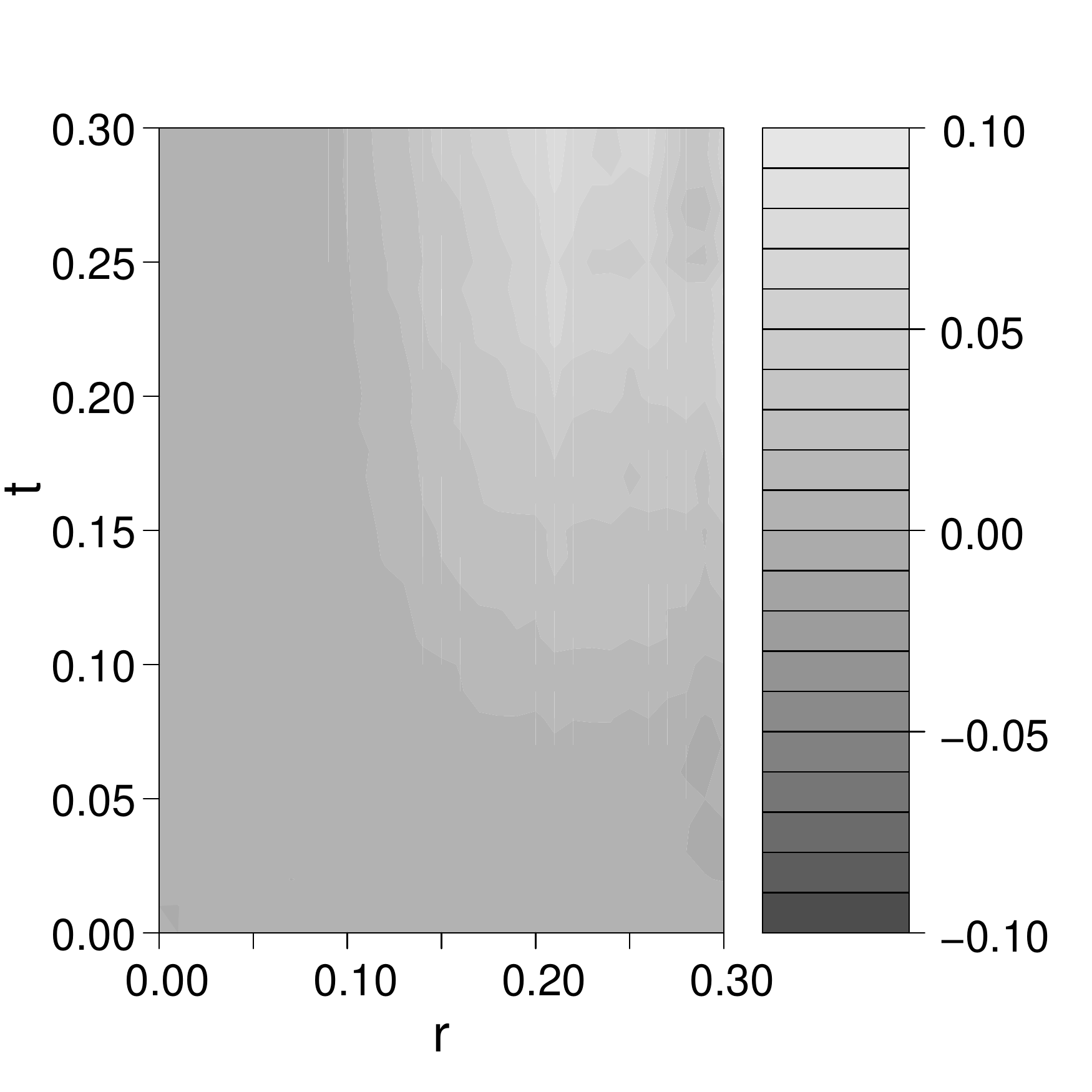}
 \caption{Envelopes for the estimated $K_{\mathrm{inhom}}^{CD}(r,t) - 2\pi r^2 t$, where $C=\{0\}$ and $D=\{1\}$, based on 99 realisations of the the model given in Example \ref{ExampleIndependentComp}, for fixed temporal lags $t=0.05$  and $t=0.10$ (upper row), $t=0.15$ and $t=0.30$ (middle row) and the estimate $\widehat K_{\mathrm{inhom}}^{CD}(r,t) - 2\pi r^2 t$ for all $r$ and $t$ (lower row).}
\label{pois_lgcp_different_t}
\end{figure}

%
%
%
%


\subsubsection{Testing random labelling}\label{RandomLabellingTest}
For a bivariate process $Y=(Y_1,Y_2)$, random labelling coincides with the concept of independent components, when $\nu(\cdot)$ is chosen as the counting measure. This is exploited in the stationary Lotwick-Silverman test \citep{LotwickSilverman} as well as in the inhomogeneous Lotwick-Silverman test \citep{cronie_marked}, which tests if $Y_1$ and $Y_2$ are randomly labelled.
We here offer an alternative idea to testing the hypothesis of random labelling in the context of general MSTPPs, which does not require a particular shape of the study region. Note that we merely indicate how such a test may be constructed and that we do not formally test hypotheses here. 
For Monte-Carlo tests (see e.g.\ \citep{diggle_book}) such as the one described here we note that there are issues related to the choice of the number of simulations used to construct envelopes (see e.g.\ \citep{Myllymaki}); unless executed properly, it is wise not to draw too strong conclusions and instead use them more loosely, 
as mere indicators of some hypothesis. 
Although $\widehat K_{\mathrm{inhom}}^{CD}(r,t)-\widehat K_{\mathrm{inhom}}^{Y_g}(r,t)$ gives us an indication on whether we have independent marking/random labelling, we cannot say exactly how large it has to be for us to infer anything. Hence, we need some formal way of testing such a hypothesis. 

To construct a test, with the hypotheses 
$H_0:$ the marks are randomly labelled, and $H_1:$ the marks are not randomly labelled, 
we recall from Theorem \ref{TheoremCommutativity} that a necessary condition for $H_0$ to hold is that $\mathcal K^{CD}(E) = \mathcal K^{DC}(E)$ for any $E\in\BB(\RdR)$ and any mark Borel sets $C,D$, with non-null $\nu$-content. 
Hence, as test statistic we will use
\[
\Delta(r,t)=K_{\mathrm{inhom}}^{CD}(r,t)-K_{\mathrm{inhom}}^{DC}(r,t), \quad r,t\geq0.
\]
This may be exploited to construct a Monte-Carlo test, where the envelopes are generated by resampling the marks of $Y$, {\em without replacement}, and for each such mark-permuted version of $Y$ estimate $\Delta(r,t)$. In essence, rejection of $H_0$ is based on whether the estimate of the original $\Delta(r,t)$, based on $Y$, sticks out of the envelopes for any $r,t\geq0$ and any $C,D$.

\begin{remark}
Note that the resampling of the marks requires that we assume that there is a common mark distribution, i.e.\ that we have random labelling. If one would have repeated observations of $Y$, on the other hand, one would also be able to test for independent marking. 

Furthermore, an alternative which we will not mention any further here is to consider, instead, resampling the marks {\em with} replacement. 

\end{remark}

We next evaluate the test above for a realisation of Example \ref{ExampleGeostatisticallyMarkedLGCP}. 
More explicitly, we estimate $\Delta(r,t)$ for the realisation found in Figure \ref{geost_lgcp} and then permute the marks in order to generate estimates $\Delta_i(r,t)$, $i=1,\ldots,99$, which in turn give rise to the envelopes. As we can see in Figure \ref{geost_lgcp_test_diff_t}, the estimate of $\Delta(r,t)$ for $t=0.20$ moves outside the envelopes, for certain values of $r$, which indeed would indicate that we do not have random labelling. 

\begin{figure}[!htbp]
\begin{center}
\begin{tabular}{cc}
  \includegraphics*[width=0.38\textwidth]{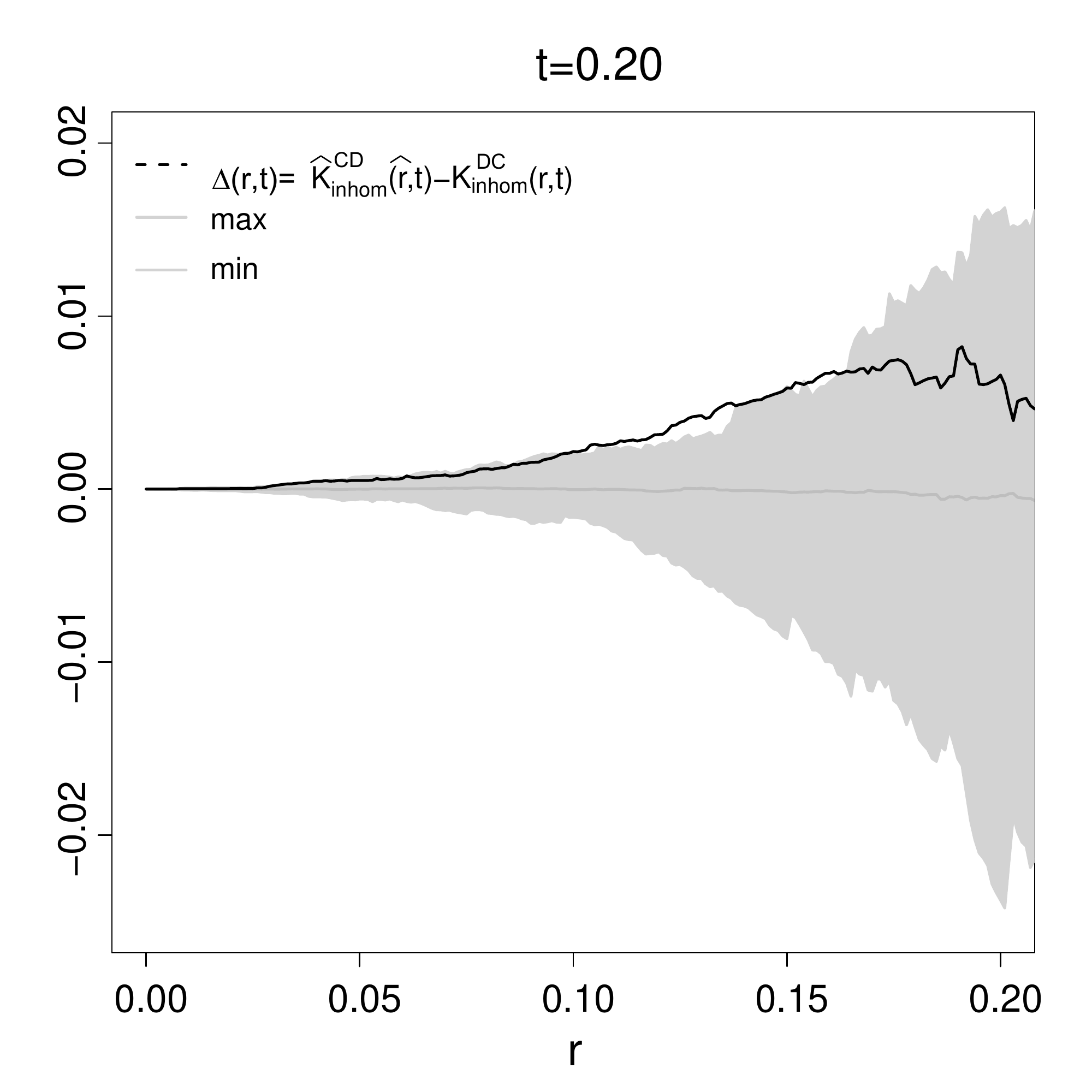}
\end{tabular}
 \caption{The estimated $\Delta(r,t)=K_{\mathrm{inhom}}^{CD}(r,t)-K_{\mathrm{inhom}}^{DC}(r,t)$, 
$Y_C=\{(x,t,m)\in[0,1]\times[0,1]\times[0,0.5]\}$, $Y_D=\{(x,t,m)\in[0,1]\times[0,1]\times(0.5,1]\}$, for the realisation in Figure \ref{geost_lgcp}, together with $\Delta(r,t)$-envelopes generated by 99 resamples/permutations of the marks, for fixed temporal lag $t=0.20$. 
}
\label{geost_lgcp_test_diff_t}
\end{center}
\end{figure}

Through Theorem \ref{TheoremCommutativity} and its proof we note that the stronger the (spatio-temporal) dependence between the marks, the more clear the deviation of $\Delta(r,t)$ from the envelopes. 
Note further that the larger the size of the sample, the better the performance of the test, as one would expect.

\section{Second order analysis of the earthquake data}\label{app_phuket}

As stated in Section \ref{SectionData}, earthquakes are a huge threat to mankind's safety. Large magnitude earthquakes have produced serious landscape damage, but also human casualties; recall the effects of the huge Sumatra-Andaman event from 2004 with magnitude $8.8$. The epicentre of the earthquake was located offshore, thus creating a huge tsunami which led to the tragedy where a large number of people died. In their paper, \cite{vingy2005} state that, after the Sumatra-Andaman earthquake, post-seismic motion was detected at more than $3,000$ km away, and as late as $50$ days after. This is an indication of a domino effect triggered by a big earthquake.

An {\em aftershock} is an earthquake following a previous large shock, the {\em main shock}. A major event tends to displace the crust of a tectonic plate, thus giving rise to the formation of aftershocks. The magnitude of an aftershock is smaller than the main shock. If the aftershock is larger than the main shock, the aftershock is labelled main shock and the original main quake is labelled {\em foreshock} \citep{USGS}. We want to study how far in space and time one may find aftershocks or foreshocks of different sizes. 

The data analysed in this section, which consists of a total of $n=1248$ earthquakes registered from 2004 to 2008, includes all earthquakes with magnitude larger than or equal to 5. The modified Mercalli intensity scale \citep{USGS} classifies earthquakes into twelve classes, where shocks with magnitude larger than $6$ can cause severe building and landscape damage, and human fatalities. Approximately $94.8\%$ of all earthquakes registered in the Sumatra area have magnitude $\leq 6$. These events cause minor wreckage, with limited damage to buildings and other structures. We want to study how far in space and time aftershocks and foreshocks (earthquakes with magnitude $\leq 6$) appear after a big shock (magnitude $> 6$). 

As previously mentioned, we have focused our analysis on developing point process tools which allow us to carry out second-order non-parametric analyses; recall that we consider the magnitudes as marks. Our objective is to analyse the interaction between different types of earthquakes, classified according to their magnitudes. More precisely, the $K$-function will give us information about the spatial and temporal scales at which points with marks (magnitude) in a certain category $C$, e.g.\ $C=\{\text{magnitude larger than 6}\}$, tend to cluster or tend to separate from points with marks in some other category $D$, e.g.\ $D=\{\text{magnitude less than or equal to 6}\}$, in the presence of inhomogeneity. 

Formally, we consider a marked spatio-temporal point pattern $Y=\{(x_i,y_i,t_i,m_i)\}_{i=1}^n\subseteq (W_S\times W_T)\times\M$, $n=1248$. Here $(x_i,y_i)\in\R^2$ is the spatial location of the $i$th event, $t_i\in \R$ is the number of days passed since the midnight of 1 January 2004 until the occurrence of the $i$th event, and $m_i$ is the associated magnitude. As explained in Section \ref{SectionData}, we transform the spatial latitude/longitude coordinates to UTM coordinates expressed in metres and rescale them. We use the following rescaling. Define $a=\min(y_i)$ and $b=\max(y_i)$. The new rescaled coordinates are $x'_i=(x_i-a)/(b-a)$ and $y'_i=(y_i-a)/(b-a)$, respectively, and $|b-a|=2295032$ metres. The spatial study region becomes $W_S=[0, 0.7]\times[0, 1]$. We also rescale time. If $c=\min(t_i)$ and $d=\max(t_i)$, then the rescaled temporal component is $t'_i=(t_i-c)/(d-c)$, where $d-c=1779.242$ days. The temporal window hereby becomes $W_T=[0,1]$. Theorem \ref{TheoremRescaling} in Section \ref{Section_Scaling} tells us that if we rescale the spatial and/or the temporal domain, the actual $K$-function estimates are obtained by simply scaling back the spatial and temporal lags. The largest earthquake ever to be recorded was in 1960, in Chile \citep{Kanamori}, with a magnitude of 9.5. Therefore, we set the magnitude scale to $[0,10]$. Hence, we consider earthquakes with magnitude greater than $6$ as belonging to mark set $C$, and shocks with magnitude less than or equal to $6$ to $D$. Furthermore, as reference measure for the mark space we use the Lebesgue measure on the mark space $\M=[0,10]$. 

Figure \ref{phuket_data_MSTPP} shows the marked spatio-temporal point pattern of all $1248$ earthquakes registered in the Sumatra area from 16 of February 2004 until 30 of December 2008. Here the marks are represented as circles, with the size being proportional to the magnitude of the event. 
We suspect that the pattern is not regular since there are points that tend to be close to other points at all scales (in other words, not just inhomogeneity), so there seems to be clustering. We can identify small shocks (small circles) gathering around big earthquakes (large circles), but we cannot visually conclude anything significant. There are some areas of the study region where isolated small events are observed. Figure \ref{phuket_data_year} (last plot) shows the temporal evolution of the earthquakes' magnitudes. This figure hints that, temporally, big earthquakes are preceded and followed by smaller foreshocks and aftershocks.

\begin{figure}[!htbp]
	\centering
	\includegraphics*[width=0.48\textwidth]{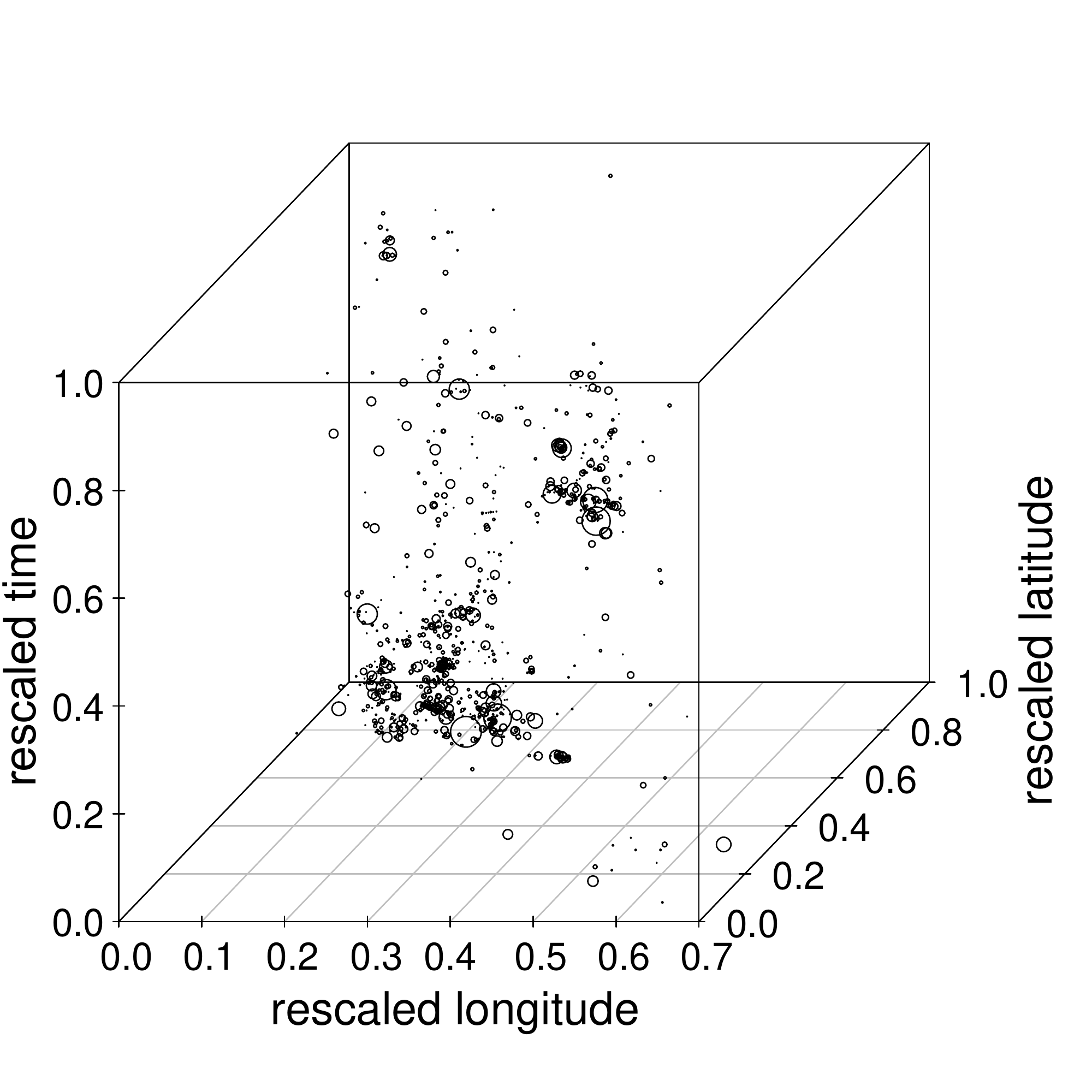}
	\caption{The marked spatio-temporal pattern of the earthquake dataset. The sizes of the circles are proportional to the magnitudes}\label{phuket_data_MSTPP}
\end{figure} 

Recall from Section \ref{SectionData} that we do not assume that there is first-order dependence between the spatial and the temporal components, i.e.\ we think it is justified to assume separability. In addition, as already mentioned above and indicated in Section \ref{SectionData}, we will assume that there is first-order dependence between the temporal component and the marks. This leads us to the intensity estimator $\widehat\lambda(x,t,m)$ given in expression \eqref{IntEstSep}. This estimator requires that we use the Voronoi tessellation $\V_{T\times M}$ in \eqref{PartialVoronoi}. Our numerical implementation of this max-metric tessellation turned out to be too slow for the analysis of this data set. As an approximation, we chose to replace $\V_{T\times M}$ in \eqref{IntEstSep} by the Euclidean Voronoi tessellation
\beann
\widetilde\V_{T\times M}
&=&
\{\widetilde \V_{(t,m)}^{T\times M}\}_{(t,m)\in\R\times\M}
\nonumber
\\
&=&
\big\{
(v,z)\in\R\times\M : 
\|(v,z)- (t,m)\|_{\R^2}
\leq \|(v,z)- (s,k)\|_{\R^2}
\\
&&\text{ for any } (s,k)\in Y_T\times Y_M\setminus\{(t,m)\}
\big\}_{(t,m)\in Y_T\times Y_M}
\eeann
and evaluated it numerically by means of the implementation 
found in the R package \verb|spatstat| \citep{spatstat}. 
We believe that this approximation generates intensity estimates of a similar kind (the difference will be particularly small when employing the smoothed $K^{CD}_{\rm inhom}(r,t)$ estimate). 

Figure \ref{K_function_phuket_data} (left) shows the estimate of $K^{CD}_{\rm inhom}(r,t)-2\pi r^2t$ for approximately a quarter of the spatio-temporal study region, which is the \verb|spatstat| default; spatial lags $r$ range between $0$ and $575$ km, and temporal lags $t$ range between $0$ and $445$ days. Figure \ref{K_function_phuket_data} (right) shows the smoothed $K$-function estimate (retention probability $p=0.5$ and $100$ bootstrap samples), for the same spatial and temporal lags $r$ and $t$. The behaviour does not change significantly for different choices of $p$. Figure \ref{K_zoom} shows the smoothed $K$-function for three different smaller temporal scales, chosen as day, week and $50$ days.

Figure \ref{K_function_phuket_data} indicates clustering, since the $K$-functions are larger than $2\pi r^2t$, at all spatio-temporal scales. This indicates that events in category $D$, meaning foreshocks or aftershocks, tend to cluster around events in category $C$. The strongest clustering between main shocks and foreshocks/aftershocks seems to occur at a temporal lag of approximately $200-300$ days, at all spatial scales. There seems to be an almost linear build-up of interaction and afterwards there seems to be a rapid decay in clustering. The majority of the fore-/aftershocks seem to occur at spatial distances larger than $200$ km from a main shock. We emphasise that aftershocks are observed at distances quite far from the main shock. Looking at Figure \ref{K_zoom}, there seem to be predominant inter-event distances at which most fore-/aftershocks tend to occur; note the peaks around $300$ and $500$ km. Figure \ref{K_zoom} (left) shows that within a day, aftershocks tend to travel far, even as far as $500$ km. Looking at the temporal lags in all three representations in Figure \ref{K_zoom} we can see that there are fore-/aftershocks occurring in direct connection to the main shock. We note that close in space and time there seem to be few fore-/aftershocks in connection to a main earthquake.

\begin{figure}[!htbp]
	\centering
	\includegraphics*[width=0.45\textwidth]{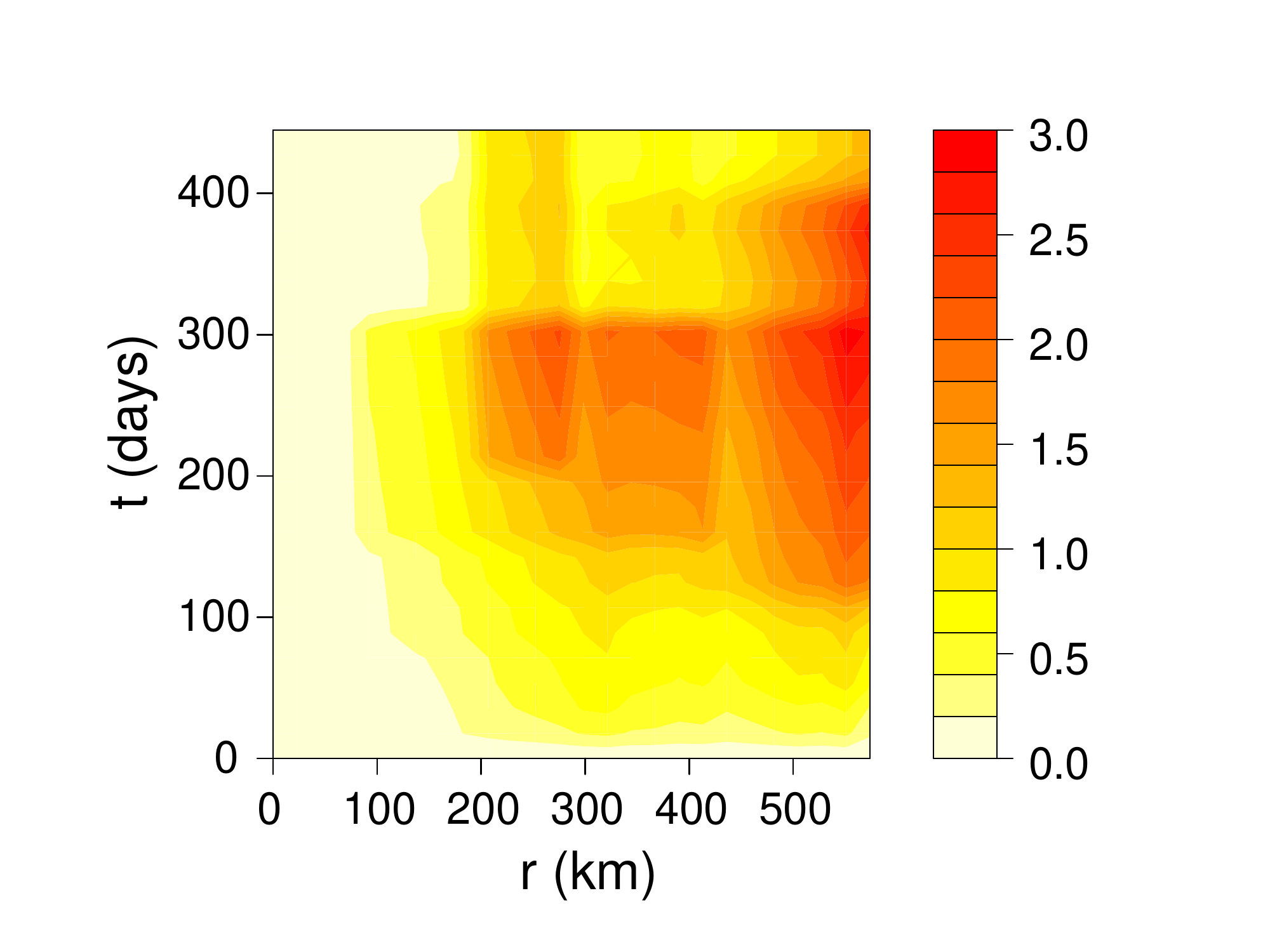}
	\includegraphics*[width=0.45\textwidth]{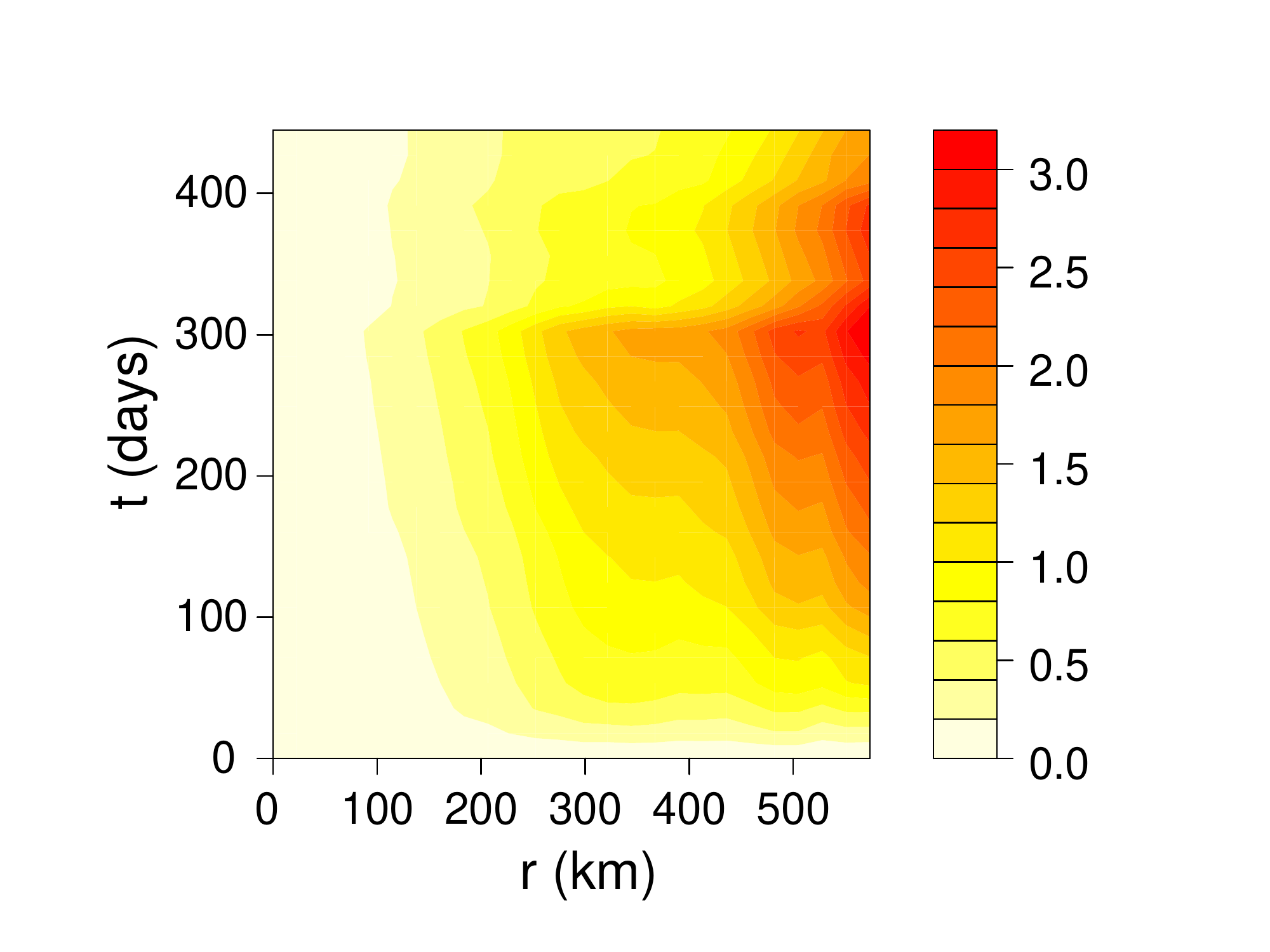}
	\caption{Estimated $K$-function for the Phuket data, $K^{CD}_{\rm inhom}(r,t)-2\pi r^2t$ (left). Smoothed $K$-function estimate, $\widetilde K^{CD}_{\rm inhom}(r,t)-2\pi r^2t$ (right).}\label{K_function_phuket_data}
\end{figure} 

\begin{figure}[!htbp]
	\centering
	\includegraphics*[width=0.3\textwidth]{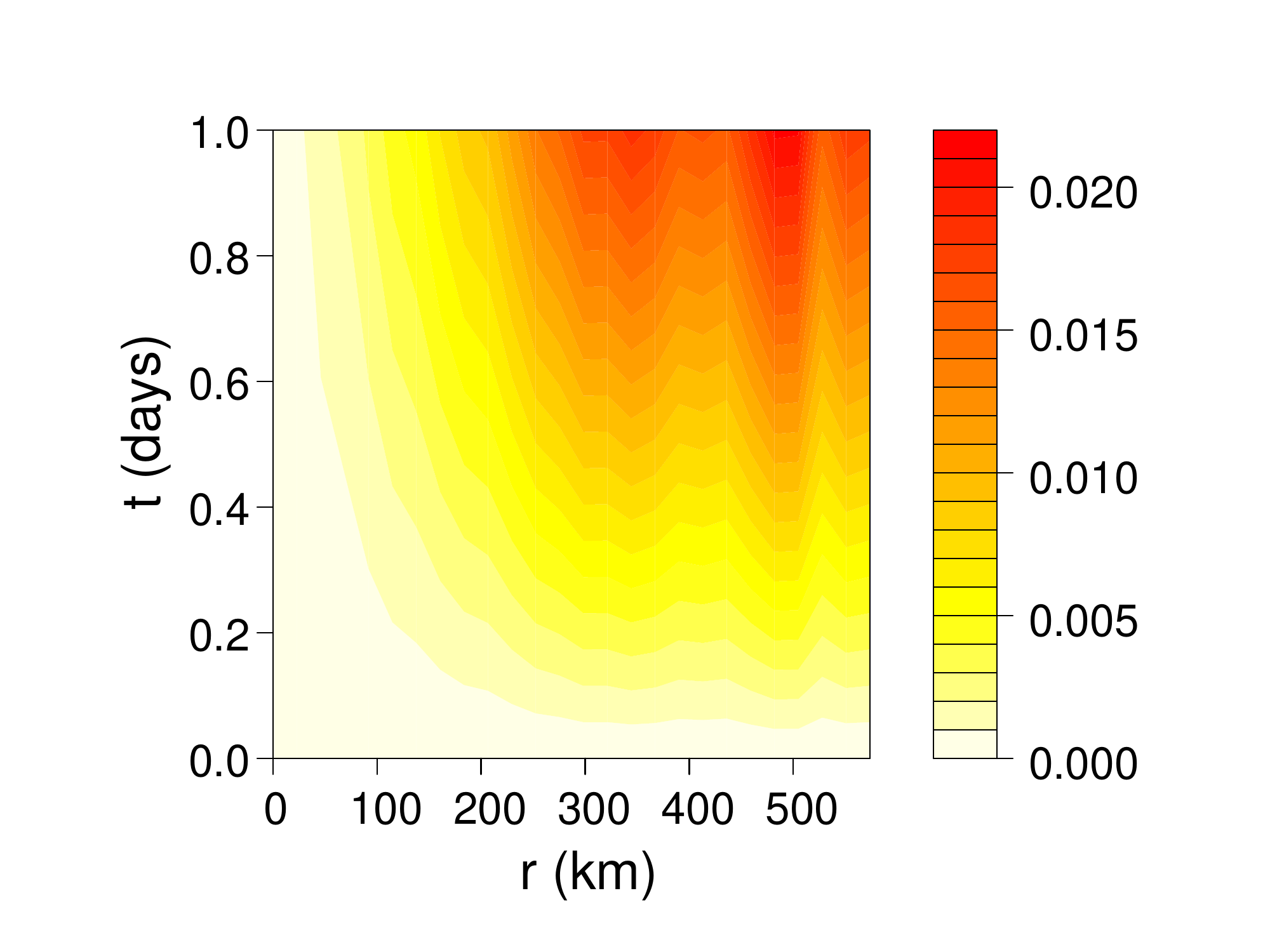}
	\includegraphics*[width=0.3\textwidth]{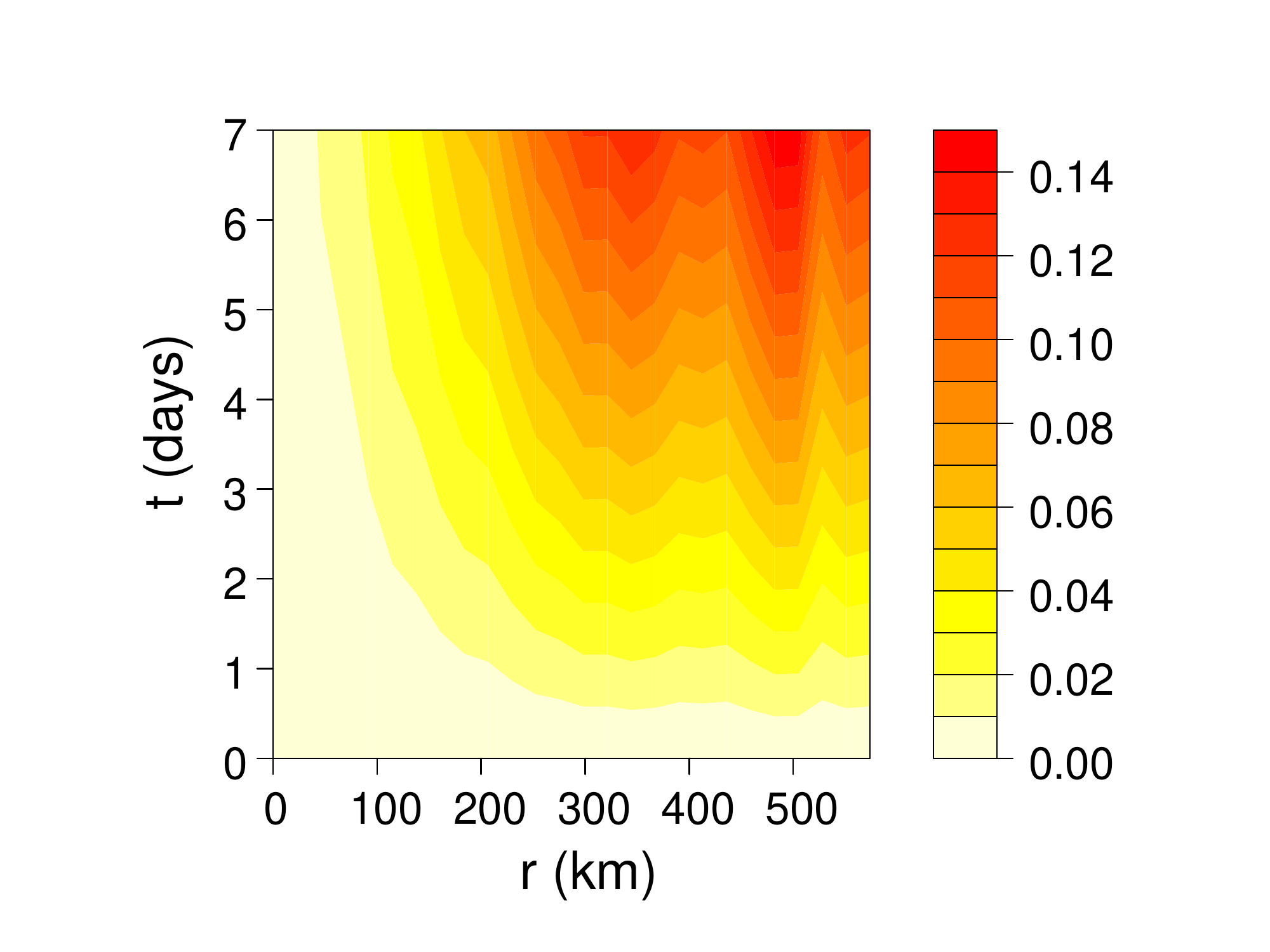}
	\includegraphics*[width=0.3\textwidth]{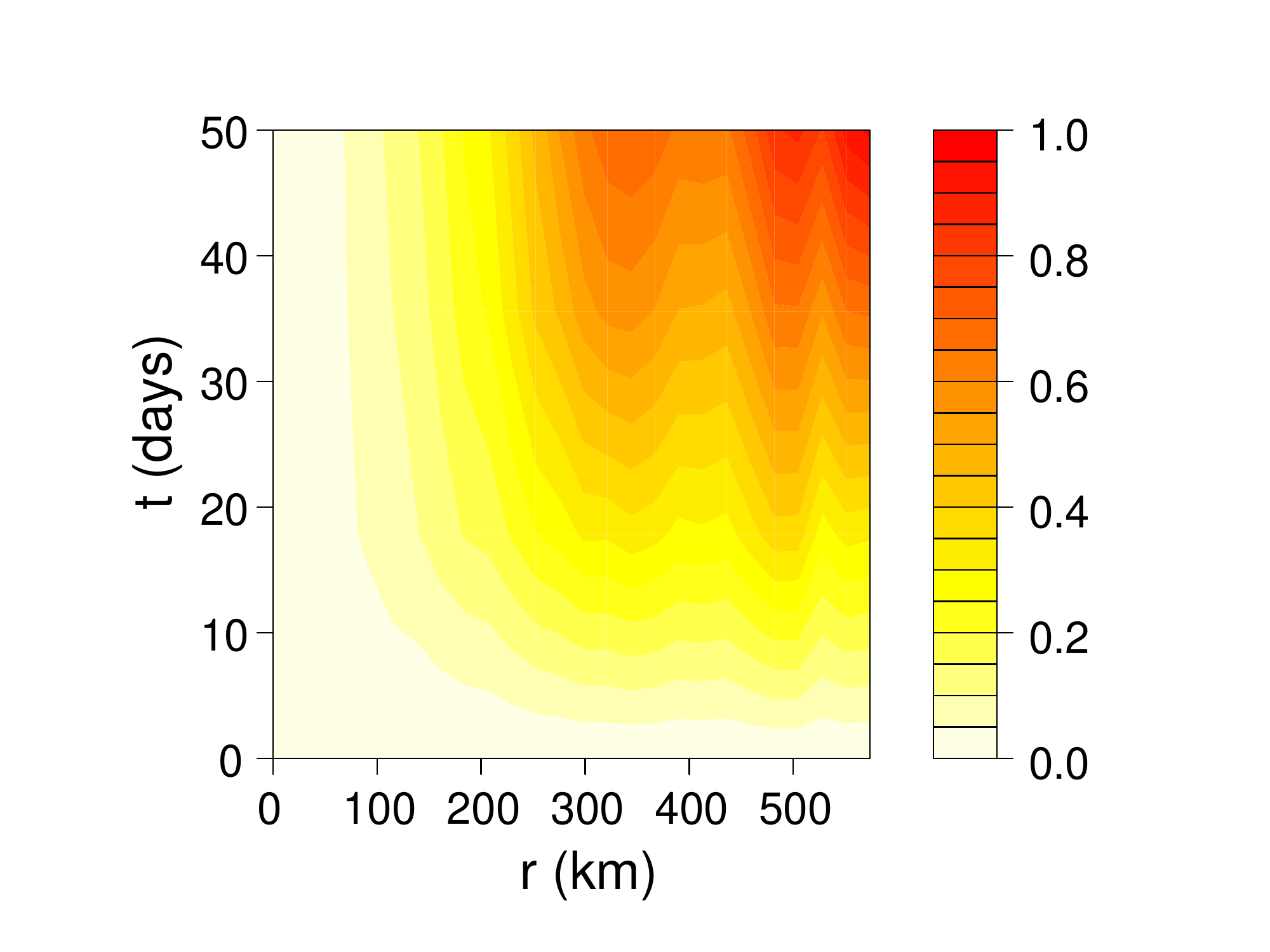}
	\caption{Smoothed $K$-function estimate, $\widetilde K^{CD}_{\rm inhom}(r,t)-2\pi r^2t$, for the time frames day (left), week (centre) and 50 days (right).}\label{K_zoom}
\end{figure}

In the literature it is sometimes considered that magnitude does not depend on the spatio-temporal location of the event \citep{USGS}. We next briefly look for indications of this belief by means of executing our random labelling test in Section \ref{RandomLabellingTest}, based on $99$ permutations of the marks, where we have used $95\%$ two-sided point-wise confidence bands. We found that for small and medium $t$ the estimate of $\Delta(r,t)$ stays within the envelopes for all considered spatial lags $r$. For very large $t$, as indicated in Figure \ref{TestData}, we see that the estimate of $\Delta(r,t)$ sticks out of the envelope, thus indicating the possibility of the marks not being randomly labelled. It is advised not to draw too strong conclusions, however, as indicated in Section \ref{RandomLabellingTest}.

\begin{figure}[!htbp]
	\centering
	\includegraphics*[width=0.45\textwidth]{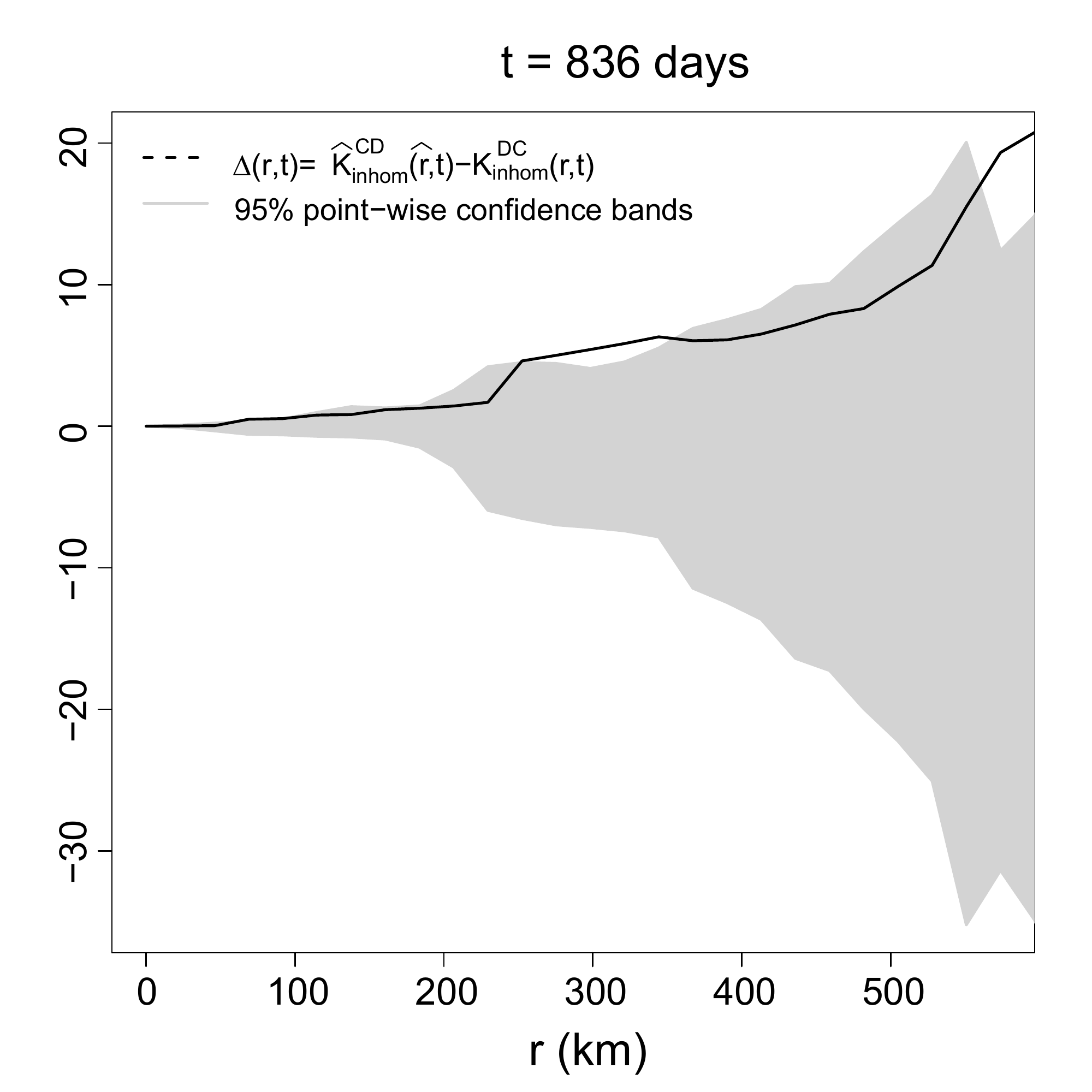}
	\caption{Estimate of $\Delta(r,t)=K_{\mathrm{inhom}}^{CD}(r,t)-K_{\mathrm{inhom}}^{DC}(r,t)$, for the earthquake dataset, together with $95\%$ two-sided point-wise Monte-Carlo confidence bands, for fixed temporal lag $t=836$ days.}\label{TestData}
\end{figure}

\section{Conclusion and discussion}\label{conlusions}

In this paper we have treated the second-order analysis of marked spatio-temporal point processes. In particular, we have defined measures of second-order spatio-temporal interaction, which allow us to quantify interactions between categories of marked points. For all statistics defined we derive unbiased estimators. In addition, we have considered an unbiased marked spatio-temporal Voronoi intensity estimation scheme, which allows us to estimate the underlying intensity function in an adaptive fashion. The set-up is quite general in the sense that the mark space as well as the corresponding mark reference measure are allowed to be arbitrary. We also exploit our newly defined tools to devise tests for particular marking structures. In the Appendix we have specialised our set-up to multivariate, directional and stationary analyses. 

The motivation behind this work comes from the necessity to analyse the interaction between main earthquakes and their fore-/aftershocks. We apply our methods to a well studied earthquake dataset \citep{ptprocess} and conclude that there are strong and far-reaching interactions between main shocks and other shocks. Also, we see some evidence that, given the spatio-temporal locations, the magnitudes are not behaving like an iid sequence of random variables (random labelling). 

Other direct applications of this methodology can be found in e.g.\ epidemiology and criminology. We are currently looking at datasets related to these fields. In particular, we are studying a dataset of chickenpox in the city of Valencia, Spain. Furthermore, we are analysing crime data in Valencia, Spain. Note that here it may be more relevant to consider multivariate versions of the summary statistics (see the Appendix). 


\cleardoublepage
\newpage

\bibliographystyle{chicago}
\bibliography{bibliografia}

\newpage

\appendix
\appendixpage

\section{Spatio-temporal point processes}\label{appendix_metric}

The most natural way of measuring distances in $\Rd$ is provided by the Euclidean metric 
$d_{\Rd}(x,y)=\|x-y\|_{\Rd}
$, $\|x\|_{\Rd}=(\sum_{i=1}^{d}x_i^2)^{1/2}$, $x,y\in \Rd$. 
Hence, we measure distances between spatial locations by means of $d_{\Rd}(\cdot,\cdot)$ and between temporal locations by means of $d_{\R}(\cdot,\cdot)$, i.e.\ absolute values. 
To combine the spatial 
and the temporal distances 
in a good way, such that we treat space and time differently, we endow our space-time domain $\RdR$ with the supremum norm $\|(x,t)\|_\infty=\max\{\|x\|_{\Rd},|t|\}$ and the supremum metric 
$$
d_{\infty}((x,t),(y,s))=\|(x,t)-(y,s)\|_\infty=\max\{d_{\Rd}(x,y),d_{\R}(t,s)\}
=\max\{\|x-y\|_{\Rd},|t-s|\},
$$ 
where $(x,t),(y,s) \in \RdR$. 
Hereby, we have combined two 
complete separable metric (csm) spaces, into the spatio-temporal csm space $(\RdR,d_{\infty}(\cdot,\cdot))$ \citep{daley03}.

Note that the $d_{\infty}$-induced Borel $\sigma$-algebra $\mathcal{B}(\RdR)=\mathcal{B}(\Rd) \otimes \mathcal{B}(\R)$, the product $\sigma$-algebra, since the underlying space is csm.
%
Following \cite{CronieSTPP}, we define a spatio-temporal point process as a simple point process in $(\RdR,\mathcal{B}(\RdR))$. 

\begin{definition}\label{def:space_time_pp}
A spatio-temporal point process with spatial locations in $\Rd$ and event times in $\R$ is a point process in $(\RdR,\mathcal{B}(\RdR))$. 
\end{definition}

\begin{remark}
If we would endow $\RdR=\R^{d+1}$ with the Euclidean distance $d_{\R^{d+1}}((x,t,(y,s))=\|(x,t)-(y,s)\|_{\R^{d+1}}=((t-s)^2+ \sum_{i=1}^{d}(x_i^2-y_i^2))^{1/2}$, we would encounter the problem that space and time are not treated differently. 
Indeed, this space is topologically equivalent to $(\RdR,d_{\mathbb{R}^{d+1}}(\cdot,\cdot))$ and we note that there are other (less natural) ways of combining $\|\cdot\|_{\Rd}$ and $|\cdot|$ such that $\RdR$ becomes a csm space.
\end{remark}

\section{The Hamilton principle}\label{appendix_Hamilton}

In estimators such as $\widehat{\mathcal{K}}^{CD}(E)$ and \eqref{eq:K_estimator}, 
\cite{StoyanStoyanRatio} advocate the Hamilton principle, which suggests replacing 
$\ell_d(W_S^{\ominus r})\ell_1(W_T^{\ominus t})\nu(C)$ by 
\[
\sum_{(x,s,m)\in Y\cap W_S^{\ominus r}\times W_T^{\ominus t}\times C} \frac{1}{\lambda(x,s,m)};
\]
the latter is an unbiased estimator of the former, due to the Campbell formula. 
In essence, we may have one of the following scenarios:

\begin{enumerate}

\item All of $\ell_d(W_S^{\ominus r})$, $\ell_1(W_T^{\ominus t})$, $\nu(C)$ and $\nu(D)$ are (assumed) known: employ \eqref{eq:K_estimator} for the estimation of $K_{\mathrm{inhom}}^{CD}(r,t)$. 

\item $\nu(C)$ and/or $\nu(D)$ is unknown but $\ell_d(W_S^{\ominus r}) \ell_1(W_T^{\ominus t})$ is known: use the estimator 
\beann
\widehat{\nu(C)} = \frac{1}{\ell_d(W_S^{\ominus r}) \ell_1(W_T^{\ominus t})}\sum_{(x,s,m)\in Y\cap W_S^{\ominus r}\times W_T^{\ominus t}\times C} \frac{1}{\lambda(x,s,m)}
\eeann
in \eqref{eq:K_estimator}. 
This is all analogous for $\nu(D)$. 

\item $\nu(C)$ and $\nu(D)$ are known explicitly but $\ell_d(W_S^{\ominus r}) \ell_1(W_T^{\ominus t})$ is unknown, with the ground intensity $\lambda_g(\cdot)$ (assumed) known explicitly: 
use the estimator
\beann
\widehat{\ell_d(W_S^{\ominus r}) \ell_1(W_T^{\ominus t})}
=
\sum_{(x,s)\in Y_g\cap W_S^{\ominus r}\times W_T^{\ominus t}} \frac{1}{\lambda_g(x,s)}
\eeann
in \eqref{eq:K_estimator}.

\item
Neither of $\ell_d(W_S^{\ominus r})$, $\ell_1(W_T^{\ominus t})$, $\nu(C)$ and $\nu(D)$ are (assumed) known but the ground intensity $\lambda_g(\cdot)$ is (assumed) known explicitly: estimate $\ell_d(W_S^{\ominus r})\ell_1(W_T^{\ominus t})\nu(C)\nu(D)$ by means of
\[
\frac{
\sum_{(x,s,m)\in Y\cap W_S^{\ominus r}\times W_T^{\ominus t}\times C} \lambda(x,s,m)^{-1}
\sum_{(x,s,m)\in Y\cap W_S^{\ominus r}\times W_T^{\ominus t}\times D} \lambda(x,s,m)^{-1}
}{
\sum_{(x,s)\in Y_g\cap W_S^{\ominus r}\times W_T^{\ominus t}} \lambda_g(x,s)^{-1}
}
\]
and plug this into \eqref{eq:K_estimator}. 

\end{enumerate}
Note that this, in fact, means that when we are given the intensity functions $\lambda(x,t,m)$ and $\lambda_g(x,t)$, we do not need to explicitly know/provide $\nu(C)$ and $\nu(D)$. 
This setup provides (ratio) unbiased estimators when the intensity is known. 
To evaluate the performance of the four scenarios above, we employ each one to 99 realisations of the model in Example \ref{ExamplePoisson} and generate min/max envelopes \citep{diggle_book}. The results can be found in Figure  \ref{FigureHamiltonScenarios} and it seems that knowing the mark set measures is the most crucial part. Note, however, that the most realistic practical scenario is number 2. 

\begin{figure}[!htbp]
\centering
  \includegraphics*[width=0.4\textwidth]{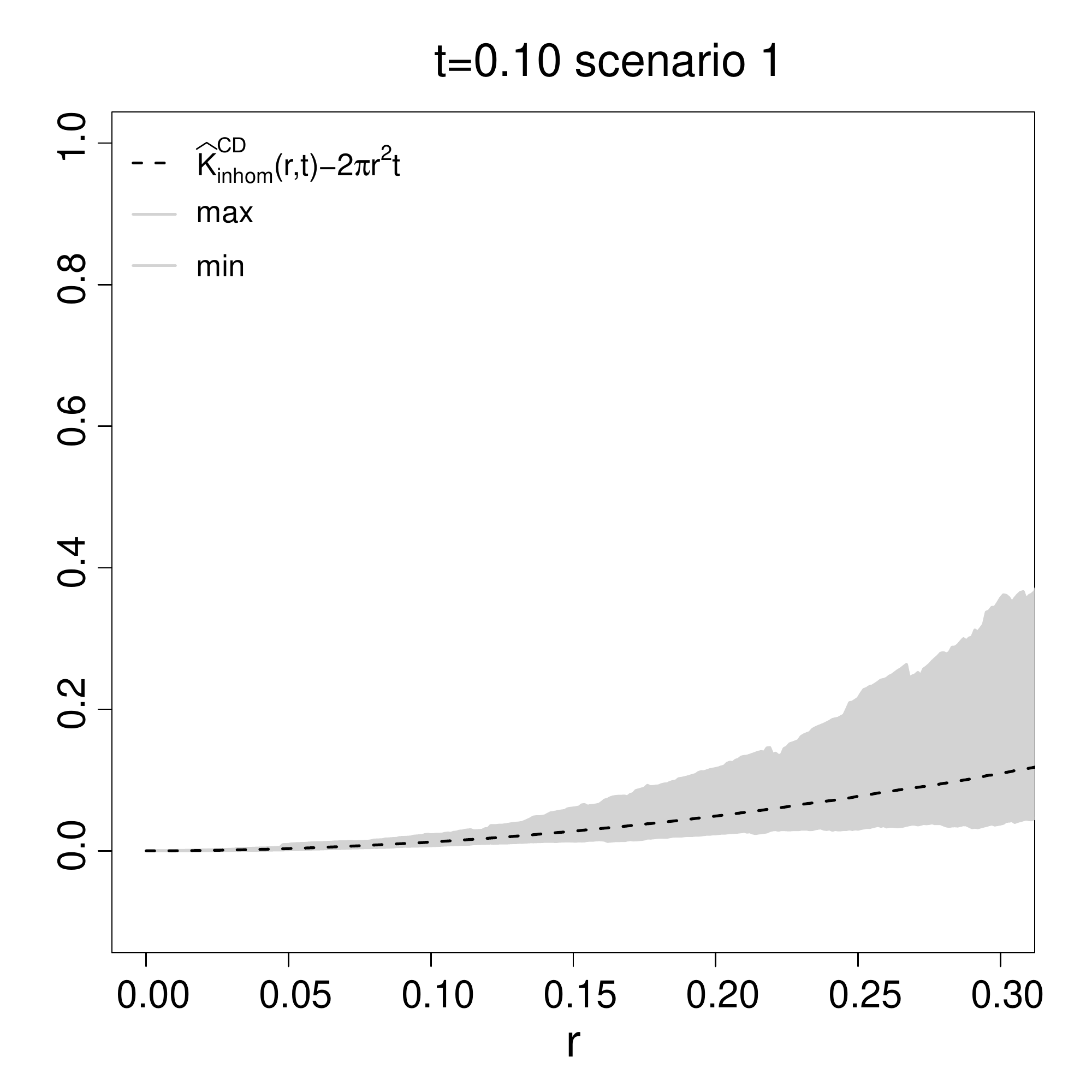}
  \includegraphics*[width=0.4\textwidth]{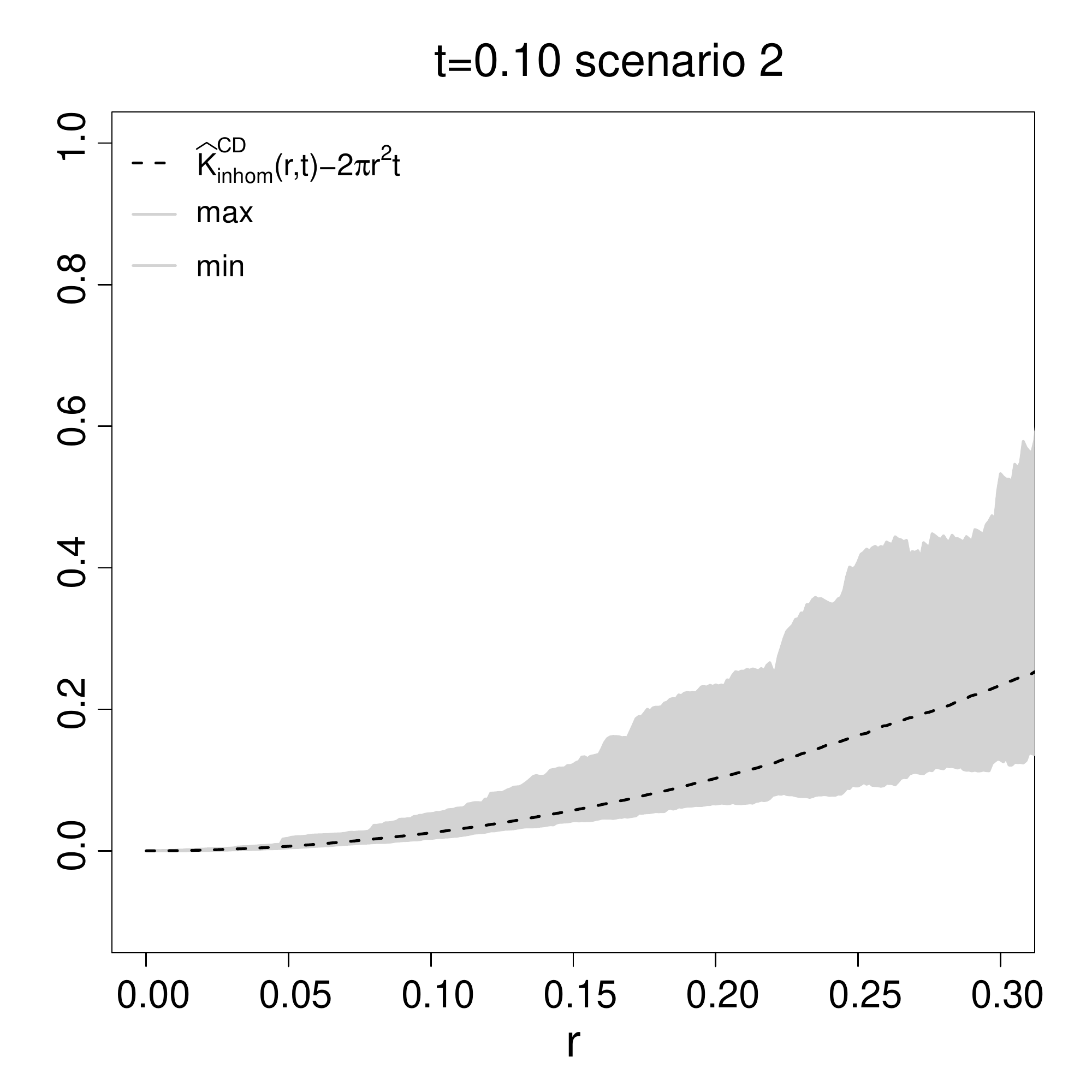}\\
  \includegraphics*[width=0.4\textwidth]{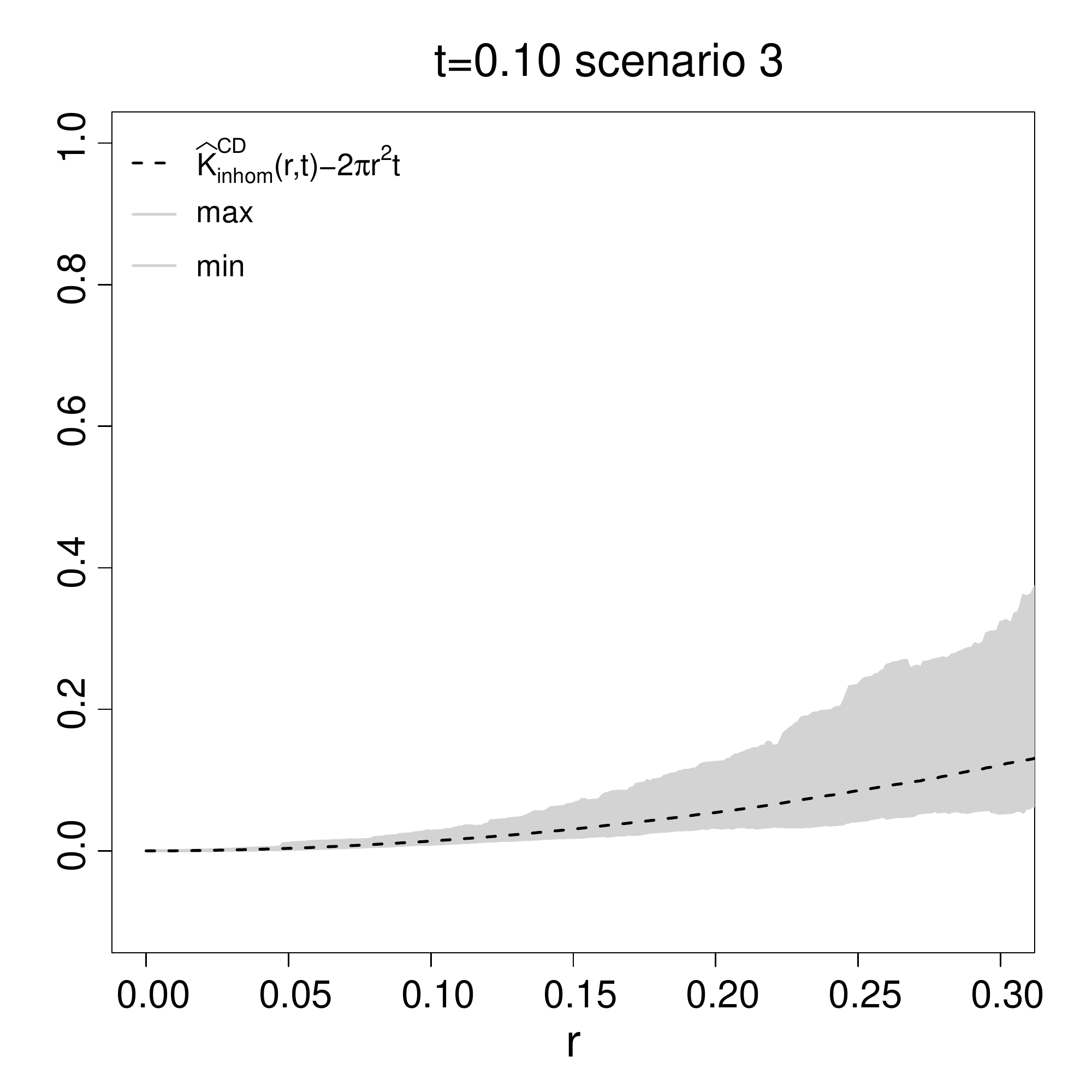}
  \includegraphics*[width=0.4\textwidth]{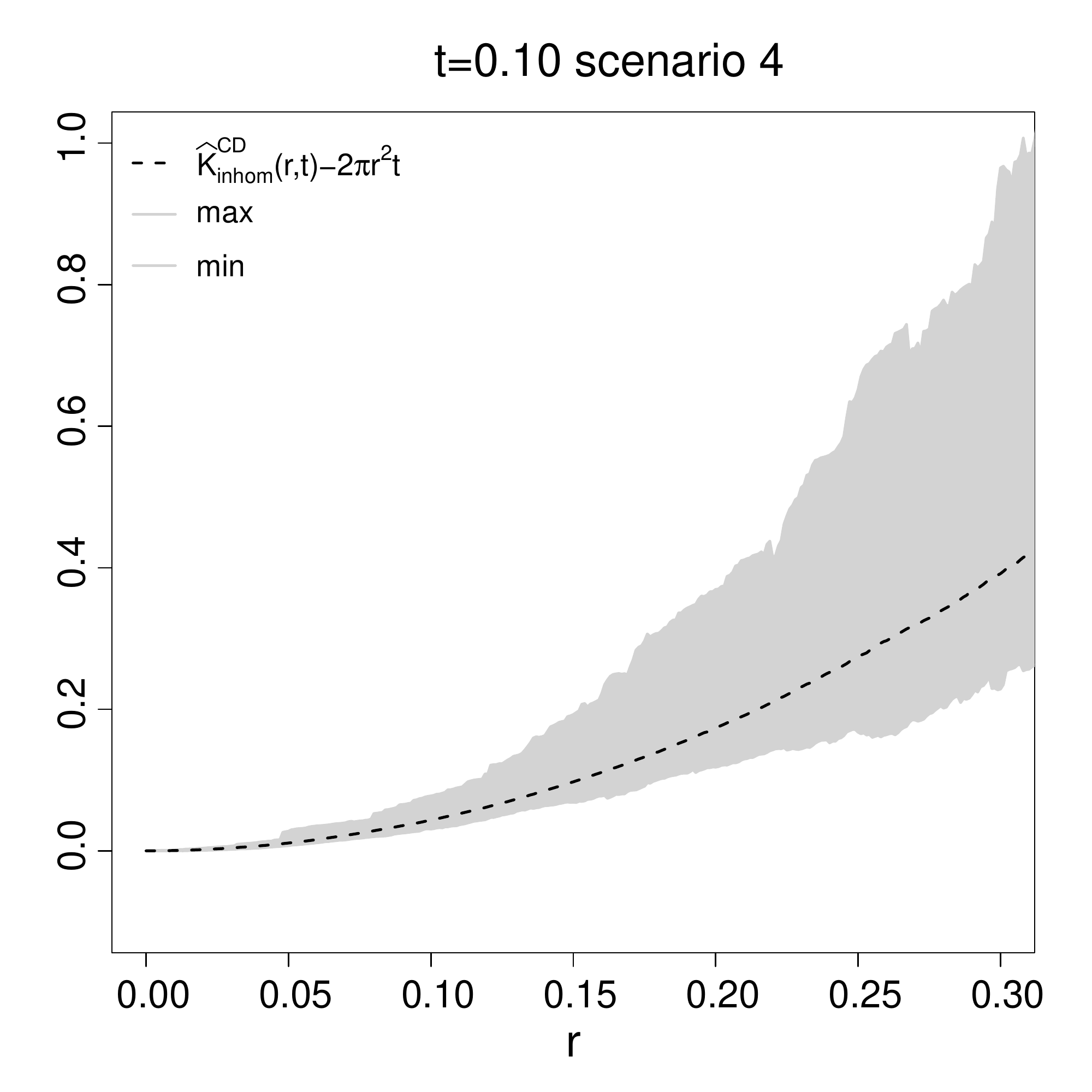}\\
 \caption{
The four Hamilton principle scenarios for the estimator \eqref{eq:K_estimator}; min-max envelopes based on 99 realisations of the randomly labelled Poisson process in Example \ref{ExamplePoisson}.}
\label{FigureHamiltonScenarios}
\end{figure}

\section{Special cases}\label{appendix_Multivariate_Stationary}

We next look closer at how the summary statistics and their estimators reduce under assumptions of $Y$ being multivariate/multi-type or stationary. We also briefly indicate a marked measure of (spatial) anisotropy for MSTPPs.

\subsection{Multivariate STPPs}

Starting with the multivariate case, where $\M$ is a finite collection of labels $\{1,\dots,k\}$, $k\geq 2$, we set $d'(m_1,m_2)=|m_1-m_2|$, $m_1,m_2 \in \M$, and employ the metric \citep[p.\ 8]{MCbook}
\beann
d((x_1,t_1,m_1),(x_2,t_2,m_2))
&=&
d_{\infty}((x_1,t_1),(x_2,t_2))+ d'(m_1,m_2)
\\
&=&\max\{\|x_1-x_2\|_{\Rd},|t_1-t_2|\}+|m_1-m_2|
.
\eeann
Note that here
\[
\int f(x,t,m) [\ell\otimes\nu](d(x,t,m)) 
= 
\int\int \sum_{i=1}^k f(x,t,i) \nu(i) \de t \de x
,
\] 
where $\nu(i)=\nu(\{i\})$; one usually lets $\nu(i)=1$ for all $i\in \M$, i.e.\ $\nu(\cdot)$ is chosen as the counting measure on $\M$.

The resulting MSTPP $Y=(Y_1,\dots,Y_k)$, a so-called \emph{multivariate/multi-type} STPP, has $i$th component STPP $Y_i=\{(x,t):(x,t,i)\in Y\}$, $i \in \M=\{1,\dots,k\}$. Hence, $Y_i$ contains all the points of $Y_g$ with $i$ as associated mark/type.  
For a multivariate STPP, the intensity satisfies $\lambda(x,t,i) = \lambda_i(x,t)/\nu(i)$, $i\in\M$, and 
\begin{align*}
\Lambda(B\times C)
=\E[Y(B\times C)]
=
\int_{B\times C} \lambda(x,t,m) \nu(dm)\de x \de t
=
\int_{B} \sum_{i\in\M} \1\{i\in C\} \lambda_i(x,t)\de x \de t
.
\end{align*}
In particular, $\lambda(x,t,i) = \lambda_i(x,t)$ when $\nu$ is the counting measure on $\M$. 


Following \cite{cronie_marked}, the $\nu$-averaged reduced Palm distribution at $(x,t)\in\RdR$, with respect to $C=\{i\}$, is given by
\[
P_i^{!(x,t)}(R)
=\frac{P^{!(x,t,i)}(R)\nu(i)}{\nu(i)}
=P^{!(x,t,i)}(R),
\quad
i\in\M=\{1,\ldots,k\}, R\in\mathcal{N},
\]
and is thus independent of the specific choice of $\nu(\cdot)$. 
By expression \eqref{ReducedMomentPalm} it now follows that
\begin{align*}
\mathcal K^{CD}(E)
&=
\sum_{i\in C}
\sum_{j\in D}
\frac{\nu(j)}{\nu(D)}
\frac{1}{\ell(B)}
\int_{B}
\E_i^{!(x_1,t_1)}
\left[
\sum_{(x_2,t_2)\in Y_j\cap((x_1,t_1)+E)}
\frac{1}{\lambda_j(x_2,t_2)}
\right]
\de x_1 \de t_1
\end{align*} 
for any $C,D\subseteq\M$, since $\lambda(x,t,i) = \lambda_i(x,t)/\nu(i)$. 
Here $\E_i^{!(x_1,t_1)}[\cdot]=\E^{!(x_1,t_1,i)}[\cdot]$ is the expectation under the reduced Palm distribution of $Y_i$. 
Note that we, in essence, scale each $j$-contribution by the probability $\nu(j)/\nu(D)$.

\begin{definition}
The {\em $i$-to-$j$ inhomogeneous spatio-temporal cross $K$-function} is given by
\begin{align*}
\label{Def:MultiK}
K_{\mathrm{inhom}}^{ij}(r,t) 
&=
\mathcal K^{\{i\}\{j\}}(\mathcal C_r^t(x_1,t_1))
\nn
\\
&= 
\frac{1}{\ell(B)}
\E\left[ 
\sum_{(x_1,t_1)\in Y_i\cap B}
\sum_{(x_2,t_2)\in Y_j\cap\mathcal{C}_r^t(x_1,t_1)\setminus\{(x_1,t_1)\}}
\frac{1}
{
\lambda_i(x_1,t_1)\lambda_j(x_2,t_2)}
\right]
\nn
\\
&=
\frac{1}{\ell(B)}
\int_{B}
\E_i^{!(x_1,t_1)}
\left[
\sum_{(x_2,t_2)\in Y_j\cap \mathcal C_r^t(x_1,t_1)}
\frac{1}{\lambda_j(x_2,t_2)}
\right]
\de x_1 \de t_1
,
\end{align*}
c.f.\ \cite[Definition 4.8]{moller}. 
\end{definition}

Note that when $i=j$, $K_{\mathrm{inhom}}^{ij}(r,t)$ reduces to the inhomogeneous spatio-temporal $K$-function \citep{GabrielDiggle} of $Y_i$, i.e.\ $K_{\rm inhom}^{i}(r,t)$. 
Also, the {\em $i$-to-any inhomogeneous spatio-temporal cross $K$-function} is given by
\[
K_{\mathrm{inhom}}^{i\bullet}(r,t) 
=
K_{\mathrm{inhom}}^{\{i\}\M}(r,t) 
= 
\sum_{j\in\M}
\frac{\nu(j)}{\nu(\M)}
\frac{1}{\ell(B)}
\int_{B}
\E_i^{!(x_1,t_1)}
\left[
\sum_{(x_2,t_2)\in Y_j\cap\mathcal C_r^t(x_1,t_1)}
\frac{\nu(j)}{\lambda_j(x_2,t_2)}
\right]
\de x_1 \de t_1
,
\]
where each $\nu(j)=1$, $j\in\M$, if $\nu(\cdot)$ is the counting measure on $\M$. 

\subsubsection{Estimation}

We next turn to the estimation of a multivariate SOIRS STPP $Y$. 
From the general estimator in Definition \ref{def:K_estimator}, where $C=\{i\}$ and $D=\{j\}$, $i\neq j$, we obtain
\begin{align*}
\widehat{K}_{\mathrm{inhom}}^{ij}(r,t)
&= 
\frac{1}{\ell_d(W_S^{\ominus r})\ell_1(W_T^{\ominus t})}
\sum_{(x_1,t_1)\in Y_i\cap W_S^{\ominus r}\times W_T^{\ominus t}}
\frac{1}{\lambda_i(x_1,t_1)}
\sum_{(x_2,t_2)\in Y_j\cap \mathcal{C}_r^t(x_1,t_1)} 
\frac{1}{\lambda_j(x_2,t_2)}
,
\end{align*}
and we see that this does not require explicit knowledge of $\nu(\cdot)$. Although not necessary here, it is common to assume that $\nu(\cdot)$ is the counting measure on $\M$. 

Since $\lambda(x,t,i) = \lambda_i(x,t)/\nu(i)$, in practise, for each $i\in\M=\{1,\ldots,k\}$ we obtain an estimate $\widehat\lambda_i(x,t)$ based on $Y_i$, which we plug into the estimators above.
One may e.g.\ use either of the separable or non-separable ground process Voronoi intensity estimators proposed previously.

\subsection{Marked stationary spatio-temporal K-functions}

When $Y$ is stationary with ground intensity $\lambda_g(x,t)\equiv\lambda>0$ and mark density $f_{(x,t)}^{\M}(m)=f^{\M}(m)$, recalling the reduced Palm distributions and \eqref{ReducedMomentPalm}, we have that (see e.g.\ \cite[Theorem C.1]{moller})
\begin{align*}
\mathcal K^{CD} (E)
&=
\frac{1}{\ell(B)\nu(C)\nu(D)}
\int_{B\times C}
\E^{!(0,0,m_1)}
\left[
\sum_{(x_2,t_2,m_2)\in Y\cap E\times D}
\frac{1}{f^{\M}(m_2)\lambda}
\right]
\de x_1 \de t_1 \nu(dm_1)
\\
&=
\frac{1}{\lambda\nu(C)\nu(D)}
\int_{C}
\E^{!(0,0,m_1)}
\left[
\sum_{(x_2,t_2,m_2)\in Y\cap E\times D}
\frac{1}{f^{\M}(m_2)}
\right]
\nu(dm_1)
\\
&=
\frac{\E_C^{!(0,0)}
\left[
\sum_{(x_2,t_2,m_2)\in Y\cap E\times D}
f^{\M}(m_2)^{-1}
\right]}{\lambda\nu(D)}
.
\end{align*} 
%
%
%
Hereby 
$
\mathcal K^{CD}(E) 
= \frac{1}{\lambda\nu(D)}\E_C^{!(0,0)}\left[Y(E\times D)\right]
$ if $\nu(\cdot)$ and $M(\cdot)$ coincide (or, equivalently, $f^{\M}(\cdot)\equiv1$).
This leads us to the definition of the $K$-function.



\begin{definition}
Given a stationary MSTPP $Y$ with intensity $\lambda>0$, under the assumption that $\nu(\cdot)=M(\cdot)$, its {\em marked stationary spatio-temporal $K$-function} is given by
\begin{align*}
K^{CD}(r,t) 
=\frac{1}{\lambda\nu(C)\nu(D)} \int_{C} \E^{!(0,0,m)} \left[Y(\mathcal{C}_r^t(0,0)\times D)\right] \nu(dm)
=
\frac{\E_C^{!(0,0)}
\left[
Y(\mathcal{C}_r^t(0,0)\times D)
\right]}{\lambda\nu(D)}
,
\end{align*}
for any $C,D\in\BB(\M)$ with $\nu(C),\nu(D)>0$. 
This is a spatio-temporal version of the form proposed by  \cite{VanLieshoutMPP}. 

In the multivariate case, where $\lambda_i(x,t)\equiv\lambda_i$>0, $i\in\M$, and $C=\{i\}$ and $D=\{j\}$, we obtain
$$
\label{eq:StationaryMultiK}
K^{ij}(r,t) =\dfrac{1}{\ell(B)\lambda_j}  \E^{!(0,0,i)}\left[Y_j(\mathcal C_r^t(0,0) \right]\int_B  dx_1dt_1
= \frac{\E^{!(0,0,i)}\left[Y_j(\mathcal C_r^t(0,0) \right]}{\lambda_j}
= \frac{\E_i^{!(0,0)}\left[Y_j(\mathcal C_r^t(0,0) \right]}{\lambda_j}
,
$$
a spatio-temporal version of the classical multivariate stationary $K$-function \citep[p.\ 60]{diggle_book}. 
In particular, $i=j$ results in the $K$-function of \cite{DiggleKstat} for $Y_i$. 

\end{definition}

In other words, given that there is a typical point of $Y$, located at the origin, with mark belonging to $C$, $K^{CD}(r,t)$ asks what the expected number of further points is, which are located within the cylinder $\mathcal{C}_r^t(0,0)$ and have marks belonging to $D$. 
In the multivariate case, assuming that $i\neq j$, $\lambda_j K_{\mathrm{inhom}}^{ij} (r,t)=\E^{!(0,0,i)}\left[Y_j(\mathcal C_r^t(0,0) \right]$ gives us the expected number of points of $Y_j$ that fall within spatial distance $r$ and temporal distance time $t$ of a typical point of $Y_i$. 
Note that  by the Slivniyak-Mecke theorem \citep{stoyan}, $\E_C^{!(0,0)}[Y(\mathcal{C}_r^t(0,0)\times D)] = \E[Y(\mathcal{C}_r^t(0,0)\times D)]=\lambda\nu(D)\ell(\mathcal{C}_r^t(0,0))$ for a Poisson process $Y$, as has already been established in the more general SOIRS case. Hence, $K^{CD}(r,t)-2t r^d\omega_d >0$ indicates clustering between points with marks in $C$ and $D$ and $K^{CD}(r,t)-2t r^d\omega_d <0$ indicates regularity.

%

\subsubsection{Estimation}

%
%


In the stationary case, when the reference measure is given by the mark distribution (see Definition \ref{DefCommonMark}), given the ground intensity $\lambda>0$, from the general estimator in Definition \ref{def:K_estimator} 
we obtain
\begin{align*}
\widehat{K}^{CD}(r,t)
&=
\frac{
\sum_{(x_1,t_1)\in Y_C\cap W_S^{\ominus r}\times W_T^{\ominus t}}
Y_D(\mathcal{C}_r^t(x_1,t_1)\setminus\{(x_1,t_1)\}) 
}{\lambda^2\ell_d(W_S^{\ominus r})\ell_1(W_T^{\ominus t})\nu(C)\nu(D)},
\end{align*}
where we in practise replace $\lambda$ by the estimate $\widehat\lambda=Y_g(W_S\times W_T)/(\ell_d(W_S)\ell_1(W_T))$ and $\nu(C)\nu(D)$ by
\(
\widehat{\nu(C)}\widehat{\nu(D)}=Y(W_S\times W_T\times C)Y(W_S\times W_T\times D)/Y_g(W_S\times W_T)^2
. 
\) 
%
%
In the stationary and multivariate case we obtain
\begin{align*}
\widehat{K}_{\mathrm{inhom}}^{ij}(r,t)
&= 
\frac{1}{\lambda_i\lambda_j\ell_d(W_S^{\ominus r})\ell_1(W_T^{\ominus t})}
\sum_{(x_1,t_1)\in Y_i\cap W_S^{\ominus r}\times W_T^{\ominus t}}
Y_j(\mathcal{C}_r^t(x_1,t_1)),
\end{align*}
where $\lambda_i$ is estimated by $\widehat\lambda_i=Y_i(W_S\times W_T)/(\ell_d(W_S)\ell_1(W_T))$, $i\in\M=\{1,\ldots,k\}$.

\subsection{Directional summary statistics}

We here just briefly touch upon how directional effects may be incorporated into the analysis. 
Recalling the marked spatio-temporal second-order reduced moment measure $\mathcal K^{CD}(E)$ of a SOIRS MSTPP $Y$, and the freedom of specifying $E$ as any Borel set in $\RdR$, following e.g.\ \citep[Section 4.2.2]{moller} we may define a directional marked inhomogeneous $K$-function: 
\begin{align*}
&\ell(B)\nu(C)\nu(D)K_{\rm inhom}^{CD}(r,t;\phi,\psi)
=
\\
&=
\E\left[ 
\sum_{(x_1,t_1,m_1)\in Y\cap B\times C}
\sum_{(x_2,t_2,m_2)\in Y}
\frac{
\1\{(x_2,t_2,m_2)\in\mathcal{C}(x_1,t_1,\phi,\psi,r,t)\times D\}
}{\lambda(x_1,t_1,m_1)\lambda(x_2,t_2,m_2)} 
\right]
,
\end{align*} 
where $\phi\in[-\pi/2,\pi/2)$, $\psi\in(\phi,\phi+\pi]$ and
\begin{align*}
&\mathcal{C}(x_1,t_1,\phi,\psi,r,t)=\\
&=\{x_1+a(\cos v, \sin v) : a\in[0,r], v\in[\phi,\psi] \text{ or } v\in[\pi+\phi,\pi+\psi]\}\times[t_1-t,t_1+t]
.
\end{align*}
This structure can in turn be used to treat the directional multivariate and/or stationary case. The estimation is obtained through Definition \ref{def:K_estimator} by setting $E=\mathcal{C}(x_1,t_1,\phi,\psi,r,t)-(x_1,t_1)$ in the estimator for $\mathcal K^{CD}(E)$. Note that it is unbiased by Lemma \ref{lemma_estimate}.

\section{Proofs}\label{appendix_proofs}

\begin{proof}[Proof of Theorem \ref{TheoremCommutativity}]


Through \eqref{eq:K_measure_def_g} we see that $\mathcal K^{CD}(\cdot) = \mathcal K^{DC}(\cdot)$ requires that 
$$
\frac{f_{(x_1,t_1),(x_2,t_2)}^{\M}(m_1,m_2)}
{f^{\M}_{(x_1,t_1)}(m_1) f_{(x_2,t_2)}^{\M}(m_2)}
\stackrel{a.e.}{=} 
\frac{f_{(x_1,t_1),(x_2,t_2)}^{\M}(m_2,m_1)}
{f^{\M}_{(x_1,t_1)}(m_2) f_{(x_2,t_2)}^{\M}(m_1)}
.
$$ 
If $Y$ is independently marked this is clearly satisfied. 
Turning to the second option, the common mark distribution translates the above statement into 
$$
f_{(x_1,t_1),(x_2,t_2)}^{\M}(m_1,m_2)
\stackrel{a.e.}{=} 
f_{(x_1,t_1),(x_2,t_2)}^{\M}(m_2,m_1)
,
$$ 
which holds by the exchangeability. 


\end{proof}

\begin{proof}[Proof of Theorem \ref{TheoremRescaling}]
Denote the pcf of $Y$ by $g_Y(\cdot)$. 
As in \cite[Section 4.3.]{CronieSTPP}, through a change of variables and the Campbell formula we find that the pcf of $\beta Y$ is given by
\begin{align*}
g_{\beta Y}((x_1,t_1,m_1),(x_2,t_2,m_2)) 
&= 
\frac{(\beta_S^d \beta_T)^{-2}\rho^{(2)}
((x_1/\beta_S,t_1/\beta_T,m_1),(x_2/\beta_S,t_2/\beta_T,m_2))}
{(\beta_S^d \beta_T)^{-1}\lambda(x_1/\beta_S,t_1/\beta_T,m_1)
(\beta_S^d \beta_T)^{-1}
\lambda(x_2/\beta_S,t_2/\beta_T,m_2)}
\\
&= 
g_{Y}((x_1/\beta_S,t_1/\beta_T,m_1),(x_2/\beta_S,t_2/\beta_T,m_2))
.
\end{align*}
Hence, by a change of variables, 
\begin{align*}
&K_{\rm inhom}^{CD}(r, t;\beta)
=\\
&=
\frac{1}{\nu(C)\nu(D)}
\int_C\int_D 
\int_{\|x\|\leq r} \int_{|s|\leq t} g_{Y}((0,0,m_1),(x/\beta_S,s/\beta_T,m_2))\de x \de s 
\nu(dm_2)\nu(dm_1)
\\
&=
\frac{1}{\nu(C)\nu(D)}
\int_C\int_D 
\int_{\|\beta_S x\|\leq r} \int_{|\beta_T s|\leq t} g((0,0,m_1),(x,s,m_2))\de x \de s 
\nu(dm_2)\nu(dm_1)
\\
&=
K_{\rm inhom}^{CD}(r/\beta_S, t/\beta_T)
.
\end{align*}

\end{proof}

\begin{proof}[Proof of Theorem \ref{TheoremVoronoi}]
We only consider the marked spatio-temporal case since the other one is analogous. Starting with the mass-preservation, we have that
\begin{align*}
&\int_{W_S\times W_T\times\M}\widehat\lambda(x,t,m)\nu(dm)\de x\de t
=\\
&=\sum_{(x_i,t_i,m_i)\in Y\cap W_S\times W_T\times\M}
\frac{
\int_{W_S\times W_T\times\M}
\1\{(x,t,m)\in\V_{(x_i,t_i,m_i)}\}
\nu(dm)\de x\de t
}{[\ell\otimes\nu](\V_{(x_i,t_i,m_i)}\cap W_S\times W_T\times\M)}
\\
&
=Y(W_S\times W_T\times\M)
.
\end{align*}
Taking expectations on both sides and applying Fubini's theorem, 
\begin{align*}
\int_{W_S\times W_T\times\M}\E[\widehat\lambda(x,t,m)]\nu(dm)\de x\de t
&=\E[Y(W_S\times W_T\times\M)]
\\
&=\int_{W_S\times W_T\times\M}\lambda(x,t,m)\nu(dm)\de x\de t
,
\end{align*}
which implies that $\int_{W_S\times W_T\times\M}|\E[\widehat\lambda(x,t,m)]-\lambda(x,t,m)|\nu(dm)\de x\de t=0$. This, in turn, implies that $|\E[\widehat\lambda(x,t,m)]-\lambda(x,t,m)|=0$ a.e.\ on $W_S\times W_T\times\M$. 

\end{proof}

\begin{proof}[Proof of Lemma \ref{lemma_estimate}]

By the Campbell formula and expression \eqref{eq:K_measure_def_g}, 
\begin{align*}
\E[\widehat{\mathcal{K}}^{CD}(E)] 
&=
\frac{
\int_{W_S^{\ominus r}\times W_T^{\ominus t}\times C} 
\int_{E\times D} 
g((x_1,t_1,m_1),(x_2,t_2,m_2))
\de x_1 \de t_1\nu(dm_2)
\de x_2 \de t_2 \nu(dm_1)
}{\ell(W_S^{\ominus r})\ell(W_T^{\ominus t})\nu(C)\nu(D)} 
\\
& 
= 
\frac{
\ell\left(W_S^{\ominus r}\right)
\ell\left(W_T^{\ominus t}\right) 
\int_C \int_{E\times D} 
g((0,0,m_1),(u,v,m_2))
\de u \de v \nu(dm_2)\nu(dm_1)
}{\ell(W_S^{\ominus r})\ell(W_T^{\ominus t})\nu(C)\nu(D)} 
\\
& = \mathcal{K}^{CD}(E)
,
\end{align*}
which implies that \eqref{eq:K_estimator} is unbiased. 
We next turn to the variance and for simplicity we write $A=W_S^{\ominus r}\times W_T^{\ominus t}$. It follows that 
\begin{align*}
&[\ell(A)\nu(C)\nu(D)]^2 \widehat{\mathcal{K}}^{CD}(E)^2
=\\
&=
\sum_{(x_1,t_1,m_1),(x_2,t_2,m_2),(x_3,t_3,m_3),(x_4,t_4,m_4)\in Y}
\frac{
\1\{(x_1,t_1,m_1)\in A\times C\}
\1\{(x_3,t_3,m_3)\in A\times C\}
}{\lambda(x_1,t_1,m_1)\lambda(x_2,t_2,m_2)\lambda(x_3,t_3,m_3)\lambda(x_4,t_4,m_4)}
\times
\\
&
\times
\1\{(x_2,t_2,m_2)\in E\times D\setminus\{(x_1,t_1,m_1)\}\}
\1\{(x_4,t_4,m_4)\in E\times D\setminus\{(x_3,t_3,m_3)\}\}
\\
&=
\sum_{(x_1,t_1,m_1)\in Y\cap A\times C} 
\sum_{(x_2,t_2,m_2)\in Y\cap E\times D\setminus\{(x_1,t_1,m_1)\}}
\frac{1}{\lambda(x_1,t_1,m_1)^2\lambda(x_2,t_2,m_2)^2}
\\
&+
\sum_{(x_1,t_1,m_1)\in Y\cap A\times C} 
\mathop{\sum\nolimits\sp{\ne}}_{(x_2,t_2,m_2),(x_4,t_4,m_4)\in Y\cap E\times D\setminus\{(x_1,t_1,m_1)\}}
\frac{\lambda(x_1,t_1,m_1)^{-2}}{\lambda(x_2,t_2,m_2)\lambda(x_4,t_4,m_4)}
\\
&+
\mathop{\sum\nolimits\sp{\ne}}_{(x_1,t_1,m_1),(x_2,t_2,m_2)\in Y\cap A\times C}
\sum_{(x_3,t_3,m_3)\in Y\cap E\times D\setminus\{(x_1,t_1,m_1),(x_2,t_2,m_2)\}}
\frac{\lambda(x_3,t_3,m_3)^{-2}}{\lambda(x_1,t_1,m_1)\lambda(x_2,t_2,m_2)}
\\
&+
\mathop{\sum\nolimits\sp{\ne}}_{(x_1,t_1,m_1),(x_2,t_2,m_2)\in Y\cap A\times C}
\mathop{\sum\nolimits\sp{\ne}}_{(x_3,t_3,m_3),(x_4,t_4,m_4)\in Y\cap E\times D\setminus\{(x_1,t_1,m_1),(x_2,t_2,m_2)\}}
\frac{1}{\prod_{i=1}^4\lambda(x_i,t_i,m_i)}
\\
&= S_1+S_2+S_3+S_4
.
\end{align*}
By the Campbell formula, 
\begin{align*}
&\E[S_4]
= \int_{A\times C}\int_{A\times C}
\int_{E\times D}\int_{E\times D}
\frac{\rho^{(4)}((x_1,t_1,m_1),\ldots,(x_4,t_4,m_4))}
{\lambda(x_1,t_1,m_1)\cdots\lambda(x_4,t_4,m_4)}
\prod_{i=1}^4\de x_i \de t_i \nu(dm_i)
,
\end{align*}

\begin{align*}
&\E[S_3]
= \int_{A\times C}\int_{A\times C}
\int_{E\times D}
\frac{1}{\lambda(x_3,t_3,m_3)}
\frac{\rho^{(3)}((x_1,t_1,m_1),\ldots,(x_3,t_3,m_3))}
{\lambda(x_1,t_1,m_1)\cdots\lambda(x_3,t_3,m_3)}
\prod_{i=1}^3\de x_i \de t_i \nu(dm_i)
,
\end{align*}

\begin{align*}
&\E[S_2]
= \int_{A\times C}\int_{E\times D}
\int_{E\times D}
\frac{1}{\lambda(x_1,t_1,m_1)}
\frac{\rho^{(3)}((x_1,t_1,m_1),\ldots,(x_3,t_3,m_3))}
{\lambda(x_1,t_1,m_1)\cdots\lambda(x_3,t_3,m_3)}
\prod_{i=1}^3\de x_i \de t_i \nu(dm_i)
,
\end{align*}

\begin{align*}
&\E[S_1]
= \int_{A\times C}
\int_{E\times D}
\frac{1}{\lambda(x_1,t_1,m_1)\lambda(x_2,t_2,m_2)}
g((x_1,t_1,m_1),(x_2,t_2,m_2))
\prod_{i=1}^2\de x_i \de t_i \nu(dm_i)
,
\end{align*}
whereby 
\bea
\label{eq:Variance_K_measure}
\Var(\widehat{\mathcal{K}}^{CD}(E))
=
\frac{\sum_{i=1}^4\E[S_i]}{[\ell(A)\nu(C)\nu(D)]^2} - \mathcal{K}^{CD}(E)^2.
\eea


\end{proof}

\begin{proof}[Proof of Lemma \ref{LemmaIndependentMarks}]
Recall that under the assumption of independent marks we have that $f_{(x_1,t_1),\dots,(x_n,t_n)}^{\M}(m_1,\dots,m_n)=\prod_{i=1}^{n} f_{(x_i,t_i)}^{\M}(m_i)$. Using Equation \eqref{eq:relation_g_gg}, we obtain that 

\begin{align*}
g((x_1,t_1,m_1),(x_2,t_2,m_2)) 
& = 
\frac{f_{(x_1,t_1),(x_2,t_2)}^{\M}(m_1,m_2)}
{f_{(x_1,t_1)}^{\M}(m_1) f_{(x_2,t_2)}^{\M}(m_2)}
g_g((x_1,t_1),(x_2,t_2))
= g_g((x_1,t_1),(x_2,t_2)),
\end{align*}
whereby $Y_g$ is SOIRS whenever $Y$ is and 
\begin{align*}
K^{CD}_{\rm inhom}(r,t)
&=
\frac{1}{\nu(C)\nu(D)}
\int_C\int_D 
\int_{\mathcal C_r^t(0,0)} g((0,0,m_1),(x,s,m_2))\de x \de s 
\nu(dm_2)\nu(dm_1)
\\
&=
\int_{\mathcal C_r^t(0,0)} g_g((0,0),(x,s))\de x \de s
.
\end{align*}
\end{proof}

\begin{proof}[Proof of Lemma \ref{LemmaIndependentComponents}]

Under the assumption of independence between $Y|_C$ and $Y|_D$, 
\begin{align*}
&\rho^{(2)}((x_1,t_1,m_1),(x_2,t_2,m_2)) 
=\\
&= 
\left\{
\begin{array}{ll}
\rho^{(2)}((x_1,t_1,m_1),(x_2,t_2,m_2)) & \text{if } (m_1,m_2)\in C\times C \text{ or } (m_1,m_2)\in D\times D,
\\
\lambda(x_1,t_1,m_1)\lambda(x_2,t_2,m_2) & \text{if }
(m_1,m_2)\in C\times D \text{ or } (m_2,m_1)\in C\times D.
\end{array}
\right.
\end{align*}
Hence, 
in the former case, 
\begin{align*}
K_{\mathrm{inhom}}^{CD}(r,t) & =\dfrac{1}{\nu(C)\nu(D)}\int_C\int_D\int_{\mathcal C_r^t(0,0)} \de u\de v \nu(dm_1)\nu(dm_2)
= \int_{\mathcal C_r^t(0,0)} \de u \de v 
= 2\omega_dr^dt
\end{align*}
and in the latter case,
%
%
%
%
%
\beann
K_{\mathrm{inhom}}^{C\M}(r,t) 
&=& \frac{1}{\nu(C)\nu(\M)} 
\int_C\int_{\M\setminus C} \int_{\mathcal C_r^t(0,0)} 
\de u\de v \nu(dm_1)\nu(dm_2) + \\
&&+ \frac{1}{\nu(C)\nu(\M)} \int_C\int_C \int_{\mathcal C_r^t(0,0)} g((0,0,m_1),(u,v,m_2)) \de u\de v \nu(dm_1)\nu(dm_2)\\
&=& \dfrac{\nu(C)\nu(\M\setminus C)}{\nu(C)\nu(\M)} \ell(\mathcal{C}_r^t(0,0)) + \\
&&+
\dfrac{\nu(C)}{\nu(\M)}  \dfrac{1}{\nu(C)\nu(C)} \int_C\int_C \int_{\mathcal C_r^t(0,0)} g((0,0,m_1),(u,v,m_2)) \de u\de v \nu(dm_1)\nu(dm_2)  \\
&=& 
\dfrac{\nu(\M\setminus C)}{\nu(\M)} 2\omega_dr^dt  + \dfrac{\nu(C)}{\nu(\M)}  K_{\mathrm{inhom}}^{CC}(r,t).
\eeann

\end{proof}

\end{document}